\documentclass[journal,comsoc, 12pt,onecolumn,draftclsnofoot]{IEEEtran}

\usepackage{tikz}
\usepackage[T1]{fontenc}% optional T1 font encoding
\usetikzlibrary{patterns, arrows,backgrounds,fit,tikzmark,positioning,calc,decorations,snakes,shapes,matrix}
\usepackage{pifont}

\usepackage{amsfonts}
\usepackage{amsmath,amssymb}
\usepackage{mathtools,hyperref}
\usepackage{amsthm}

\usepackage{subcaption} % Used for subfigure
\usepackage{nicefrac} % Used for nicefrac
\usepackage{array}
\usepackage{pgfplots}
\usepackage[subnum]{cases} % numbering equations of each case
\usepackage[inline]{enumitem}
\usepackage{lscape}
\usepackage{cleveref}
\usepackage{multicol}
\usepackage{multirow}
\usepackage{colortbl}
\usepackage{adjustbox}
\usepackage{arydshln}
\usepackage{tabstackengine}

\theoremstyle{remark}
\newtheorem{theorem}{ {Theorem}}
\newtheorem{corollary}{ {Corollary}}
\newtheorem{proposition}{{Proposition}}
\newtheorem{definition}{{Definition}}

\newtheorem{lemma}{ {Lemma}}
\newtheorem{remark}{ {Remark}}
\definecolor{RedBlue}{rgb}{0.8,0,0.5}
\definecolor{RedBlueGreen}{rgb}{0.8,0.6,0.5}
\definecolor{YellowOrange}{rgb}{0.4,0.4,0}
\definecolor{OliveGreen}{rgb}{0,0.6,0}

\DeclarePairedDelimiter\ceil{\lceil}{\rceil}

\DeclareMathOperator{\rank}{rank}
\newcommand{\greatleq}{%
	\mathrel{\ooalign{\raisebox{.6ex}{$>$}\cr\raisebox{-.6ex}{$\leq$}}}
}

\usetikzlibrary{decorations.text}
\usepgfplotslibrary{fillbetween}
\usetikzlibrary{calc, fit}
\usetikzlibrary{arrows.meta}
\usetikzlibrary{patterns}
\usetikzlibrary{fit,matrix}

\newcommand{\basestation}[3]{
	\coordinate (a) at (#1,#2);
	\draw[line width=(#3)*1.5pt,scale=#3] ($(a)+(0, -1.08)$) -- ($(a)+(0, 1)$);
	\draw[line width=(#3)*1.5pt,scale=#3] ($(a)+(0,1)$) .. controls ($(a)+(-0.25,-0.5)$) .. ($(a)+(-0.5,-0.9)$);
	\draw[line width=(#3)*1.5pt,scale=#3] ($(a)+(0,1)$) .. controls ($(a)+(0.25,-0.5)$) .. ($(a)+(0.5,-0.9)$);
	\draw[line width=(#3)*1.5pt,scale=#3] ($(a)+(0, -0.9)$) -- ($(a)+(-0.35, -0.7)$);
	\draw[line width=(#3)*1.5pt,scale=#3] ($(a)+(0, -0.9)$) -- ($(a)+(0.35, -0.7)$);
	\draw[line width=(#3)*0.75pt,scale=#3] ($(a)+(-0.35, -0.65)$) -- ($(a)+(0, -0.5)$);
	\draw[line width=(#3)*0.75pt,scale=#3] ($(a)+(0.35, -0.65)$) -- ($(a)+(0, -0.5)$);
	\draw[line width=(#3)*1pt,scale=#3] ($(a)+(0, -0.6)$) -- ($(a)+(-0.3, -0.45)$);
	\draw[line width=(#3)*1pt,scale=#3] ($(a)+(0, -0.6)$) -- ($(a)+(0.3, -0.45)$);
	\draw[line width=(#3)*0.5pt,scale=#3] ($(a)+(-0.3, -0.45)$) -- ($(a)+(0, -0.32)$);
	\draw[line width=(#3)*0.5pt,scale=#3] ($(a)+(0.3, -0.45)$) -- ($(a)+(0, -0.32)$);
	\draw[line width=(#3)*0.75pt,scale=#3] ($(a)+(0, -0.3)$) -- ($(a)+(-0.22, -0.17)$);
	\draw[line width=(#3)*0.75pt,scale=#3] ($(a)+(0, -0.3)$) -- ($(a)+(0.22, -0.17)$);
	\draw[line width=(#3)*0.5pt,scale=#3] ($(a)+(-0.22, -0.17)$) -- ($(a)+(0, -0.07)$);
	\draw[line width=(#3)*0.5pt,scale=#3] ($(a)+(0.22, -0.17)$) -- ($(a)+(0, -0.07)$);;
	\draw[line width=(#3)*0.75pt,scale=#3] (a) -- ($(a)+(-0.18, 0.11)$);
	\draw[line width=(#3)*0.75pt,scale=#3] (a) -- ($(a)+(0.18, 0.11)$);
	\draw[line width=(#3)*0.5pt,scale=#3] ($(a)+(-0.18, 0.11)$) -- ($(a)+(0,0.2)$);
	\draw[line width=(#3)*0.5pt,scale=#3] ($(a)+(0.18, 0.11)$) -- ($(a)+(0,0.2)$);   
	\draw[line width=(#3)*0.5pt,scale=#3] ($(a)+(0, 0.3)$) -- ($(a)+(-0.1, 0.37)$);
	\draw[line width=(#3)*0.5pt,scale=#3] ($(a)+(0, 0.3)$) -- ($(a)+(0.1, 0.37)$);
	\draw[line width=(#3)*0.25pt,scale=#3] ($(a)+(-0.1, 0.37)$) -- ($(a)+(0, 0.43)$);
	\draw[line width=(#3)*0.25pt,scale=#3] ($(a)+(0.1, 0.37)$) -- ($(a)+(0, 0.43)$);
	\draw[line width=(#3)*0.75pt,scale=#3] ($(a)+(0, 1.2)$) -- ($(a)+(0,1)$);
	\draw[fill=white,scale=#3] ($(a)+(0, 1.2)$) circle (0.05cm);
	\draw[line width=(#3)*1pt,decorate,decoration=expanding waves,decoration={segment length=(#3)*10pt},scale=#3] ($(a)+(0, 1.2)$) -- ($(a)+(0, 2.5)$);
}

\newcommand{\iPhone}[3]{
	\coordinate (a) at (#1,#2);
	\draw [line width=0.25pt,rounded corners=(#3)*1mm,fill=white,scale=(#3)] (a)--($(a)+(0.67,0)$)--($(a)+(0.67,1.381)$)--($(a)+(0,1.381)$)--cycle;
	\draw [color=gray,line width=0.25pt,rounded corners=(#3)*0.8mm,fill=white,scale=(#3)]($(a)+(0.015,0.015)$)--($(a)+(0.655,0.015)$)--($(a)+(0.655,1.366)$)--($(a)+(0.015,1.366)$)--cycle;
	\draw [line width=0.25pt,rounded corners=(#3)*0.04mm,scale=(#3)] ($(a)+(0.2875,1.266)$)--($(a)+(0.3825,1.266)$)--($(a)+(0.3825,1.281)$)--($(a)+(0.2875,1.281)$)--cycle;
	\draw[line width=0.25pt,scale=#3] ($(a)+(0.335,0.09)$) circle (0.055cm);
	\draw[line width=0.25pt,scale=#3] ($(a)+(0.335,0.09)$) circle (0.044cm);
	\draw[line width=0.25pt,scale=#3] ($(a)+(0.2275,1.2735)$) circle (0.015cm);
	\draw[line width=0.25pt,scale=#3] ($(a)+(0.335,1.32)$) circle (0.01cm);
	\draw [fill={rgb:black,1;white,4},line width=0.25pt,scale=(#3)] ($(a)+(0.042475,0.170195)$)--($(a)+(0.042475,0.170195)+(0.58505,0.0)$)--($(a)+(0.042475,0.170195)+(0.58505,1.04061)$)--($(a)+(0.042475,0.170195)+(0.0,1.04061)$)--cycle;
}

\newcommand{\SymMod}[0]{
	\begin{tikzpicture}[
	roundnode/.style={circle, draw=green!60, fill=green!5, very thick, minimum size=7mm},
	squarednode/.style={rectangle, draw=red!60, fill=red!5, very thick, minimum size=5mm},
	scale=1.25]
	\basestation{0.15}{2.6+2}{0.65};
	\node[scale=0.75] at (0.15,2.6+1.05) {DeNB};

	\basestation{2+0.75}{0.6+1.9}{0.425};
	\node (ch) at (2+1.25, 0.4+1.9) {\includegraphics[scale=1]{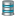}};
	\node[scale=0.75] at (2+1.0,-0.15+1.9) {RN$_1$};
	\iPhone{3.75+2.7}{2.6+1.7}{0.5};
	\node[scale=0.75] at (3.75+2.9,2.2+1.8) {UE$_1$};
	
	\draw[->, black, thick] (0.7, 2.6+2) -- (3.3+3-0.05, 2.6+2) node[pos=0.5,sloped,above] {$g_{1}$};
	\draw[->, thick, blue] (0.7, 2.6+2) -- (1.5+1-0.05, 0.6+2) node[pos=0.5,sloped,above] {$f_1$};
	\draw[->, magenta, thick] (2.45+1+0.1, 0.6+2) -- (3.25+3-0.05, 2.5+2) node[pos=0.5,sloped,below] {$h_{11}$};

	\basestation{2+0.75}{0}{0.425};
	\node (ch) at (2+1.25, -0.2) {\includegraphics[scale=1]{content/database-icon16px.png}};
	\node[scale=0.75] at (2+1.0,-0.75) {RN$_M$};
	\iPhone{3.75+2.7}{-0.3}{0.5};
	\node[scale=0.75] at (3.75+2.9,-0.6) {UE$_K$};

	\draw[->, black, thick] (0.7, 2.6+2) to [out=-10, in=-230] (3.25+3-0.05, 0);
	\node at (3.45, 3.65) {$g_{K}$};
	\draw[->, thick, blue] (0.7, 2.6+2) -- (1.4+1-0.05, 0) node[pos=0.5,sloped,above] {$f_M$};
	\draw[->, magenta, thick] (2.45+1+0.1, 0.6+2) -- (3.25+3-0.05, 0) node[pos=0.5,sloped,below] {$h_{K1}$};
	\draw[->, magenta, thick] (2.5+1+0.1, 0) -- (3.25+3-0.05, 2.5+2) node[pos=0.5,sloped,below] {$h_{1M}$};
	\draw[->, magenta, thick] (2.5+1+0.1, 0) -- (3.25+3-0.05, 0) node[pos=0.5,sloped,below] {$h_{KM}$};
	
	\path (3.75+2.9, 2.6+2) -- (3.75+2.9, 0) node [black, scale=3.0, midway, sloped] {$\dots$};
	\path (2+1, 0.24+1.8) -- (2+1, 0.45) node [black, scale=1.25, midway, sloped] {$\dots$};
	\end{tikzpicture}
}

\newcommand{\SymModSchemePhaseOne}[0]{
	\begin{tikzpicture}[
	roundnode/.style={circle, draw=green!60, fill=green!5, very thick, minimum size=7mm},
	squarednode/.style={rectangle, draw=red!60, fill=red!5, very thick, minimum size=5mm},scale=0.95
	]
	
	\basestation{-0.1}{2.6+2}{0.65};
	\node[scale=0.75] at (-0.1,2.6+1.05) {DeNB};

	\basestation{2+0.75}{0.6+1.9}{0.425};
	\node (ch) at (2+1.25, 0.4+1.9) {\includegraphics[scale=1]{content/database-icon16px.png}};
	\node[scale=0.75] at (2+1.0,-0.15+1.9-0.05) {RN$_1$};
	\node[rectangle,draw,rounded corners=5pt,minimum width=1.25cm, minimum height=1.625cm, very thick, blue] at (2+1,0.9+1.9) {};
	
	\basestation{2+0.75}{-1.7+1.9}{0.425};
	\node (ch) at (2+1.25, -1.9+1.9) {\includegraphics[scale=1]{content/database-icon16px.png}};
	\node[scale=0.75] at (2+1.0,-2.45+1.9-0.05) {RN$_2$};
	\node[rectangle,draw,rounded corners=5pt,minimum width=1.25cm, minimum height=1.625cm, very thick, blue] at (2+1,-1.4+1.9) {};
		
	\basestation{2+0.75}{-4+1.9}{0.425};
	\node (ch) at (2+1.25, -4.2+1.9) {\includegraphics[scale=1]{content/database-icon16px.png}};
	\node[scale=0.75] at (2+1.0,-4.85+1.9-0.05) {RN$_3$};
	\node[rectangle,draw,rounded corners=5pt,minimum width=1.25cm, minimum height=1.625cm, very thick, blue] at (2+1,-3.7+1.9) {};
	\node[rectangle,draw,rounded corners=5pt,minimum width=1.45cm, minimum height=1.825cm, very thick, red] at (2+1,-3.7+1.9) {};
	
	\basestation{2+0.75}{-6.3+1.9}{0.425};
	\node (ch) at (2+1.25, -6.5+1.9) {\includegraphics[scale=1]{content/database-icon16px.png}};
	\node[scale=0.75] at (2+1.0,-7.05+1.9-0.05) {RN$_4$};
	\node[rectangle,draw,rounded corners=5pt,minimum width=1.25cm, minimum height=1.625cm, very thick, red] at (2+1,-6+1.9) {};
		
	\iPhone{3.75+2.7+0.25+1}{2.6+1.7-0.1}{0.5};
	\node[scale=0.75] at (3.75+2.9+0.25+1,2.2+1.8-0.1) {UE$_1$};
	\iPhone{3.75+2.7+0.25+1}{0.4+1.9}{0.5};
	\node[scale=0.75] at (3.75+2.9+0.25+1,0+1.9) {UE$_2$};
	
	\draw[dashed,blue, very thick] (1.3,-1.35+1.9) ellipse (0.5cm and 3cm);
	\node[scale=0.65, align=center, blue] at (1.3,-1.35+1.9) {SISO \\ [-1.5ex]Broad- \\[-1.5ex] cast:\\[-1.5ex] MAN\\[-1.5ex] scheme};
	
	\draw[dashed,blue, thick] (5.05,0.6+1.9) ellipse (0.5cm and 3cm);
	\node[scale=0.65, align=center, blue] at (5.05,0.6+1.9) {MISO \\ [-1.5ex]Broad- \\[-1.5ex] cast:\\[-1.5ex] ZF\\[-1.5ex] Inter-\\[-1.5ex]ference};
	
	\draw[dashed,red, very thick] (5.05,-4.95+1.9) ellipse (0.5cm and 2cm);
	\node[scale=0.65, align=center, red] at (5.05,-4.65+1.9) {MISO \\ [-1.5ex]Broad- \\[-1.5ex] cast:\\[-1.5ex] UE\\[-1.5ex] Desired\\[-1.5ex]Symbols};
	
	\draw[-, thick, blue] (0.35, 2.6+2) -- (0.95, 0.6+2);
	\draw[->, thick, blue] (1.675, 0.6+2) -- (2.35, 0.6+2);
	\draw[-, thick, blue] (0.35, 2.6+2) -- (0.815, -1.6+2);
	\draw[->, thick, blue] (1.825, -1.6+2) -- (2.35, -1.6+2);
	\draw[-, thick, blue] (0.35, 2.6+2) -- (0.7, -3+2) -- (0.9, -3+2);
	\draw[->, thick, blue] (1.725, -3+2) -- (2.2, -3.8+2);
	\draw[-, thick, gray!55] (0.35, 2.6+2) -- (0.50, -3.8+2) node[pos=0.5,sloped,below,scale=0.85] {Inactive} -- (0.95, -3.8+2);
	\draw[->, thick, gray!55] (1.65, -3.8+2) -- (2.3, -6+2);
	
	\draw[-, thick, blue] (3.65,0.6+2) -- (4.575, 0.6+2);
	\draw[->, thick, blue] (5.565,0.6+2) -- (6.575+1, 2.6+2);
	\draw[->, thick, blue] (5.565,0.6+2) -- (6.575+1, 0.6+2);
	\draw[-, thick, blue] (3.65,-1.6+2) -- (4.675, -1.6+2);
	\draw[->, thick, blue] (5.435,-1.6+2) -- (6.575+1, 2.6+2);
	\draw[->, thick, blue] (5.435,-1.6+2) -- (6.575+1, 0.6+2);
	\draw[-, thick, blue] (3.775,-3.8+2) -- (4.725, -2+2);
	\draw[-, thick, red] (3.775,-3.8+2) -- (4.625, -3.8+2);
	\draw[->, thick, red] (5.435,-3.8+2) -- (6.575+1, 2.6+2);
	\draw[->, thick, red] (5.435,-3.8+2) -- (6.575+1, 0.6+2);
	\draw[-, thick, red] (3.65,-6+2) -- (4.625, -6+2);
	\draw[->, thick, red] (5.535,-6+2) -- (6.575+1, 2.6+2);
	\draw[->, thick, red] (5.535,-6+2) -- (6.575+1, 0.6+2);
	
	\draw[-, thick, blue] (0.356, 2.6+2) -- (4.7, 2.6+2);
	\draw[->, thick, blue] (5.425,2.6+2) -- (6.575+1, 2.6+2);
	\draw[->, thick, blue] (5.425,2.6+2) -- (6.575+1, 0.6+2);	
		
	\end{tikzpicture}
}

\newcommand{\SymModSchemePhaseTwo}[0]{
	\begin{tikzpicture}[
	roundnode/.style={circle, draw=green!60, fill=green!5, very thick, minimum size=7mm},
	squarednode/.style={rectangle, draw=red!60, fill=red!5, very thick, minimum size=5mm},
	scale=0.95]
	
	\basestation{-0.1}{2.6+2}{0.65};
	\node[scale=0.75] at (-0.1,2.6+1.05) {DeNB};

	\basestation{2+0.75}{0.6+1.9}{0.425};
	\node (ch) at (2+1.25, 0.4+1.9) {\includegraphics[scale=1]{content/database-icon16px.png}};
	\node[scale=0.75] at (2+1.0,-0.15+1.9-0.05) {RN$_1$};
	\node[rectangle,draw,rounded corners=5pt,minimum width=1.25cm, minimum height=1.625cm, very thick, blue] at (2+1,0.9+1.9) {};
	
	\basestation{2+0.75}{-1.7+1.9}{0.425};
	\node (ch) at (2+1.25, -1.9+1.9) {\includegraphics[scale=1]{content/database-icon16px.png}};
	\node[scale=0.75] at (2+1.0,-2.45+1.9-0.05) {RN$_2$};
	\node[rectangle,draw,rounded corners=5pt,minimum width=1.25cm, minimum height=1.625cm, very thick, blue] at (2+1,-1.4+1.9) {};
	
	\basestation{2+0.75}{-4+1.9}{0.425};
	\node (ch) at (2+1.25, -4.2+1.9) {\includegraphics[scale=1]{content/database-icon16px.png}};
	\node[scale=0.75] at (2+1.0,-4.85+1.9-0.05) {RN$_3$};
	\node[rectangle,draw,rounded corners=5pt,minimum width=1.25cm, minimum height=1.625cm, very thick, gray!55] at (2+1,-3.7+1.9) {};
	
	\node[scale=0.9, gray!55, align=center] at (0.95,-5.05+1.9) {Inactive RNs \\[-1.5ex] (RN$_3$, RN$_4$)};
	
	\basestation{2+0.75}{-6.3+1.9}{0.425};
	\node (ch) at (2+1.25, -6.5+1.9) {\includegraphics[scale=1]{content/database-icon16px.png}};
	\node[scale=0.75] at (2+1.0,-7.05+1.9-0.05) {RN$_4$};
	\node[rectangle,draw,rounded corners=5pt,minimum width=1.25cm, minimum height=1.625cm, very thick, gray!55] at (2+1,-6+1.9) {};
	
	\iPhone{3.75+2.7+0.25+1}{2.6+1.7-0.1}{0.5};
	\node[scale=0.75] at (3.75+2.9+0.25+1,2.2+1.8-0.1) {UE$_1$};
	\iPhone{3.75+2.7+0.25+1}{0.4+1.9}{0.5};
	\node[scale=0.75] at (3.75+2.9+0.25+1,0+1.9) {UE$_2$};
		
	\draw[dashed,red, thick] (5.05,0.6+1.9) ellipse (0.5cm and 3cm);
	\node[scale=0.65, align=center, red] at (5.05,0.9+1.9) {MISO \\ [-1.5ex]Broad- \\[-1.5ex] cast:\\[-1.5ex] UE\\[-1.5ex] Desired\\[-1.5ex]Symbols};
		
	\draw[-, thick, red] (3.65,0.6+2) -- (4.575, 0.6+2);
	\draw[->, thick, red] (5.565,0.6+2) -- (6.575+1, 2.6+2);
	\draw[->, thick, red] (5.565,0.6+2) -- (6.575+1, 0.6+2);
	\draw[-, thick, red] (3.65,-1.6+2) -- (4.675, -1.6+2);
	\draw[->, thick, red] (5.435,-1.6+2) -- (6.575+1, 2.6+2);
	\draw[->, thick, red] (5.435,-1.6+2) -- (6.575+1, 0.6+2);
	
	\draw[-, thick, red] (0.356, 2.6+2) -- (4.7, 2.6+2);
	\draw[->, thick, red] (5.425,2.6+2) -- (6.575+1, 2.6+2);
	\draw[->, thick, red] (5.425,2.6+2) -- (6.575+1, 0.6+2);	
	
	\end{tikzpicture}
}

\newcommand{\SymModLB}[0]{
	\begin{tikzpicture}[
	roundnode/.style={circle, draw=green!60, fill=green!5, very thick, minimum size=7mm},
	squarednode/.style={rectangle, draw=red!60, fill=red!5, very thick, minimum size=5mm},
	]
	
	\basestation{0.15}{2.6+2}{0.65};
	\node[scale=0.75] at (0.15,2.6+1.05) {DeNB};

	\basestation{2+0.75}{0.6+1.9}{0.425};
	\node (ch) at (2+1.25, 0.4+1.9) {\includegraphics[scale=1]{content/database-icon16px.png}};
	\node[scale=0.75] at (2+1.0,-0.15+1.9) {RN$_1$};

	\basestation{2+0.75}{0.6+1.9-2}{0.425};
	\node (ch) at (2+1.25, 0.4+1.9-2) {\includegraphics[scale=1]{content/database-icon16px.png}};
	\node[scale=0.75] at (2+1.0,-0.15+1.9-2) {RN$_2$};
	\path (2+1.0,-0.15+1.9-2.25) -- (2+1.0,-0.15+1.9-3.25) node [black, scale=1.25, midway, sloped] {$\dots$};
	
	\basestation{2+0.75}{0.6+1.9-5}{0.425};
	\node (ch) at (2+1.25, 0.4+1.9-5) {\includegraphics[scale=1]{content/database-icon16px.png}};
	\node[scale=0.75] at (2+1.0,-0.15+1.9-5) {RN$_\ell$};
	\node[rectangle,draw,rounded corners=5pt,minimum width=1.25cm, minimum height=7.2cm, very thick, blue] at (2+1,-0.15+1.9-1.7) {};
	\path (2+1.0,-0.15+1.9-5.25) -- (2+1.0,-0.15+1.9-6.25) node [black, scale=1.25, midway, sloped] {$\dots$};
	
	\basestation{2+0.75}{0.6+1.9-8}{0.425};
	\node (ch) at (2+1.25, 0.4+1.9-8) {\includegraphics[scale=1]{content/database-icon16px.png}};
	\node[scale=0.75] at (2+1.0,-0.15+1.9-8) {RN$_M$};
	
	\iPhone{3.75+2.7}{2.6+1.7}{0.5};
	\node[scale=0.75] at (3.75+2.9,2.2+1.8) {UE$_1$};
	
	\iPhone{3.75+2.7}{2.6+1.7-1.35}{0.5};
	\node[scale=0.75] at (3.75+2.9,2.2+1.8-1.35) {UE$_2$};
	\path (3.75+2.9,2.2+1.8-1.6) -- (3.75+2.9,2.2+1.8-2.6) node [black, scale=1, midway, sloped] {$\dots$};
	
	\iPhone{3.75+2.7}{2.6+1.7-2.6-1}{0.5};
	\node[scale=0.75] at (3.75+2.9,2.2+1.8-2.6-1) {UE$_s$};
	\node[rectangle,draw,rounded corners=5pt,minimum width=0.9cm, minimum height=5.1cm, very thick, red] at (3.9+2.7,2.6+1.7-1.65) {};
	\path (3.75+2.9,2.2+1.8-2.6-1.25) -- (3.75+2.9,2.2+1.8-2.6-2.25) node [black, scale=1, midway, sloped] {$\dots$};
	
	\iPhone{3.75+2.7}{2.6+1.7-2.6-3.2}{0.5};
	\node[scale=0.75] at (3.75+2.9,2.2+1.8-2.6-3.2) {UE$_K$};
	
	\draw[-] (3.65,0.6+2) -- (3.75+2.7-0.32, 2.6+2);
	\draw[-] (3.65,0.6+2) -- (3.75+2.7-0.32, 2.6+2-1.35);
	\draw[-] (3.65,0.6+2) -- (3.75+2.7-0.32, 2.6+2-1.35-2.25);
	\draw[-] (3.65,0.6+2) -- (3.75+2.7-0.32, 2.6+2-1.35-4.5);
	\draw[-] (3.65,0.6+2-2) -- (3.75+2.7-0.32, 2.6+2);
	\draw[-] (3.65,0.6+2-2) -- (3.75+2.7-0.32, 2.6+2-1.35);
	\draw[-] (3.65,0.6+2-2) -- (3.75+2.7-0.32, 2.6+2-1.35-2.25);
	\draw[-] (3.65,0.6+2-2) -- (3.75+2.7-0.32, 2.6+2-1.35-4.5);
	\draw[-] (3.65,0.6+2-2-3) -- (3.75+2.7-0.32, 2.6+2);
	\draw[-] (3.65,0.6+2-2-3) -- (3.75+2.7-0.32, 2.6+2-1.35);
	\draw[-] (3.65,0.6+2-2-3) -- (3.75+2.7-0.32, 2.6+2-1.35-2.25);
	\draw[-] (3.65,0.6+2-2-3) -- (3.75+2.7-0.32, 2.6+2-1.35-4.5);
	\draw[-] (3.65,0.6+2-2-3-3) -- (3.75+2.7-0.32, 2.6+2);
	\draw[-] (3.65,0.6+2-2-3-3) -- (3.75+2.7-0.32, 2.6+2-1.35);
	\draw[-] (3.65,0.6+2-2-3-3) -- (3.75+2.7-0.32, 2.6+2-1.35-2.25);
	\draw[-] (3.65,0.6+2-2-3-3) -- (3.75+2.7-0.32, 2.6+2-1.35-4.5);
	
	\draw[-] (0.45,2.6+2) -- (3.75+2.7-0.32,2.6+2);
	\draw[-] (0.45,2.6+2) to [out=-2, in=-230] (3.75+2.7-0.32, 2.6+2-1.35);
	\draw[-] (0.45,2.6+2) to [out=-1, in=100] (3.75+2.7-0.32, 2.6+2-1.35-2.25);
	\draw[-] (0.65,2.6+2) to [out=-0.2, in=90] (3.75+2.7-0.32, 2.6+2-1.35-4.5);
	
	\draw[-] (0.45,2.6+2) -- (2.35,0.6+2);
	\draw[-] (0.45,2.6+2) -- (2.35,0.6+2-2);
	\draw[-] (0.45,2.6+2) -- (2.35,0.6+2-2-3);
	\draw[-] (0.45,2.6+2) -- (2.35,0.6+2-2-3-3);
	
	\node[rectangle,draw,rounded corners=5pt,minimum width=2.25cm, minimum height=1cm, very thick, blue, scale=0.75, align=center] at (2+1,8.5) {$\ell$ RNs caches:\\[-1.5ex] $S_{[1:\ell]}$};
	\node[rectangle,draw,rounded corners=5pt,minimum width=2.25cm, minimum height=1cm, very thick, scale=0.75, align=center] at (2+1-3,8.5) {$\ell$ RNs received \\[-1.5ex]signals: $\mathbf{y}_{r,[1:\ell]}[t]$};
	\node[rectangle,draw,rounded corners=5pt,minimum width=2.25cm, minimum height=1cm, very thick, scale=0.75, align=center] at (2+1,10.5) {$\ell$ RNs transmitted \\[-1.5ex]signals: $\mathbf{x}_{r,[1:\ell]}[t]$};
	\node[rectangle,draw,rounded corners=5pt,minimum width=2.25cm, minimum height=1cm, very thick, scale=0.75, align=center, red] at (3.75+2.9,10.5) {$s$ UEs received \\[-1.5ex]signals: $\mathbf{y}_{u,[1:s]}[t]$};
	\node[rectangle,draw,rounded corners=5pt,minimum width=2.25cm, minimum height=1cm, very thick, scale=0.75, align=center] at (3.75+2.9-1.75,12.5) {All transmitted\\[-1.5ex] signals};
	\node[rectangle,draw,minimum width=2.25cm, minimum height=1cm, thick, scale=0.75, align=center] at (2+1-1.4,6.85) {\textcolor{blue}{$\ell$ RN files}};
	\node[rectangle,draw,minimum width=2.25cm, minimum height=1cm, thick, scale=0.75, align=center] at (3.75+2.9-1.75+3,12.5) {\textcolor{red}{$K$ UE files}};
	
	\draw[->, blue] (2+1,3.65) -- (2+1,8.05);
	\draw[->, red] (3.75+2.9,5.2) -- (3.75+2.9,10.05);
	\draw[->, thick] (2+1,8.9) -- (2+1,10.05) node[pos=0.7,sloped,below,rotate=-90,scale=0.7] {$\qquad\qquad\qquad t=1,\ldots,T$};
	\draw[->, thick] (2+1-3,8.9) -- (2+1-3,9.475) -- (2+1-0.75,9.475) node[pos=0.4,sloped,above,scale=0.7] {$t=1,\ldots,T-1$} -- (2+1-0.75,10.05);
	\draw[->, thick] (2+1,10.9) -- (2+1,11.475) -- (2+1+1.5,11.475) node[pos=0.3,sloped,above,scale=0.7] {$t=1,\ldots,T$} -- (2+1+1.5,12.05);
	\draw[->, thick] (3.75+2.9,10.9) -- (3.75+2.9,11.475) -- (3.75+2.9-1.4,11.475) node[pos=0.3,sloped,above,scale=0.7] {$t=1,\ldots,T$} -- (3.75+2.9-1.4,12.05);
	\draw[->, thick] (3.8,12.5) -- (3.8-5.75,12.5) node[pos=0.5,sloped,above,scale=0.7] {DeNB signal: $x_s[t],t=1,\ldots,T$} -- (3.8-5.75,8.5) -- (3.8-5,8.5);
	\draw[->, very thick, blue] (2+1-0.75,8.05) -- (2+1-0.75,7.25);
	\draw[->, very thick, blue] (2+1-2,8.05) -- (2+1-2,7.25);
	\draw[->, very thick, red] (3.75+2.9-1.75+3-1.9,12.5) -- (3.75+2.9-1.75+2.15,12.5);
	\end{tikzpicture}
}

\newcommand{\PlotNDTMOneKOne}[0]{
		\begin{axis}[
			axis x line=bottom, axis y line=left,
			ymin=0, ymax=3, ytick={1,...,2}, ylabel={\small Normalized Delivery Time},
			xmin=0, xmax=1.2, xtick={1,...,1.2}, xlabel=$\mu$, grid=major
			]
			
			\addplot+[name path=f1, no markers, decoration={
				text={}, raise=1ex,
				text align={center}
			}, postaction={decorate}, very thick, blue, domain=0:1, smooth]{2-x};
			\addplot[very thick, blue, domain=1:1.2, smooth]{1};
			
			\path[name path=axis] (axis cs:0,3) -- (axis cs:0,2) -- (axis cs:1,1) -- (axis cs:1.2,1) -- (axis cs:1.2,3) -- (axis cs:0,3);
			\addplot[blue!10] fill between[of=axis and f1];
			\addplot[red, mark=otimes*, mark size=3.75pt, only marks] coordinates {(0,2) (1,1)};
		\end{axis}
}

\newcommand{\PlotNDTMOneKTwo}[0]{
	\begin{axis}[
	axis x line=bottom, axis y line=left,
	ymin=0, ymax=4, ytick={1,3}, ylabel={\small Normalized Delivery Time},
	xmin=0, xmax=1.2, xtick={1}, xlabel=$\mu$, grid=major,
	]
	
	\addplot+[name path=f1, no markers, decoration={
		text={}, raise=1ex,
		text align={center}
	}, postaction={decorate}, very thick, blue, domain=0:1, smooth]{3-x-x};
	\addplot[very thick, blue, domain=1:1.2, smooth]{1};
	
	\path[name path=axis] (axis cs:0,4) -- (axis cs:0,3) -- (axis cs:1,1) -- (axis cs:1.2,1) -- (axis cs:1.2,4) -- (axis cs:0,4);
	\addplot[blue!10] fill between[of=axis and f1];
	\addplot[red, mark=otimes*, mark size=3.75pt, only marks] coordinates {(0,3) (1,1)};
	\end{axis}
}

\newcommand{\PlotNDTMOneKGreatTwo}[0]{
	\begin{axis}[
		axis x line=bottom, 
		axis y line=left,
		ymin=0, 
		ymax=5,
		ytick={1.5,1.6,2.5, 4}, 
		ylabel={\small Normalized Delivery Time},
		xmin=0, 
		xmax=1.2, 
		xtick={0.5, 0.8,1}, 
		xlabel=$\mu$, 
		xticklabels={$\nicefrac{1}{2}$,$\mu_1=\frac{K+1}{2K-1}$,$1$},
		yticklabels={$\ $,$\ $ ,$\frac{K+2}{2}$, \scriptsize{$K+1$}}, 
		grid=major, 
		]
		
		\addplot+[name path=f1, no markers, decoration={
			text along path,
			text={}, raise=1ex,
			text align={center}
		}, postaction={decorate}, very thick, blue, domain=0:0.8, smooth]{4-x-x-x};
		\addplot+[no markers, decoration={
			text={}, raise=1ex,
			text align={center}
		}, postaction={decorate}, very thick, blue, domain=0.8:1, smooth] plot coordinates {
		(0.8,1.6)
		(1,1.5)
	};
	
	\addplot+[no markers, decoration={
		text={}, raise=1ex,
		text align={center}
	}, postaction={decorate}, red, dashed, domain=0:1, smooth, very thick] plot coordinates {
	(0,4)
	(1,1.5)
};
\addplot[blue, very thick, domain=1:1.2, smooth]{1.5};

\path[name path=axis1] (axis cs:0,5) -- (axis cs:0,4) -- (axis cs:0.8,1.6) -- (axis cs:1,1.5) -- (axis cs:1.2,1.5) -- (axis cs:1.2,5) -- (axis cs:0,5);
\addplot[blue!10] fill between[of=axis1 and f1];
\addplot[red, mark=otimes*, mark size=3pt, only marks] coordinates {(0,4) (0.8, 1.6) (1,1.5)};
\addplot[blue, mark=square*, mark size=3pt, only marks] coordinates {(0.5, 2.5)};
\end{axis}
\node at (-.3,1.4) {$\frac{K}{2}$};
\node at (-.5,2.0) {$\frac{K^2-1}{2K-1}$};
}

\newcommand{\PlotNDTMOneKThree}[0]{
	\begin{axis}[
		axis x line=bottom, 
		axis y line=left,
		ymin=0, 
		ymax=5, 
		ytick={1.5,1.6, %2.5, 
			4}, 
		ylabel={\small Normalized Delivery Time},
		xmin=0, 
		xmax=1.2, 
		xtick={ %0.5, 
			0.8,1}, 
		xlabel=$\mu$, 
		xticklabels={%$\nicefrac{1}{2}$,
			$\nicefrac{4}{5}$,$1$},
		yticklabels={$\ $,$\ $,
			$4$}, %\scriptsize{$4$}}, 
		grid=major, 
		]
		
		\addplot+[name path=f1, no markers, decoration={
			text along path,
			text={}, raise=1ex,
			text align={center}
		}, postaction={decorate}, very thick, blue, domain=0:0.8, smooth]{4-x-x-x};
		\addplot+[no markers, decoration={
			text={}, raise=1ex,
			text align={center}
		}, postaction={decorate}, very thick, blue, domain=0.8:1, smooth] plot coordinates {
		(0.8,1.6)
		(1,1.5)
	};
	
	\addplot+[no markers, decoration={
		text={}, raise=1ex,
		text align={center}
	}, postaction={decorate}, red, dashed, very thick,domain=0:1, smooth] plot coordinates {
	(0,4)
	(1,1.5)
};
\addplot[blue, very thick, domain=1:1.2, smooth]{1.5};

\path[name path=axis1] (axis cs:0,5) -- (axis cs:0,4) -- (axis cs:0.8,1.6) -- (axis cs:1,1.5) -- (axis cs:1.2,1.5) -- (axis cs:1.2,5) -- (axis cs:0,5);
\addplot[blue!10] fill between[of=axis1 and f1];
\addplot[red, mark=otimes*, mark size=3pt, only marks] coordinates {(0,4) (0.8, 1.6) (1,1.5)};

\end{axis}
\node at (-.3,1.4) {$\frac{3}{2}$};
\node at (-.3,2.0) {$\frac{8}{5}$};
}

\newcommand{\PlotNDTMTwoKOne}[0]{
	\begin{axis}[
		axis x line=bottom, axis y line=left,
		ymin=0, ymax=4, ytick={1,3}, ylabel={\small Normalized Delivery Time},
		xmin=0, xmax=1.2, xtick={0.5,1}, xlabel=$\mu$, grid=major, 	xticklabels={$\nicefrac{1}{2}$,$1$}
		]
		
		\addplot+[name path=f1, no markers, decoration={
			text={}, raise=1ex,
			text align={center}
		}, postaction={decorate}, very thick, blue, domain=0:0.5, smooth]{3-4*x};
		\addplot[very thick, blue, domain=0.5:1.2, smooth]{1};
		
		\path[name path=axis] (axis cs:0,4) -- (axis cs:0,3) -- (axis cs:0.5,1) -- (axis cs:1.2,1) -- (axis cs:1.2,4) -- (axis cs:0,4);
		\addplot[blue!10] fill between[of=axis and f1];
		\addplot[red, mark=otimes*, mark size=3pt, only marks] coordinates {(0,3) (0.5,1) (1,1)};
		
		\addplot+[no markers, decoration={
				text={}, raise=1ex,
				text align={center}
			}, postaction={decorate}, red, dashed, very thick, domain=0:1, smooth] plot coordinates {
			(0,3)
			(1,1)
		};
		
	\end{axis}
}

\newcommand{\PlotNDTMTwoKTwo}[0]{
	\begin{axis}[
		axis x line=bottom, axis y line=left,
		ymin=0, ymax=5, ytick={1,4}, ylabel={\small Normalized Delivery Time},
		xmin=0, xmax=1.2, xtick={0.4444444444444, 0.5, 1}, xlabel=$\mu$, grid=major, 	% xticklabels={$\frac{4}{9}$,$\frac{1}{2}$, $1$},
		xticklabels={$\ $,$\ $, $1$}, 
		yticklabels={$\ $, $4$},
		extra y ticks={1.33333333333333, 1.25}, %extra y tick labels={$\quad\frac{4}{3}$,$\frac{5}{4}$}, 
		extra y tick labels={$\ $,$\ $}, 
		extra y tick style={y tick label style={rotate=90}},
		]
		
		\addplot+[name path=f1, no markers, decoration={
			text={}, raise=1ex,
			text align={center}
		}, postaction={decorate}, very thick, blue, domain=0:4/9, smooth]{4-6*x};
		
		\addplot+[name path=f2, no markers, decoration={
			text={}, raise=1ex,
			text align={center}
			}, postaction={decorate}, very thick, blue, domain=4/9:0.5, smooth]{2-1.5*x};
		
		\addplot+[name path=f3, no markers, decoration={
			text={}, raise=1ex,
			text align={center}
			}, postaction={decorate}, very thick, blue, domain=0.5:1, smooth]{1.5-0.5*x};
			
		\addplot[very thick, blue, domain=1:1.2, smooth]{1};
		
		\path[name path=axis] (axis cs:0,5) -- (axis cs:0,4) -- (axis cs:4/9,4/3) -- (axis cs:0.5,5/4) -- (axis cs:1,1) -- (axis cs:1.2,1) -- (axis cs:1.2,5) -- (axis cs:0,5);
		\addplot[blue!10] fill between[of=axis and f1];
		\addplot[red, mark=otimes*, mark size=3pt, only marks] coordinates {(0,4) (4/9,4/3) (0.5,5/4) (1,1)};
		
		\addplot+[no markers, decoration={
					text={}, raise=1ex,
					text align={center}
				}, postaction={decorate}, red, very thick, dashed, domain=0:1, smooth] plot coordinates {
				(0,4)
				(1,1)
		};
		
	\end{axis}
\node[scale=0.8] at (2.55,-0.5) {$\frac{4}{9}$};	
\node[scale=0.8] at (2.90,-0.5) {$\frac{1}{2}$};
\node[scale=0.8] at (-.3,0.9) {$1$};
\node[scale=0.8] at (-.3,1.4) {$\frac{5}{4}$};
\node[scale=0.8] at (-.3,2.0) {$\frac{12}{9}$};
}

\newcommand{\PlotDoFMOneKThree}[0]{
	\begin{axis}[
		axis x line=bottom, 
		axis y line=left,
		ymin=0, 
		ymax=2.2, 
		ytick={1,2}, 
		ylabel={\small Sum Degrees-of-Freedom},
		xmin=0, 
		xmax=1.2, 
		xtick={
			0.8,1}, 
		xlabel=$\mu$, 
		xticklabels={
			$\nicefrac{4}{5}$,$1$},
		yticklabels={$1$, $2$},
		grid=major, 
		]
		
		\addplot+[name path=f3, no markers, decoration={
			text along path,
			text={}, raise=1ex,
			text align={center}
		}, postaction={decorate}, very thick, blue, domain=0:0.8, smooth]{(4-x)/(4-3*x)};
		\addplot+[name path=f4,no markers, decoration={
			text={}, raise=1ex,
			text align={center}
		}, postaction={decorate}, very thick, blue, domain=0.8:1, smooth] plot coordinates {
		(0.8,2)
		(1.2,2)
	};
	
	\addplot+[name path=f4, no markers, decoration={
				text along path,
				text={}, raise=1ex,
				text align={center}
			}, postaction={decorate}, dashed, red, very thick, domain=0:1, smooth]{(4-x)/(4-2.5*x)};
\addplot[red, mark=otimes*, mark size=3pt, only marks] coordinates {(0,1) (0.8, 2) (1,2)};

\end{axis}
}

\newcommand{\PlotDoFMTwoKOne}[0]{
	\begin{axis}[
		axis x line=bottom, axis y line=left,
		ymin=0, ymax=2.2, ytick={1,2}, ylabel={\small Sum Degrees-of-Freedom},
		xmin=0, xmax=1.2, xtick={0.5,1}, xlabel=$\mu$, grid=major, 	xticklabels={$\nicefrac{1}{2}$,$1$}
		]
		
		\addplot+[name path=f1, no markers, decoration={
			text={}, raise=1ex,
			text align={center}
		}, postaction={decorate}, very thick, blue, domain=0:0.5, smooth]{(3-2*x)/(3-4*x)};

		\addplot+[name path=f2, no markers, decoration={
					text={}, raise=1ex,
					text align={center}
				}, postaction={decorate}, very thick, blue, domain=0.5:1, smooth]{(3-2*x)};
		\addplot+[name path=f3,no markers, decoration={
				text={}, raise=1ex,
				text align={center}
			}, postaction={decorate}, very thick, blue, domain=0.8:1, smooth] plot coordinates {
			(1,1)
			(1.2,1)
		};
		
		\addplot+[name path=f4,no markers, decoration={
					text={}, raise=1ex,
					text align={center}
				}, postaction={decorate}, dashed, very thick, red, domain=0:1, smooth] plot coordinates {
				(0,1)
				(1,1)
		};
		
		\addplot[red, mark=otimes*, mark size=3pt, only marks] coordinates {(0,1) (0.5, 2) (1,1)};
	
\end{axis}
}

\newcommand{\PlotDoFMTwoKTwo}[0]{
	\begin{axis}[
		axis x line=bottom, axis y line=left,
		ymin=0, ymax=2.6, ytick={1,2}, ylabel={\small Sum Degrees-of-Freedom},
		xmin=0, xmax=1.2, xtick={0.4444444444444, 0.5, 1}, xlabel=$\mu$, grid=major, 	% xticklabels={$\frac{4}{9}$,$\frac{1}{2}$, $1$},
		xticklabels={$\ $,$\ $, $1$}, 
		yticklabels={$1$, $2$},
		extra y ticks={2.33333333333333, 2.4}, 
		extra y tick labels={$\ $,$\ $}, 
		extra y tick style={y tick label style={rotate=90}},
		]
		
		\addplot+[name path=f1, no markers, decoration={
			text={}, raise=1ex,
			text align={center}
		}, postaction={decorate}, very thick, blue, domain=0:4/9, smooth]{(4-2*x)/(4-6*x)};
		
		\addplot+[name path=f2, no markers, decoration={
			text={}, raise=1ex,
			text align={center}
		}, postaction={decorate}, very thick, blue, domain=4/9:0.5, smooth]{(8-4*x)/(4-3*x)};
		
		\addplot+[name path=f3, no markers, decoration={
			text={}, raise=1ex,
			text align={center}
		}, postaction={decorate}, very thick, blue, domain=0.5:1, smooth]{(8-4*x)/(3-x)};
		
		\addplot[very thick, blue, very thick, domain=1:1.2, smooth]{2};
		
		\addplot[red, mark=otimes*, mark size=3pt, only marks] coordinates {(0,1) (4/9,7/3) (0.5,2.4) (1,2)};
	
		\addplot+[name path=f4, no markers, decoration={
				text along path,
				text={}, raise=1ex,
				text align={center}
			}, postaction={decorate}, dashed, very thick, red, domain=0:1, smooth]{(4-2*x)/(4-3*x)};
	
\end{axis}
\node[scale=0.8] at (2.55,-0.5) {$\frac{4}{9}$};	
\node[scale=0.8] at (2.90,-0.5) {$\frac{1}{2}$};
\node[scale=0.8] at (-.3,4.8) {$\frac{7}{3}$};
\node[scale=0.8] at (-.3,5.4) {$\frac{12}{5}$};
}

\newcommand{\FilesKThree}[0]{
	\centering
\begin{tabular}{c|ccccc}
\multicolumn{1}{c}{} & \multicolumn{4}{c}{\textcolor{red}{RN \& DeNB}} & \textcolor{blue}{DeNB only}\\   
File $W_1$ & \tikz[remember picture]{\node[anchor=base,inner sep=0pt] (a1) {\textcolor{OliveGreen}{$\eta_{1,1}$}};} & \textcolor{OliveGreen}{$\eta_{1,2}$} & \textcolor{OliveGreen}{$\eta_{1,3}$} & $\eta_{1,4}$ & \tikz[remember picture]{\node[anchor=base,inner sep=0pt] (a5) {$\eta_{1,5}$};} \\
File $W_2$ & \textcolor{OliveGreen}{$\eta_{2,1}$} & \textcolor{OliveGreen}{$\eta_{2,2}$} & \textcolor{OliveGreen}{$\eta_{2,3}$} & $\eta_{2,4}$ & $\eta_{2,5}$ \\
File $W_3$ & \textcolor{OliveGreen}{$\eta_{3,1}$} & \textcolor{OliveGreen}{$\eta_{3,2}$} & \textcolor{OliveGreen}{$\eta_{3,3}$} & $\eta_{3,4}$ & $\eta_{3,5}$ \\
File $W_4$ & $\eta_{4,1}$ & $\eta_{4,2}$ & $\eta_{4,3}$ & \tikz[remember picture]{\node[anchor=base,inner sep=-3pt] (d4) {$\eta_{4,4}$};} & \tikz[remember picture]{\node[anchor=base,inner sep=-3pt](d5) {$\eta_{4,5}$};}\\ [0.5\baselineskip] \multicolumn{1}{c}{} & \multicolumn{3}{c}{\textcolor{OliveGreen}{ZF symbols}}
\tikz[remember picture,overlay]{
\node [draw = red, fit = (a1) (d4)] {};
}
\tikz[remember picture,overlay]{
	\node [draw = blue, fit = (a5) (d5)] {};
}
\end{tabular}
}

\newcommand{\FilesMTwoKThree}[0]{
	\centering
	\begin{tabular}{c|ccccccccc}
		\multicolumn{1}{c}{} & \multicolumn{4}{c}{\textcolor{red}{RN$_1$ \& DeNB}} & \multicolumn{4}{c}{\textcolor{cyan}{RN$_2$ \& DeNB}} & \textcolor{blue}{DeNB only}\\   
		File $W_1$ & \tikz[remember picture]{\node[anchor=base,inner sep=0pt] (a1) {\textcolor{OliveGreen}{$\eta_{1,1}$}};} & \textcolor{OliveGreen}{$\eta_{1,2}$} & $\eta_{1,3}$ & $\eta_{1,4}$ & \tikz[remember picture]{\node[anchor=base,inner sep=0pt] (a5) {\textcolor{OliveGreen}{$\eta_{1,5}$}};} & \textcolor{OliveGreen}{$\eta_{1,6}$} & $\eta_{1,7}$ & 
		$\eta_{1,8}$ &
		 \tikz[remember picture]{\node[anchor=base,inner sep=0pt] (a9) {$\eta_{1,9}$};} \\
		File $W_2$ & \textcolor{OliveGreen}{$\eta_{2,1}$} & \textcolor{OliveGreen}{$\eta_{2,2}$} & $\eta_{2,3}$ & $\eta_{2,4}$ & \textcolor{OliveGreen}{$\eta_{2,5}$} & \textcolor{OliveGreen}{$\eta_{2,6}$} & $\eta_{2,7}$ & 
		$\eta_{2,8}$ & $\eta_{2,9}$ \\
		File $W_3$ & $\eta_{3,1}$ & $\eta_{3,2}$ & $\eta_{3,3}$ & $\eta_{3,4}$ & \textcolor{OliveGreen}{$\eta_{3,5}$} & \textcolor{OliveGreen}{$\eta_{3,6}$} & \textcolor{OliveGreen}{$\eta_{3,7}$} & 
		\textcolor{OliveGreen}{$\eta_{3,8}$} & $\eta_{3,9}$ \\
		File $W_4$ & \textcolor{OliveGreen}{$\eta_{4,1}$} & \textcolor{OliveGreen}{$\eta_{4,2}$} & \textcolor{OliveGreen}{$\eta_{4,3}$} & \tikz[remember picture]{\node[anchor=base,inner sep=-3pt] (d4) {\textcolor{OliveGreen}{$\eta_{4,4}$}};} & $\eta_{4,5}$ & $\eta_{4,6}$ & $\eta_{4,7}$ & 
		\tikz[remember picture]{\node[anchor=base,inner sep=-3pt] (d8) {$\eta_{4,8}$};} & \tikz[remember picture]{\node[anchor=base,inner sep=-3pt](d9) {$\eta_{4,9}$};}\\ [0.5\baselineskip] \multicolumn{1}{c}{} & \multicolumn{8}{c}{\textcolor{OliveGreen}{$\quad\:$ZF symbols}}
		\tikz[remember picture,overlay]{
			\node [draw = red, fit = (a1) (d4)] {};
		}
		\tikz[remember picture,overlay]{
					\node [draw = cyan, fit = (a5) (d8)] {};
		}
		\tikz[remember picture,overlay]{
			\node [draw = blue, fit = (a9) (d9)] {};
		}
	\end{tabular}
}

\newcommand{\ZFKThree}[0]{
	\centering
	\begin{tabular}{c|ccc}
		 & \multicolumn{3}{c}{ZF map}\\ \hline 
		UE$_1$ &  \textcolor{OliveGreen}{$\eta_{2,1}$} & \textcolor{OliveGreen}{$\eta_{2,2}$} & \textcolor{OliveGreen}{$\eta_{3,3}$} \\ \hline 
		UE$_2$ & \textcolor{OliveGreen}{$\eta_{3,1}$} & \textcolor{OliveGreen}{$\eta_{3,2}$} & \textcolor{OliveGreen}{$\eta_{1,3}$} \\ \hline 
		UE$_3$ & \textcolor{OliveGreen}{$\eta_{1,1}$} & \textcolor{OliveGreen}{$\eta_{1,2}$} & \textcolor{OliveGreen}{$\eta_{2,3}$} \\ \hline
	\end{tabular}
}

\newcommand{\ZFMTwoKTwo}[0]{
	\centering
	\begin{tabular}{c|cccccccc}
		& \multicolumn{8}{c}{ZF map}\\ \hline 
		UE$_1$ &  \textcolor{OliveGreen}{$\eta_{2,1}$} & \textcolor{OliveGreen}{$\eta_{2,2}$} & \textcolor{OliveGreen}{$\eta_{2,5}$} & \textcolor{OliveGreen}{$\eta_{2,6}$} & \textcolor{OliveGreen}{$\eta_{3,5}$} & \textcolor{OliveGreen}{$\eta_{3,6}$} & \textcolor{OliveGreen}{$\eta_{4,1}$} & \textcolor{OliveGreen}{$\eta_{4,2}$} \\ \hline 
		UE$_2$ & \textcolor{OliveGreen}{$\eta_{1,1}$} & \textcolor{OliveGreen}{$\eta_{1,2}$} & \textcolor{OliveGreen}{$\eta_{1,5}$} & \textcolor{OliveGreen}{$\eta_{1,6}$} & \textcolor{OliveGreen}{$\eta_{3,7}$} & \textcolor{OliveGreen}{$\eta_{3,8}$} & \textcolor{OliveGreen}{$\eta_{4,3}$} & \textcolor{OliveGreen}{$\eta_{4,4}$} \\ \hline 
	\end{tabular}
}

\newcommand{\AlignGraph}[0]{
	
	\draw[fill=blue!20,draw=white] (-3.5,-0.5) rectangle (-0.5,5.5);
	\draw[fill=green!20,draw=white] (4.5,-0.5) rectangle (-0.5,5.5);
	\draw[fill=magenta!20,draw=white] (4.5,-0.5) rectangle (10,5.5);
	
	\draw [thick,decoration={brace,raise=0.1cm
	},decorate] (-3.5,5.5) -- (-0.5,5.5) 
	node [pos=0.5,anchor=north,yshift=1.15cm,align=left] {\small \textcolor{blue}{Layer 1}: Align 2 \\[-2ex] \small symbols per UE};
	\draw [thick,decoration={brace,raise=0.1cm
	},decorate] (-0.5,5.5) -- (4.5,5.5) 
	node [pos=0.5,anchor=north,yshift=1.15cm,align=left] {\small \textcolor{OliveGreen}{Layer 2}: Align 3 \\[-2ex] \small symbols per UE};
	\draw [thick,decoration={brace,raise=0.1cm
	},decorate] (4.5,5.5) -- (10,5.5) 
	node [pos=0.5,anchor=north,yshift=1.15cm,align=left] {\small \textcolor{magenta}{Layer 3}: Align 3 \\[-2ex] \small symbols per UE};
	
	\begin{scope}[every node/.style={circle,thick,draw}]
			\node (d5) at (-3,2.5) {\footnotesize $\eta_{4,5}$};
			\node (c4) at (-0.5,0) {\footnotesize $\eta_{3,4}$};
			\node (c2) at (2.0,0) {\footnotesize $\eta_{3,2}$};
			\node (b5) at (4.5,0) {\footnotesize $\eta_{2,5}$};
			\node (a3) at (7.0,0) {\footnotesize $\eta_{1,3}$};
			\node (b1) at (9.5,0) {\footnotesize $\eta_{2,1}$};		
			\node (b4) at (-0.5,2.5) {\footnotesize $\eta_{2,4}$};
			\node (b2) at (2.0,2.5) {\footnotesize $\eta_{2,2}$};
			\node (a5) at (4.5,2.5) {\footnotesize $\eta_{1,5}$};
			\node (c3) at (7.0,2.5) {\footnotesize $\eta_{3,3}$};
			\node (a1) at (9.5,2.5) {\footnotesize $\eta_{1,1}$};		
			\node (a4) at (-0.5,5) {\footnotesize $\eta_{1,4}$};
			\node (a2) at (2.0,5) {\footnotesize $\eta_{1,2}$};
			\node (c5) at (4.5,5) {\footnotesize $\eta_{3,5}$};
			\node (b3) at (7.0,5) {\footnotesize $\eta_{2,3}$};
			\node (c1) at (9.5,5) {\footnotesize $\eta_{3,1}$};
		\end{scope}
	
	\begin{scope}[>={Stealth[black]},
		every node/.style={fill=white,circle},
		every edge/.style={draw=red,very thick}]
		\path [<->] (d5) edge node {\footnotesize UE$_1$} (b4);
		\path [<->] (d5) edge node {\footnotesize  UE$_2$} (c4);
		\path [<->] (d5) edge node {\footnotesize  UE$_3$} (a4);
		
		\path [<-] (b4) edge node {\footnotesize  UE$_3$} (b2);
		\path [->] (b2) edge node {\footnotesize  UE$_3$} 
		(a5);
		
		\path [<-] (c4) edge node {\footnotesize  UE$_1$} (c2);
		\path [->] (c2) edge node {\footnotesize  UE$_1$}
		(b5);
		
		\path [<-] (a4) edge node {\footnotesize  UE$_2$} (a2);
		\path [->] (a2) edge node {\footnotesize  UE$_2$}
		(c5);
		
		\path [<-] (a5) edge node {\footnotesize  UE$_2$} (c3);
		\path [->] (c3) edge node {\footnotesize  UE$_2$} 
		(a1);
		
		\path [<-] (b5) edge node {\footnotesize  UE$_3$} (a3);
		\path [->] (a3) edge node {\footnotesize  UE$_3$} 
		(b1);
		
		\path [<-] (c5) edge node {\footnotesize  UE$_1$} (b3);
		\path [->] (b3) edge node {\footnotesize  UE$_1$} 
		(c1);

	\end{scope}
}

\newcommand{\AlignGraphMTwoKTwo}[0]{
	
	\draw[fill=blue!20,draw=white] (-3.5,-3) rectangle (5,6.75);
		
	\begin{scope}[every node/.style={circle,thick,draw}]
		\node (c9) at (-3,1.25) {\footnotesize $\eta_{3,9}$};
		\node (a3) at (-0.5,0) {\footnotesize $\eta_{1,3}$};
		\node (a7) at (2.0,0) {\footnotesize $\eta_{1,7}$};		
		\node (d4) at (-0.5,2.5) {\footnotesize $\eta_{4,4}$};
		\node (b8) at (2.0,2.5) {\footnotesize $\eta_{2,8}$};
		\node (c8) at (4.5,2.5) {\footnotesize $\eta_{3,8}$};

		\node (d9) at (-3,5) {\footnotesize $\eta_{4,9}$};
		\node (b3) at (-0.5,3.75) {\footnotesize $\eta_{2,3}$};
		\node (b7) at (2.0,3.75) {\footnotesize $\eta_{2,7}$};		
		\node (d2) at (-0.5,6.25) {\footnotesize $\eta_{4,2}$};
		\node (a8) at (2.0,6.25) {\footnotesize $\eta_{1,8}$};
		\node (c6) at (4.5,6.25) {\footnotesize $\eta_{3,6}$};
		
		\node (b9) at (-3,-1.25) {\footnotesize $\eta_{2,9}$};
		\node (d3) at (-0.5,-1.25) {\footnotesize $\eta_{4,3}$};
		\node (b4) at (2.0,-1.25) {\footnotesize $\eta_{2,4}$};	
	    \node (c7) at (4.5,-1.25) {\footnotesize $\eta_{3,7}$};	
	    
	    \node (a9) at (-3,-2.5) {\footnotesize $\eta_{1,9}$};
	    \node (d1) at (-0.5,-2.5) {\footnotesize $\eta_{4,1}$};
	    \node (a4) at (2.0,-2.5) {\footnotesize $\eta_{1,4}$};	
	    \node (c5) at (4.5,-2.5) {\footnotesize $\eta_{3,5}$};
	    	    		
	\end{scope}
	
	\begin{scope}[>={Stealth[black]},
		every node/.style={fill=white,circle},
		every edge/.style={draw=red,very thick}]
		\path [<-] (d9) edge node {\footnotesize UE$_2$} (d2);
		\path [<-] (d9) edge node {\footnotesize  UE$_1$} (b3);
		
		\path [<-] (c9) edge node {\footnotesize UE$_1$} (d4);
		\path [<-] (c9) edge node {\footnotesize  UE$_2$} (a3);

		\path [<-] (b9) edge node {\footnotesize UE$_1$} (d3);
		
		\path [<-] (a9) edge node {\footnotesize UE$_2$} (d1);
	
		\path [] (d2) edge node {\footnotesize UE$_2$} (a8);
		\path [->] (a8) edge node {\footnotesize UE$_2$} (c6);

		\path [->] (b3) edge node {\footnotesize UE$_1$} (b7);
		
		\path [] (d4) edge node {\footnotesize UE$_1$} (b8);
		\path [->] (b8) edge node {\footnotesize UE$_1$} (c8);
				
		\path [->] (a3) edge node {\footnotesize UE$_2$} (a7);
		
		\path [] (d3) edge node {\footnotesize UE$_1$} (b4);
		\path [->] (b4) edge node {\footnotesize UE$_1$} (c7);	
		
		\path [] (d1) edge node {\footnotesize UE$_2$} (a4);
		\path [->] (a4) edge node {\footnotesize UE$_2$} (c5);

	\end{scope}
}

\newcommand{\PlotMMuRegions}[0]{
	\begin{axis}[
		axis x line=bottom, axis y line=left,
		ymin=1, ymax=8, ytick={1}, extra y ticks={2,5}, ylabel={$M$},
		xmin=0, xmax=1.05, xtick={0,0.2, 0.4, 1}, xlabel=$\mu$,
		xticklabels={\footnotesize $0$, \scriptsize $\frac{1}{(2K+1)}=\nicefrac{1}{5}$, \scriptsize $\quad\quad\frac{K}{(2K+1)}=\nicefrac{2}{5}$,\scriptsize $1$},
		yticklabels={\footnotesize $1$}, extra y tick labels={\footnotesize $\quad K=2$,\footnotesize $2K+1=5$}, extra y tick style={y tick label style={rotate=90}},
		]
		
		\addplot[name path=f1, very thick, dashed, blue, domain=0:1, smooth]{2};
		
		\path[name path=axis] (axis cs:1,2) -- (axis cs:1,0) -- (axis cs:0,0);
		\addplot[pattern=north west lines, pattern color=OliveGreen] fill between[of=axis and f1];

		\addplot+[name path=f2, no markers, decoration={
		text={}, raise=1ex,
		text align={center}
		}, postaction={decorate}, dashed, very thick, blue, domain=0.4:1, smooth]{2/x};
		
		\addplot[name path=f3, very thick, dashed, blue, domain=0.4:0.4375, smooth]{1/(1-2*x)};
		
		\path[name path=yc] (axis cs:0.4375,8) -- (axis cs:1,8);
		
		\path[name path=xc] (axis cs:1,2) -- (axis cs:1,8);
		
		\addplot[pattern=north west lines, pattern color=magenta!55] fill between[of=f2 and yc];
		
		\addplot+[name path=f4, no markers, decoration={
		text={}, raise=1ex,
		text align={center}
		}, postaction={decorate}, dashed, very thick, blue, domain=0.25:0.4, smooth]{2/x};
		
		\path[name path=yc1] (axis cs:0.25,8) -- (axis cs:0.4375,8);
				
		\addplot[pattern=north west lines, pattern color=black!50] fill between[of=f4 and yc1];
		
		\addplot+[name path=f5, no markers, decoration={
		text={}, raise=1ex,
		text align={center}
		}, postaction={decorate}, dashed, very thick, blue, domain=0.02:0.4, smooth]{2/(2*(1-x))+sqrt(4/(4*(1-x)^(2))+2/(x*(1-x)))};
		
		\addplot[pattern=north west lines, pattern color=cyan!90] fill between[of=f4 and f5];
		
		\foreach \M in {1,2,...,8}{
			\foreach \k in {1,..., \M}{ 
				\addplot[red, mark=otimes*, only marks] coordinates {({\k/\M},{\M})};
			}
		}
		
		\node[anchor=west, scale=0.9] at (axis cs: 0.65,4.5) {\textbf{Region A}};
		\node[anchor=west, rotate=-80, scale=0.9] at (axis cs: 0.335,7.95) {\textbf{Region B}};
		\node[scale=0.9] at (axis cs: 0.5,1.5) {\textbf{Region C}};
		\node[rotate=-60, scale=0.9] at (axis cs: 0.15,3.5) {\textbf{Region D}};
		\node[rotate=-65, scale=0.9] at (axis cs: 0.225,6.5) {\textbf{Region E}};
	\end{axis}
}

\newcommand{\PlotMMuRegionsTwo}[0]{
	\begin{axis}[
		axis x line=bottom, axis y line=left,
		ymin=1, ymax=8.3, ytick={1}, extra y ticks={2,5}, ylabel={$M\qquad\qquad\qquad\qquad\qquad\quad$},
		xmin=0, xmax=1.05, xtick={0,0.2, 0.4, 1}, xlabel=$\mu\qquad\qquad\qquad\quad\qquad\qquad\qquad\quad$,
		xticklabels={\footnotesize $0$, \scriptsize $\frac{1}{(2K+1)}=\nicefrac{1}{5}$, \scriptsize $\quad\quad\frac{K}{(2K+1)}=\nicefrac{2}{5}$,\scriptsize $1$},
		yticklabels={\footnotesize $1$}, extra y tick labels={\footnotesize $K$,\footnotesize $2K+1$}, extra y tick style={y tick label style={rotate=90}},
		]
		\addplot[name path=f1, very thick, dashed, cyan, domain=0:1, smooth]{2} node[above, pos=0.35] {$K$};
		
		\path[name path=axis] (axis cs:1,2) -- (axis cs:1,0) -- (axis cs:0,0);
		\addplot[pattern=north west lines, pattern color=OliveGreen] fill between[of=axis and f1];
		
		\addplot+[name path=f2, no markers, decoration={
			text={}, raise=1ex,
			text align={center}
		}, postaction={decorate}, dashed, very thick, blue, domain=0.4:1, smooth]{2/x} node[above, pos=0.95, scale=1.2] {$\frac{K}{\mu}$};
		
		\addplot[name path=f3, very thick, dashed, OliveGreen, domain=0.4:0.4375, smooth]{1/(1-2*x)} node[below, pos=0.65, scale=1.05] {$\qquad\frac{1}{1-2\mu}$};
		
		\path[name path=yc] (axis cs:0.4375,8) -- (axis cs:1,8);
		
		\path[name path=xc] (axis cs:1,2) -- (axis cs:1,8);
		
		\addplot[pattern=north west lines, pattern color=magenta!55] fill between[of=f2 and yc];
		
		\addplot+[name path=f4, no markers, decoration={
			text={}, raise=1ex,
			text align={center}
		}, postaction={decorate}, dashed, very thick, blue, domain=0.25:0.4, smooth]{2/x};
		
		\path[name path=yc1] (axis cs:0.25,8) -- (axis cs:0.4375,8);
		
		\addplot[pattern=north west lines, pattern color=black!50] fill between[of=f4 and yc1];
		
		\addplot+[name path=f5, no markers, decoration={
			text={}, raise=1ex,
			text align={center}
		}, postaction={decorate}, dashed, very thick, magenta, domain=0.044:0.4, smooth]{2/(2*(1-x))+sqrt(4/(4*(1-x)^(2))+2/(x*(1-x)))} node[below, pos=0.925, scale=0.8] {$K=\mu M\delta_{\text{MAN}}(\mu)$}; %{\scriptsize$K=\mu M\delta_{\text{MAN}}(\mu)$};
		
		\addplot +[mark=none, dashed, very thick, violet] coordinates {(0.5, 4) (0.5, 6)} node[below, pos=0.375, scale=0.8] {$\qquad\quad\mu=\frac{1}{2}$}; %{\scriptsize$\qquad\quad\mu=\frac{1}{2}$};
		\addplot +[mark=none, dashed, very thick, violet] coordinates {(0.5, 7) (0.5, 8)};
		
		\addplot[pattern=north west lines, pattern color=cyan!90] fill between[of=f4 and f5];
		
		\foreach \M in {1,2,...,8}{
			\foreach \k in {0,..., \M}{ 
				\addplot[red, mark=otimes*, only marks] coordinates {({\k/\M},{\M})};
			}
		}
		
		\node[anchor=west, scale=0.9, align=left] at (axis cs: 0.60,4.5) {\textbf{Region A}\\[-0.95ex] (Constant NDT-plateau)};
		\node[anchor=west, rotate=-80, scale=0.9] at (axis cs: 0.335,7.65) {\textbf{Region B}};
		\node[scale=0.9] at (axis cs: 0.45,1.5) {\textbf{Region C} ($\#$ RN bottleneck)};
		\node[rotate=-60, scale=0.9, align=center] at (axis cs: 0.125,3.2) {\textbf{Region D} \\[-0.75ex]($\#$ UE bottleneck) };
		\node[rotate=-65, scale=0.9] at (axis cs: 0.225,6.5) {\textbf{Region E}};
		
	\end{axis}
	\draw[->, thick, black] (-0.75,1.25) -- (0.5,1.25);
	\draw[->, thick, black] (-0.75,1.25) -- (1.35,0.4);
	\node[align=center, scale=0.85] at (-1.35,1.25) {Interference \\[-0.75ex](channel)\\[-0.75ex]-limited\\[-0.75ex](UE side)};
	\draw[->, thick, black] (1.6,6.25) -- (1.3,4.85);
	\draw[->, thick, black] (1.6,6.25) -- (2.2,5.15);
	\node[align=center, scale=0.85] at (1.6,6.55) {Broadcast (channel)-limited\\[-0.75ex](RN side)};
	\draw[->, thick, blue] (7.0,1.4) -- (6.4,0.9);
	\node[align=center, scale=0.85, color=blue] at (7.6,1.44) {RN standalone \\[-0.75ex]frontier};
}

\newcommand{\AchMTwoKone}[0]{
	\begin{tikzpicture}[,
	roundnode/.style={circle, draw=green!60, fill=green!5, very thick, minimum size=7mm},
	squarednode/.style={rectangle, draw=red!60, fill=red!5, very thick, minimum size=5mm},
	]
	\node[squarednode] (eNB) at (0.0, 0) {DeNB};
	\node[squarednode, fill=magenta!10, draw=black] (RN1) at (3, 3) {RN$_1$};
	\node[squarednode, fill=cyan!20, draw=black] (RN2) at (3, -3) {RN$_2$};
	\node[squarednode, fill=green!10, draw=black] (Uone) at (6, 0) {UE$_1$};
	
	\draw[->, thick, black] (0.7, 0) -- (2.45, 3) node[pos=0.5,sloped,above,scale=0.9] {$f_1[t]$};
	\draw[->, thick, black] (0.7, 0) -- (2.45, -3) node[pos=0.5,sloped,below,scale=0.9] {$f_2[t]$};
	\draw[->, black, thick] (0.7, 0) -- (5.5, 0) node[pos=0.5,sloped,above,scale=0.9] {$g_{1}[t]$};
	\draw[->, black, thick] (3.55, 3) -- (5.45, 0.1) node[pos=0.5,sloped,above,scale=0.9] {$h_{11}[t]$};
	\draw[->, black, thick] (3.55, -3) -- (5.45, 0.1) node[pos=0.5,sloped,below,scale=0.9] {$h_{12}[t]$};
	
	\node[scale=0.8] at (-4.75, 0.5) {$t$};
	\node[scale=0.8, anchor=west] at (-4.25, 0.5) {$x_{s}[t]$};
	\node[scale=0.8] at (-4.75, 0) {1};
	\node[scale=0.8, anchor=west] at (-4.25, 0) {$\nu_{\eta_{2,2}}[t]\textcolor{magenta}{\eta_{2,2}}+\nu_{\eta_{3,1}}[t]\textcolor{cyan}{\eta_{3,1}}$};
	\draw[] (-5,0.25) -- (-1,0.25); 
	\draw[] (-4.35,-0.75) -- (-4.35,0.75);
	\node[scale=0.8] at (-4.75, -0.5) {2};
	\node[scale=0.8] at (-2.55, -0.5) {---};
	\draw[] (-5,-0.25) -- (-1,-0.25); 
	
	\node[scale=0.8] at (1.5-0.75-0.25, 4.5+0.5) {$t$};
	\node[scale=0.8, anchor=west] at (2.0-0.75-0.25, 4.5+0.5) {$y_{r,1}'[t]$};
	\node[scale=0.8, anchor=west] at (5.35-0.75, 4.5+0.5) {$x_{r,1}[t]$};
	\node[scale=0.8] at (1.5-0.75-0.25, 4+0.5) {1};
	\node[scale=0.8, anchor=west] at (2.0-0.75-0.3, 4+0.5) {$f_1[t]\nu_{\eta_{2,2}}[t]\textcolor{magenta}{\eta_{2,2}}+z_{r,1}[t]$};
	\node[scale=0.8, anchor=west] at (5.35-0.75, 4+0.5) {$\beta_{\eta_{3,1}}[t]\textcolor{cyan}{\eta_{3,1}}+\textcolor{OliveGreen}{\eta_{1,1}}$};
	\draw[] (1.25-0.75-0.25,4.25+0.5) -- (7.15,4.25+0.5); 
	\draw[] (1.25-0.75-0.25,4.25) -- (7.15,4.25); 
	\draw[] (1.9-0.75-0.25,3.75) -- (1.9-0.75-0.25,4.75+0.5);
	\draw[] (5.175-0.75,3.75) -- (5.175-0.75,4.75+0.5); \draw[] (5.25-0.75,3.75) -- (5.25-0.75,4.75+0.5);
	\node[scale=0.8] at (1.5-0.75-0.25, 3.5+0.5) {2};
	\node[scale=0.8] at (3.5-0.75, 3.5+0.5) {---};
	\node[scale=0.8] at (7.1-1.25, 3.5+0.5) {---};
	\node[scale=0.72] at (2.75, 1.75) {$S_1=\begin{pmatrix}
		\textcolor{OliveGreen}{\eta_{1,1}} \\ \textcolor{magenta}{\eta_{2,1}} \\ \textcolor{cyan}{\eta_{3,1}} 
		\end{pmatrix}$};
	
	\node[scale=0.8] at (1.5-0.75-0.25, 4.5+0.5-9) {$t$};
	\node[scale=0.8, anchor=west] at (2.0-0.75-0.25, 4.5+0.5-9) {$y_{r,2}'[t]$};
	\node[scale=0.8, anchor=west] at (5.35-0.75, 4.5+0.5-9) {$x_{r,2}[t]$};
	\node[scale=0.8] at (1.5-0.75-0.25, 4+0.5-9) {1};
	\node[scale=0.8, anchor=west] at (2.0-0.75-0.3, 4+0.5-9) {$f_2[t]\nu_{\eta_{3,1}}[t]\textcolor{cyan}{\eta_{3,1}}+z_{r,2}[t]$};
	\node[scale=0.8, anchor=west] at (5.35-0.75, 4+0.5-9) {$\beta_{\eta_{2,2}}[t]\textcolor{magenta}{\eta_{2,2}}$};
	\draw[] (1.25-0.75-0.25,4.25+0.5-9) -- (7.15,4.25+0.5-9); 
	\draw[] (1.25-0.75-0.25,4.25-9) -- (7.15,4.25-9); 
	\draw[] (1.9-0.75-0.25,3.75-9) -- (1.9-0.75-0.25,4.75+0.5-9);
	\draw[] (5.175-0.75,3.75-9) -- (5.175-0.75,4.75+0.5-9); \draw[] (5.25-0.75,3.75-9) -- (5.25-0.75,4.75+0.5-9);
	\node[scale=0.8] at (1.5-0.75-0.25, 3.5+0.5-9) {2};
	\node[scale=0.8] at (3.5-0.75, 3.5+0.5-9) {---};
	\node[scale=0.8, anchor=west] at (5.35-0.75, 3.5+0.5-9) {$\textcolor{OliveGreen}{\eta_{1,2}}$};
	\node[scale=0.72] at (2.75, -1.75) {$S_2=\begin{pmatrix}
		\textcolor{OliveGreen}{\eta_{1,2}} \\ \textcolor{magenta}{\eta_{2,2}} \\ \textcolor{cyan}{\eta_{3,2}} 
		\end{pmatrix}$};
	
	\node[scale=0.8] at (-4.75+11.75, 0.5) {$t$};
	\node[scale=0.8, anchor=west] at (-4.25+11.75, 0.5) {$y_{u,1}[t]$};
	\node[scale=0.8] at (-4.75+11.75, 0) {1};
	\node[scale=0.8, anchor=west] at (-4.25+11.75, 0) {$h_{11}[t]\textcolor{OliveGreen}{\eta_{1,1}}+z_{u,1}[t]$};
	\draw[] (-5+11.75,0.25) -- (-1.5+11.75,0.25); 
	\draw[] (-4.35+11.75,-0.75) -- (-4.35+11.75,0.75);
	\node[scale=0.8] at (-4.75+11.75, -0.5) {2};
	\node[scale=0.8, anchor=west] at (-4.25+11.75, -0.5) {$h_{12}[t]\textcolor{OliveGreen}{\eta_{1,2}}+z_{u,1}[t]$};
	\draw[] (-5+11.75,-0.25) -- (-1.5+11.75,-0.25); 
	\end{tikzpicture}
}

\newcommand{\AchMTwoKTwo}[0]{
	\begin{tikzpicture}[,
	roundnode/.style={circle, draw=green!60, fill=green!5, very thick, minimum size=7mm},
	squarednode/.style={rectangle, draw=red!60, fill=red!5, very thick, minimum size=5mm},
	]
	\node[squarednode] (eNB) at (0.0, 0) {DeNB};
	\node[squarednode, fill=magenta!10, draw=black] (RN1) at (3.5, 3) {RN$_1$};
	\node[squarednode, fill=cyan!20, draw=black] (RN2) at (3.5, -3) {RN$_2$};
	\node[squarednode, fill=green!10, draw=black] (Uone) at (7, 3) {UE$_1$};
	\node[squarednode, fill=orange!10, draw=black] (Utwo) at (7, -3) {UE$_2$};
	
	\draw[->, thick, black] (0.7, 0) -- (2.95, 3) node[pos=0.5,sloped,above,scale=0.9] {$f_1[t]$};
	\draw[->, thick, black] (0.7, 0) -- (2.95, -3) node[pos=0.5,sloped,below, scale=0.9] {$f_2[t]$};
	\draw[->, black, thick] (0.7, 0) to [out=0, in=-135] (6.45, 2.95);
	\node[scale=0.9, rotate=20.5] at (2.25, 0.45) {$g_1[t]$};
	\node[scale=0.9, rotate=-20.5] at (2.25, -0.45) {$g_2[t]$};	
	\draw[->, black, thick] (0.7, 0) to [out=0, in=135] (6.45, -2.95);
	\draw[->, black, thick] (4.05, 3) -- (6.5, 3) node[pos=0.5,sloped,above, scale=0.9] {$h_{11}[t]$};
	\draw[->, black, thick] (4.05, 3) -- (6.45, -2.9) node[pos=0.15,sloped,above, scale=0.9] {$h_{21}[t]$};
	\draw[->, black, thick] (4.05, -3) -- (6.45, 2.9) node[pos=0.15,sloped,below, scale=0.9] {$h_{12}[t]$};
	\draw[->, black, thick] (4.05, -3) -- (6.5, -3) node[pos=0.5,sloped,below, scale=0.9] {$h_{22}[t]$};
	
	\node[scale=0.8] at (-4.75, 0.5+1) {$t$};
	\node[scale=0.8, anchor=west] at (-4.25, 0.5+1) {$x_{s}[t]$};
	\node[scale=0.8] at (-4.75, 0+1) {1};
	\node[scale=0.8, anchor=west] at (-4.25, 0+1) {$\nu_{\eta_{3,3}}[t]\textcolor{magenta}{\eta_{3,3}}+\nu_{\eta_{4,1}}[t]\textcolor{cyan}{\eta_{4,1}}$};
	\node[scale=0.8] at (-4.75, -0.5+1) {2};
	\node[scale=0.8, anchor=west] at (-4.25, -0.5+1) {$\nu_{\eta_{3,4}}[t]\textcolor{magenta}{\eta_{3,4}}+\nu_{\eta_{4,2}}[t]\textcolor{cyan}{\eta_{4,2}}$};
	\node[scale=0.8] at (-4.75, -1+1) {3};
	\node[scale=0.8, anchor=west] at (-4.25, -1+1) {$\nu_{\eta_{1,2}}[t]\textcolor{OliveGreen}{\eta_{1,2}}+\nu_{\eta_{2,4}}[t]\textcolor{orange}{\eta_{2,4}}$};
	\node[scale=0.8] at (-4.75, -1.5+1) {4};
	\node[scale=0.8, anchor=west] at (-4.25, -1.5+1) {$\nu_{\eta_{1,3}}[t]\textcolor{OliveGreen}{\eta_{1,3}}+\nu_{\eta_{2,1}}[t]\textcolor{orange}{\eta_{2,1}}$};
	\node[scale=0.8] at (-4.75, -2+1) {5};
	\node[scale=0.8, anchor=west] at (-4.25, -2+1) {$\nu_{\eta_{1,4}}[t]\textcolor{OliveGreen}{\eta_{1,4}}+\nu_{\eta_{2,2}}[t]\textcolor{orange}{\eta_{2,2}}$};
	\draw[] (-5,0.25) -- (-1,0.25); 
	\draw[] (-5,0.75) -- (-1,0.75); 
	\draw[] (-5,1.25) -- (-1,1.25); 
	\draw[] (-5,-0.25) -- (-1,-0.25); 
	\draw[] (-5,-0.75) -- (-1,-0.75); 
	\draw[] (-4.35,-1.25) -- (-4.35,1.75);

	\node[scale=0.8] at (1.5-0.75-0.25, 4.5+0.5+1.5) {$t$};
	\node[scale=0.8, anchor=west] at (2.0-0.75-0.25, 4.5+0.5+1.5) {$y_{r,1}'[t]$};
	\node[scale=0.8, anchor=west] at (5.35-0.75, 4.5+0.5+1.5) {$x_{r,1}[t]$};
	\node[scale=0.8] at (1.5-0.75-0.25, 4+0.5+1.5) {1};
	\node[scale=0.8, anchor=west] at (2.0-0.75-0.3, 4+0.5+1.5) {$f_1[t]\nu_{\eta_{3,3}}[t]\textcolor{magenta}{\eta_{3,3}}+z_{r,1}[t]$};
	\node[scale=0.8, anchor=west] at (5.35-0.75, 4+0.5+1.5) {$\beta_{\eta_{4,1}}[t]\textcolor{cyan}{\eta_{4,1}}+\textcolor{OliveGreen}{\eta_{1,1}}$};
	\node[scale=0.8] at (1.5-0.75-0.25, 3.5+0.5+1.5) {2};
	\node[scale=0.8, anchor=west] at (2.0-0.75-0.3, 3.5+0.5+1.5) {$f_1[t]\nu_{\eta_{3,4}}[t]\textcolor{magenta}{\eta_{3,4}}+z_{r,1}[t]$};
	\node[scale=0.8, anchor=west] at (5.35-0.75, 3.5+0.5+1.5) {$\beta_{\eta_{4,2}}[t]\textcolor{cyan}{\eta_{4,2}}$};
	\node[scale=0.8] at (1.5-0.75-0.25, 3.5+0.5-0.5+1.5) {3};
	\node[scale=0.8, anchor=west] at (2.0-0.75-0.3, 3.5+0.5-0.5+1.5) {$f_1[t]\nu_{\eta_{2,4}}[t]\textcolor{orange}{\eta_{2,4}}+z_{r,1}[t]$};
	\node[scale=0.8, anchor=west] at (5.35-0.75, 3.5+0.5-0.5+1.5) {$\beta_{\eta_{1,2}}[t]\textcolor{OliveGreen}{\eta_{1,2}}$};
	\node[scale=0.8] at (1.5-0.75-0.25, 3.5+0.5-0.5-0.5+1.5) {4};
	\node[scale=0.8, anchor=west] at (2.0-0.75-0.3, 3.5+0.5-0.5-0.5+1.5) {$f_1[t]\nu_{\eta_{1,3}}[t]\textcolor{OliveGreen}{\eta_{1,3}}+z_{r,1}[t]$};
	\node[scale=0.8, anchor=west] at (5.35-0.75, 3.5+0.5-0.5-0.5+1.5) {$\beta_{\eta_{2,1}}[t]\textcolor{orange}{\eta_{2,1}}$};
	\node[scale=0.8] at (1.5-0.75-0.25, 3.5+0.5-0.5-0.5-0.5+1.5) {5};
	\node[scale=0.8, anchor=west] at (2.0-0.75-0.3, 3.5+0.5-0.5-0.5-0.5+1.5) {$f_1[t]\nu_{\eta_{1,4}}[t]\textcolor{OliveGreen}{\eta_{1,4}}+z_{r,1}[t]$};
	\node[scale=0.8, anchor=west] at (5.35-0.75, 3.5+0.5-0.5-0.5-0.5+1.5) {$\beta_{\eta_{2,2}}[t]\textcolor{orange}{\eta_{2,2}}$};
	\draw[] (1.25-0.75-0.25,4.25+0.5+1.5) -- (7.15,4.25+0.5+1.5); 
	\draw[] (1.25-0.75-0.25,4.25+0.5-0.5+1.5) -- (7.15,4.25+0.5-0.5+1.5);
	\draw[] (1.25-0.75-0.25,4.25+0.5-0.5-0.5+1.5) -- (7.15,4.25+0.5-0.5-0.5+1.5);
	\draw[] (1.25-0.75-0.25,4.25+0.5-0.5-1+1.5) -- (7.15,4.25+0.5-0.5-1+1.5);
	\draw[] (1.25-0.75-0.25,4.25+0.5-0.5-1.5+1.5) -- (7.15,4.25+0.5-0.5-1.5+1.5);
	\draw[] (1.9-0.75-0.25,3.75) -- (1.9-0.75-0.25,4.75+0.5+1.5);
	\draw[] (5.175-0.75,3.75) -- (5.175-0.75,4.75+0.5+1.5); \draw[] (5.25-0.75,3.75) -- (5.25-0.75,4.75+0.5+1.5);
	\node[scale=0.72] at (0.1, 2.5) {$S_1=\setlength\arraycolsep{2pt}\begin{pmatrix}
		\textcolor{OliveGreen}{\eta_{1,1}} & \textcolor{OliveGreen}{\eta_{1,2}} \\ \textcolor{orange}{\eta_{2,1}} & \textcolor{orange}{\eta_{2,2}} \\ \textcolor{magenta}{\eta_{3,1}} & \textcolor{magenta}{\eta_{3,2}} \\ \textcolor{cyan}{\eta_{4,1}} & \textcolor{cyan}{\eta_{4,2}}
		\end{pmatrix}$};
	
	\node[scale=0.8] at (1.5-0.75-0.25, 4.5+0.5+1.5-10.5) {$t$};
	\node[scale=0.8, anchor=west] at (2.0-0.75-0.25, 4.5+0.5+1.5-10.5) {$y_{r,2}'[t]$};
	\node[scale=0.8, anchor=west] at (5.35-0.75, 4.5+0.5+1.5-10.5) {$x_{r,2}[t]$};
	\node[scale=0.8] at (1.5-0.75-0.25, 4+0.5+1.5-10.5) {1};
	\node[scale=0.8, anchor=west] at (2.0-0.75-0.3, 4+0.5+1.5-10.5) {$f_2[t]\nu_{\eta_{4,1}}[t]\textcolor{cyan}{\eta_{4,1}}+z_{r,2}[t]$};
	\node[scale=0.8, anchor=west] at (5.35-0.75, 4+0.5+1.5-10.5) {$\beta_{\eta_{3,3}}[t]\textcolor{magenta}{\eta_{3,3}}$};
	\node[scale=0.8] at (1.5-0.75-0.25, 3.5+0.5+1.5-10.5) {2};
	\node[scale=0.8, anchor=west] at (2.0-0.75-0.3, 3.5+0.5+1.5-10.5) {$f_2[t]\nu_{\eta_{4,2}}[t]\textcolor{cyan}{\eta_{4,2}}+z_{r,2}[t]$};
	\node[scale=0.8, anchor=west] at (5.35-0.75, 3.5+0.5+1.5-10.5) {$\beta_{\eta_{3,4}}[t]\textcolor{magenta}{\eta_{3,4}}+\textcolor{orange}{\eta_{2,3}}$};
	\node[scale=0.8] at (1.5-0.75-0.25, 3.5+0.5-0.5+1.5-10.5) {3};
	\node[scale=0.8, anchor=west] at (2.0-0.75-0.3, 3.5+0.5-0.5+1.5-10.5) {$f_2[t]\nu_{\eta_{1,2}}[t]\textcolor{OliveGreen}{\eta_{1,2}}+z_{r,2}[t]$};
	\node[scale=0.8, anchor=west] at (5.35-0.75, 3.5+0.5-0.5+1.5-10.5) {$\beta_{\eta_{2,4}}[t]\textcolor{orange}{\eta_{2,4}}$};
	\node[scale=0.8] at (1.5-0.75-0.25, 3.5+0.5-0.5-0.5+1.5-10.5) {4};
	\node[scale=0.8, anchor=west] at (2.0-0.75-0.3, 3.5+0.5-0.5-0.5+1.5-10.5) {$f_2[t]\nu_{\eta_{2,1}}[t]\textcolor{orange}{\eta_{2,1}}+z_{r,2}[t]$};
	\node[scale=0.8, anchor=west] at (5.35-0.75, 3.5+0.5-0.5-0.5+1.5-10.5) {$\beta_{\eta_{1,3}}[t]\textcolor{OliveGreen}{\eta_{1,3}}$};
	\node[scale=0.8] at (1.5-0.75-0.25, 3.5+0.5-0.5-0.5-0.5+1.5-10.5) {5};
	\node[scale=0.8, anchor=west] at (2.0-0.75-0.3, 3.5+0.5-0.5-0.5-0.5+1.5-10.5) {$f_2[t]\nu_{\eta_{2,2}}[t]\textcolor{orange}{\eta_{2,2}}+z_{r,2}[t]$};
	\node[scale=0.8, anchor=west] at (5.35-0.75, 3.5+0.5-0.5-0.5-0.5+1.5-10.5) {$\beta_{\eta_{1,4}}[t]\textcolor{OliveGreen}{\eta_{1,4}}$};
	\draw[] (1.25-0.75-0.25,4.25+0.5+1.5-10.5) -- (7.15,4.25+0.5+1.5-10.5); 
	\draw[] (1.25-0.75-0.25,4.25+0.5-0.5+1.5-10.5) -- (7.15,4.25+0.5-0.5+1.5-10.5);
	\draw[] (1.25-0.75-0.25,4.25+0.5-0.5-0.5+1.5-10.5) -- (7.15,4.25+0.5-0.5-0.5+1.5-10.5);
	\draw[] (1.25-0.75-0.25,4.25+0.5-0.5-1+1.5-10.5) -- (7.15,4.25+0.5-0.5-1+1.5-10.5);
	\draw[] (1.25-0.75-0.25,4.25+0.5-0.5-1.5+1.5-10.5) -- (7.15,4.25+0.5-0.5-1.5+1.5-10.5);
	\draw[] (1.9-0.75-0.25,3.75-10.5) -- (1.9-0.75-0.25,4.75+0.5+1.5-10.5);
	\draw[] (5.175-0.75,3.75-10.5) -- (5.175-0.75,4.75+0.5+1.5-10.5); \draw[] (5.25-0.75,3.75-10.5) -- (5.25-0.75,4.75+0.5+1.5-10.5);
	\node[scale=0.72] at (0.1, -2.5) {$S_2=\setlength\arraycolsep{2pt}\begin{pmatrix}
		\textcolor{OliveGreen}{\eta_{1,3}} & \textcolor{OliveGreen}{\eta_{1,4}} \\ \textcolor{orange}{\eta_{2,3}} & \textcolor{orange}{\eta_{2,4}} \\ \textcolor{magenta}{\eta_{3,3}} & \textcolor{magenta}{\eta_{3,4}} \\ \textcolor{cyan}{\eta_{4,3}} & \textcolor{cyan}{\eta_{4,4}}
		\end{pmatrix}$};
	
	\node[scale=0.8] at (-4.75+12.75, 0.5+1+3) {$t$};
	\node[scale=0.8, anchor=west] at (-4.25+12.75, 0.5+1+3) {$y_{u,1}[t]-z_{u,1}[t]$};
	\node[scale=0.8] at (-4.75+12.75, 0+1+3) {1};
	\node[scale=0.8, anchor=west] at (-4.25+12.75, 0+1+3) {$h_{11}[t]\textcolor{OliveGreen}{\eta_{1,1}}$};
	\node[scale=0.8] at (-4.75+12.75, -0.5+1+3) {2};
	\node[scale=0.8, anchor=west] at (-4.25+12.75, -0.5+1+3) {$I(\textcolor{orange}{\eta_{2,3}},\textcolor{magenta}{\eta_{3,4}},\textcolor{cyan}{\eta_{4,2}})$};
	\node[scale=0.8] at (-4.75+12.75, -1+1+3) {3};
	\node[scale=0.8, anchor=west] at (-4.25+12.75, -1+1+3) {$l_{13}[t]\textcolor{OliveGreen}{\eta_{1,2}}$};
	\node[scale=0.8] at (-4.75+12.75, -1.5+1+3) {4};
	\node[scale=0.8, anchor=west] at (-4.25+12.75, -1.5+1+3) {$l_{23}[t]\textcolor{OliveGreen}{\eta_{1,3}}$};
	\node[scale=0.8] at (-4.75+12.75, -2+1+3) {5};
	\node[scale=0.8, anchor=west] at (-4.25+12.75, -2+1+3) {$l_{23}[t]\textcolor{OliveGreen}{\eta_{1,4}}$};
	\draw[] (-5+12.75,0.25+3) -- (-1.85+12.75,0.25+3); 
	\draw[] (-5+12.75,0.75+3) -- (-1.85+12.75,0.75+3); 
	\draw[] (-5+12.75,1.25+3) -- (-1.85+12.75,1.25+3); 
	\draw[] (-5+12.75,-0.25+3) -- (-1.85+12.75,-0.25+3); 
	\draw[] (-5+12.75,-0.75+3) -- (-1.85+12.75,-0.75+3); 
	\draw[] (-4.35+12.75,-1.25+3) -- (-4.35+12.75,1.75+3);
	
	\node[scale=0.8] at (-4.75+12.75, 0.5+1-3) {$t$};
	\node[scale=0.8, anchor=west] at (-4.25+12.75, 0.5+1-3) {$y_{u,2}[t]-z_{u,2}[t]$};
	\node[scale=0.8] at (-4.75+12.75, 0+1-3) {1};
	\node[scale=0.8, anchor=west] at (-4.25+12.75, 0+1-3) {$I(\textcolor{OliveGreen}{\eta_{1,1}},\textcolor{magenta}{\eta_{3,3}},\textcolor{cyan}{\eta_{4,1}})$};
	\node[scale=0.8] at (-4.75+12.75, -0.5+1-3) {2};
	\node[scale=0.8, anchor=west] at (-4.25+12.75, -0.5+1-3) {$h_{22}[t]\textcolor{orange}{\eta_{2,3}}$};
	\node[scale=0.8] at (-4.75+12.75, -1+1-3) {3};
	\node[scale=0.8, anchor=west] at (-4.25+12.75, -1+1-3) {$l_{23}[t]\textcolor{orange}{\eta_{2,4}}$};
	\node[scale=0.8] at (-4.75+12.75, -1.5+1-3) {4};
	\node[scale=0.8, anchor=west] at (-4.25+12.75, -1.5+1-3) {$l_{13}[t]\textcolor{orange}{\eta_{2,1}}$};
	\node[scale=0.8] at (-4.75+12.75, -2+1-3) {5};
	\node[scale=0.8, anchor=west] at (-4.25+12.75, -2+1-3) {$l_{13}[t]\textcolor{orange}{\eta_{2,2}}$};
	\draw[] (-5+12.75,0.25-3) -- (-1.85+12.75,0.25-3); 
	\draw[] (-5+12.75,0.75-3) -- (-1.85+12.75,0.75-3); 
	\draw[] (-5+12.75,1.25-3) -- (-1.85+12.75,1.25-3); 
	\draw[] (-5+12.75,-0.25-3) -- (-1.85+12.75,-0.25-3); 
	\draw[] (-5+12.75,-0.75-3) -- (-1.85+12.75,-0.75-3); 
	\draw[] (-4.35+12.75,-1.25-3) -- (-4.35+12.75,1.75-3);
	\end{tikzpicture}
}

\ifCLASSINFOpdf
\else
\fi
\usepackage{amsmath}
\usepackage[cmintegrals]{newtxmath}
\newcommand{\alert}[1]{\textcolor{black}{#1}}
\newcommand{\Alert}[1]{\textcolor{black}{#1}}
\begin{document}
	
%\title{\alert{Download Time in Cache-Assisted Broadcast-Relay Wireless Networks: A Cache-Storage Latency Tradeoff}}

\title{\alert{Cache-Assisted Broadcast-Relay Wireless Networks: A Delivery-Time Cache-Memory Tradeoff}}

\author{Jaber~Kakar,~\IEEEmembership{Member,~IEEE,}
		Alaa~Alameer,~\IEEEmembership{Member,~IEEE,}
        Anas~Chaaban,~\IEEEmembership{Senior Member,~IEEE,}
        Aydin~Sezgin,~\IEEEmembership{Senior Member,~IEEE,}
        and~Arogyaswami~Paulraj,~\IEEEmembership{Fellow,~IEEE}% <-this % stops a space
%\thanks{This paper was presented in part at the IEEE International Conference on Communications 2017 \cite{Kakar}.}
}

\markboth{Draft}%
{Shell \MakeLowercase{\textit{et al.}}: Bare Demo of IEEEtran.cls for IEEE Communications Society Journals}

\makeatletter
\newcommand*{\rom}[1]{\expandafter\@slowromancap\romannumeral #1@}
\makeatother

\maketitle

\begin{abstract}
An emerging trend of next generation communication systems is to provide network edges with additional capabilities such as storage resources in the form of caches to reduce file delivery latency. To investigate this aspect, we study the fundamental limits of a cache-aided broadcast-relay wireless network consisting of one central base station, $M$ \alert{cache-equipped} transceivers and $K$ receivers from a latency-centric perspective. We use the normalized delivery time (NDT) to capture the per-bit latency for the worst-case file request pattern at high signal-to-noise ratios (SNR), normalized with respect to a reference interference-free system with unlimited transceiver cache capabilities. \alert{The objective is to design the schemes for cache placement and file delivery in order to minimize the NDT. To this end, we establish a novel converse (for arbitrary $M$ and $K$) and two types of achievability schemes applicable to both time-variant and  invariant channels. The \emph{first} scheme is a general \emph{one-shot} scheme for any $M$ and $K$ that synergistically exploits both multicasting (coded) caching and distributed zero-forcing opportunities. Apart from the obvious advantage of low signaling complexity, we show that the proposed one-shot scheme (i) attains gains attributed to both individual and collective transceiver caches %observable in the seminal work of Maddah-Ali Niesen on coded caching, 
(ii) is NDT-optimal for various parameter settings, particularly at higher cache sizes. The \emph{second} scheme, on the other hand, designs beamformers to facilitate both subspace interference alignment and zero-forcing \alert{at lower cache sizes}. Exploiting both schemes, we are able to characterize for various special cases of $M$ and $K$ which satisfy $K+M\leq 4$ the \emph{optimal} tradeoff between cache storage and latency. The tradeoff illustrates that the NDT is the preferred choice to capture the latency of a system rather than the commonly used sum degrees-of-freedom (DoF). In fact, our optimal tradeoff refutes the popular belief that increasing cache sizes translates to increasing the achievable sum DoF. As such, we identify and discuss cases where increasing cache sizes decreases both the delivery time and the achievable DoF.} 
%This is facilitated through establishing a novel converse (for arbitrary $M$ and $K$) and an achievability scheme on the NDT. Our achievability scheme is a synergistic combination of multicasting, zero-forcing beamforming and interference alignment. Further, \alert{in order to decrease} signaling complexity, we also propose a general \emph{one-shot} scheme feasible for any $M$ and $K$ that maintains both local and global caching gains observable in the seminal work of Maddah-Ali Niesen on coded caching. In this regard, we characterize its performance, including its optimality, for various parameter settings through a multiplicative gap comparison with the lower bounds.               
\end{abstract}
%-- one-shot and alignment-based schemes -- on the NDT that synergistically integrate multicasting, zero-forcing beamforming and subspace interference alignment

\IEEEpeerreviewmaketitle

\section{Introduction}
\label{sec:intro}

In the last decade, mobile usage in wireless networks has shifted from being connection-centric driven (e.g., phone calls) to content-centric (e.g., HD video) behaviors \cite{Bastug}. In this context, integrating \emph{content caching} in heterogeneous networks (HetNet) represents a viable solution for highly content-centric next generation (5G) mobile networks. \Alert{Specifically, caching the most popular contents in  HetNet \emph{edge nodes}, e.g., eNBs and relays, alleviates backhaul traffic, reduces latency and ameliorates quality of service of mobile users. For example, edge caches were used to balance backhaul costs against transmission power costs which results in optimal sparse beamforming solutions \cite{Tao16}. Further, in \cite{Rezvani17} caching was deployed to minimize the weighted average latency subject to proportional fairness and ergodic resource allocation constraints. In these previous works, the cache was placed only at the base stations. However, it is expected that future networks will be heterogeneous in nature, vastly deploying relay nodes (RN) (e.g., fixed RNs in LTE-A \cite{network_m2_2011} or mobile RNs in form of drones \cite{Kakar_Thesis,KakarUAV}) endowed with content cache capabilities. In this work, we assume that RNs not only provide files to mobile users but also have their own requests from the central base station in a HetNet scenario.} 

A simplistic HetNet modeling this aspect is shown in Fig. \ref{fig:HetNet}. In this model, $M$ RNs act as cache-aided transceivers. Thus, aspects of both transmitter and receiver caching in RNs \alert{are} captured through this network model enabling a low \emph{delivery time} of requested files by $M$ RNs and $K$ user equipments (UE).\footnote{We use the words \emph{delivery time} and \emph{latency} interchangeably.} \alert{Delivery time refers} to the timing overhead required to satisfy all file demands of requesting nodes in the network. \alert{Such type of model is of importance from an \emph{online cache update} perspective in which RNs refresh their cached contents while simultaneously satisfying the UEs file demands in collaboration with the donor eNB (DeNB).} In this work, we are interested in studying the fundamental \alert{delivery-time cache-memory tradeoff} of this particular network. % for arbitrary instances of $M$ and $K$.

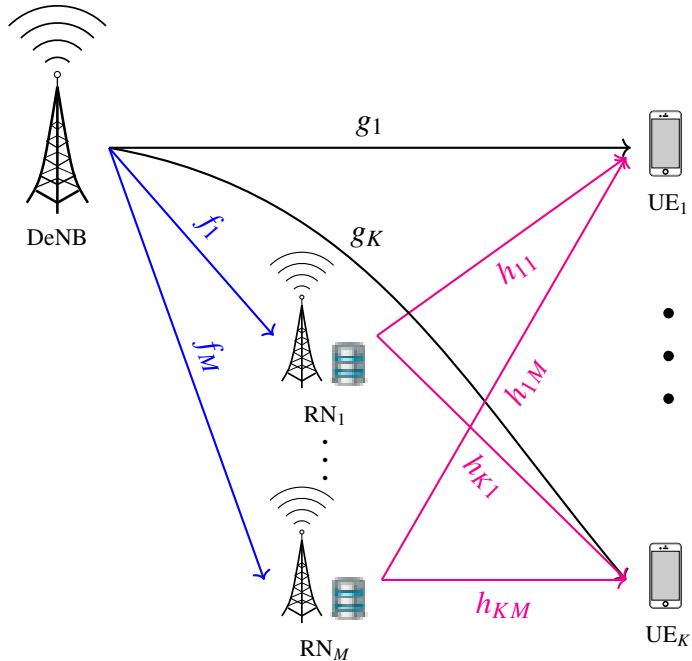
\begin{figure}[h]
	\begin{center}
	\vspace{1em}
	\begin{tikzpicture}[scale=0.8]
	\SymMod
	\end{tikzpicture}
	\vspace{-3.5em}
	\caption{\small A transceiver cache-aided HetNet consisting of one DeNB, $M$ RNs and $K$ UEs. \alert{These nodes are connected through the wireless links $f_i,g_k$ and $h_{ij}$, $i=1,\ldots,M$, $j=1,\ldots,K$. Each RN is equipped with a finite size cache.}}	
	\label{fig:HetNet}
	\end{center}
\end{figure}

In prior work, it was shown that both receiver (Rx) and transmitter (Tx) caching can offer significant latency reduction. Rx caching was first studied in \cite{Maddah-Ali2} for a shared \alert{error-free broadcast} channel with one server and multiple cache-enabled receivers. The authors show that their \emph{coded caching} approach of popular content exploits multicast opportunities and consequently reduces latency. \alert{Coded caching has received considerable attention for various related settings of \cite{Maddah-Ali2}. This includes the rate-memory tradeoff under uncoded cache placement \cite{Wan2016,YuISIT}, decentralized caching under homogeneous \cite{Maddah-Ali_decentral} and heterogeneous cache sizes \cite{Amiri_decentral}, caching with distinct file sizes \cite{Zhang_file_15}, online caching \cite{Pedarsani16}, caching with non-uniform demands \cite{Ji14} and multiple requests \cite{Wei2017} amongst others. Further, coded caching concepts have been applied to device-to-device networks with \cite{Awan15} and without secure delivery \cite{Ji16}, multi-server networks \cite{Shariatpanahi16} and hierarchical networks \cite{Karamchandani16}.}       

On the other hand, the impact of Tx caching on latency has mainly been investigated by analyzing the inverse degrees-of-freedom (DoF) metric of Gaussian interference networks \cite{Soheil}. To this end, the authors of \cite{Maddah_Ali} developed an interference alignment scheme characterizing the \alert{inverse DoF} as a function of the cache storage size for a 3-user Gaussian interference network. The caches are prefetched to allow transmitter cooperation so that interference coordination techniques are applicable. \alert{In \cite{avik}, the authors introduce the normalized delivery time (NDT) as a performance metric which is proportional to the inverse DoF. The first lower bounds on the NDT as a delivery time metric were derived in the same paper for a cache-aided interference channel with an arbitrary number of edge nodes and users.} With these bounds, the optimality of schemes presented in \cite{Maddah_Ali} for certain regimes of cache sizes was shown under uncoded prefetching of the cached content. \alert{These concepts have been recently applied to Fog radio access networks (F-RAN) that consist of a centralized cloud server, cache-assisted edge nodes and mobile users. The NDT of F-RANs has been first fully characterized for two edge nodes and two mobile users \cite{Tandon}. Later on for the setting of arbitrary number of edge nodes and receivers a constant factor characterization of $2$ has been established in \cite{Sengupta17_TIT}. The effect of channel strength and fading on the delivery time of partially connected F-RANs has been investigated in \cite{Azimi16,KakarArxiv,KakarICC} on the basis of the binary fading model \cite{Vahid14} and the linear deterministic model \cite{Avestimehr}.} 

\alert{Recently, the effect of Tx-Rx caching on the delivery time of interference networks was studied in two new lines of research. The \emph{first} being, where \emph{distinct} nodes, i.e., transmitters and receivers are equipped with caches. For this setting, the authors in \cite{Naderializadeh} establish one-shot linear delivery schemes, which avoid channel extension, and show their optimality within a factor of $2$ of the lower bounds on the delivery time. In \cite{Hachem16}, an architecture that separates physical and network layers is proposed and shown to be approximately DoF-optimal for Tx-Rx cache enabled interference networks. In \cite{Xu17}, achievability schemes on \emph{cooperative X-multicast channels} are used to show the multiplicative NDT-optimality of interference channels with caches at both ends of the network. The \emph{second} line of research focuses on the latency-influence of Tx-Rx caching at a \emph{single} node, i.e., \emph{transceiver caching}. Such type of caching is of importance from an \emph{online cache update} perspective. The first paper studying transceiver caching from a channel strength point of view is \cite{conference214}. However, in this paper the authors only characterize a single RN and UE setting. Our paper examines the second line of research but as opposed to \cite{conference214} with an arbitrary number of RNs and UEs. This examination is of interest to understand the compatibility of multicasting (observable in Rx-caching) and interference coordination techniques (observable in Tx-caching) in cache-aided \emph{transceiver} networks.} 

In this paper, we study the fundamental limits on the delivery time for a \emph{transceiver} cache-aided HetNet consisting of one donor eNB (DeNB), \alert{$M$ RNs and $K$ UEs}. \alert{The RNs are equipped with a cache memory of $\mu NL$ bits with $NL$ being the entire library size composed of $N$ files with file size $L$.} We measure the performance through \alert{the latency-centric metric \emph{normalized delivery time per bit} (NDT)} (cf. formal definition of NDT in Eq. \eqref{eq:NDT} in Section \ref{sec:Sym_Model}). This metric, first introduced in \cite{avik}, indicates the worst-case per-bit latency incurred in the wireless network with respect to a reference interference-free system without cache capacity restrictions %. Similarly to the DoF, it is a 
\alert{in the high signal-to-noise ratio (SNR) regime.} %metric. 
The main contributions of this paper are as follows: 
\begin{itemize}[leftmargin=*]
	\item In Section \ref{sec:lw_bd}, we develop a novel class of information theoretic lower bounds on the NDT under the assumption of perfect channel state information (CSI) and uncoded prefetching of the cached content. 
	\item \alert{We show that the optimal schemes for the \alert{extreme} cases of no caching and full caching are DeNB broadcasting and joint DeNB-RN zero-forcing beamforming, respectively.} 
	\item In Section \ref{cha_one_shot}, we propose a generalized \emph{one-shot} scheme that \alert{integrates multicasting schemes used in Rx caching with ZF beamforming} typically deployed in the context of Tx caching. Further, we evaluate its performance with respect to the lower bounds and show its effectiveness at higher cache capacities. \alert{Despite the low complexity of the one-shot scheme, we identify regimes as a function $(\mu,K,M)$ for which it is indeed NDT-optimal. Moreover,} we show that caching more than $\frac{\ceil{\nicefrac{(M-1)}{2}}}{M}$ \alert{fractions of a file} attains (at most) a multiplicative gap of $\frac{8}{3}$ with respect to the optimal NDT.
	\item \alert{In another scheme, we design precoders that synergistically interlace subspace interference alignment and zero-forcing. This design exploits spatially correlated file fractions through balancing zero-forcing opportunities and alignment opportunities. 
	%in terms of uncached and cached file fractions. 
	The scheme is NDT-optimal at lower cache sizes for both time-variant and invariant channels requiring \emph{finite} signal dimensions (e.g., time).}
	\item With the existence of both schemes, we are able to \emph{completely} characterize the \alert{latency-memory tradeoff} for the settings of (a) $M=1$ RNs and $K\in\{1,2,3\}$ UEs, (b) $M=2$ RNs and $K\in\{1,2\}$ UEs and (c) $M=3$ RNs and $K=1$ UEs. 
	%To this end, we establish NDT-optimal schemes that synergistically design precoders facilitating zero-forcing (ZF) beamforming, multicasting and interference alignment. 
	%Our schemes are optimal for both time-variant and invariant channels requiring \emph{finite} signal dimensions (time, frequency, etc.).
	\item Along with our results, we discuss the relationship between achievable (sum) DoF and NDT. To this end, we assess the results from both a rate (e.g., DoF), and latency (e.g., NDT) perspective. \alert{In particular, our optimal latency-memory tradeoff for $K+M\leq 4$ refutes the popular belief that increasing cache sizes translates to increasing the achievable sum DoF. In fact, there are cases where an increase in the cache size decreases the delivery time but also the achievable DoF.}      
\end{itemize}

%The rest of the paper is organized as follows.  In  Section \ref{sec:Sym_Model}, we  introduce the system model. The main results on the DTB, including achievability and converse, for serial and parallel cloud-edge transmissions are  presented in sections \ref{sec:ser_trans} through \ref{sec:ser_trans_ub} and \ref{sec:par_trans} through \ref{sec:par_upp_bound}, respectively. Finally, Section \ref{sec:conclusion} concludes the paper. The appendix of this paper is
%devoted to give further details on lower and upper bounds.

\textbf{Notation:} For any two integers $a$ and $b$ with $a\leq b$, we define $[a:b]\triangleq\{a,a+1,\ldots,b\}$ \alert{and we denote $[1,b]$ simply as $[b]$.} % When $a=1$, we simply write $[b]$ for $\{1,\ldots,b\}$. 
\alert{We use $\mathbf{a}_{t_1}^{t_2}$ and $\mathbf{A}_{t_1}^{t_2}$ with $t_1\leq t_2$ to refer to a vector $\mathbf{a}[t]$ and  a matrix $\mathbf{A}[t]$ concatenated across $t\in[t_1:t_2]$. When $t_1=1$, we simply write $\mathbf{a}^{t_2}$ and $\mathbf{A}^{t_2}$, respectively.} The superscript $(\cdot)^{\dagger}$ represents the transpose of a matrix. \alert{We use $\otimes$ to denote the Kronecker product.}%Furthermore, we define the function $(x)^{+}\triangleq\max\{0,x\}$ and the \emph{modified} \alert{modulo} operator $c=a\Mod b$ for integers $a$ and $b$ as $c=a$ if $a\leq b$ and $c=a\Modreg{b}$ if $a>b$. 

%The definition of the "modified modulo" operation should be reformulated so that it is unambiguous. For example, you could state that the notation "a MOD b" denotes that unique integer in [b] which is congruent to a modulo b, or else that the notation "a MOD b" denotes that unique integer c in [b] such that b divides (c-a). You should also remove the set notation (curly brackets) on b in the definition of the MOD notation.

\section{System Model and Latency Metric}
\label{sec:Sym_Model}

In  this  section, we first outline the system model of the cache-assisted \alert{broadcast-relay wireless network} illustrated in Fig. \ref{fig:HetNet}. Then, we introduce the normalized delivery time per bit (NDT) metric, along with its operational meaning to provide additional context on the adopted model and performance metric.

The network under study consists of $M$ causal full-duplex RNs and a donor eNB (DeNB) which serves $K$ UEs with its desired content over a shared wireless channel. Simultaneously, each RN also requests information from the DeNB. At every transmission interval, we assume that RNs and UEs request \alert{one file each} from the set $\mathcal{W}$ of $N$ popular files, whose elements are all of $L$ bits in size. The transmission interval terminates when the requested files have been delivered. The system model, notation and main assumptions for a \emph{single} transmission interval are summarized as follows:
\begin{itemize}[leftmargin=*]
	\item Let $\mathcal{W}=\{W_1,\ldots,W_{N}\}$ denote the library of popular files, where each file $W_n$ is of size $L$ bits. Each file $W_n$ is chosen uniformly at random from $[2^{L}]$\alert{, where $N=2^{L}$}. UEs and RNs request files $W_{d_u}$, $\forall u\in[K]$, and $W_{d_r}$, $\forall r\in[K+1:K+M]$, from the library $\mathcal{W}$, respectively. The demand vector $\mathbf{d}=(d_1,\ldots,d_{K+M})\in[N]^{K+M}$ denotes the request pattern of RNs and UEs. %Hereby, the first $K$ elements represent the demand of UEs and the remaining $M$ elements the demand of RNs. 
	\alert{This vector is shared among all nodes.}  
	\item The RNs are endowed with a cache capable of storing $\mu NL$ bits, where $\mu\in[0,1]$ corresponds to the \emph{fractional cache size}. It denotes how much
	content can be stored at each RN \alert{relative to the size of the entire library $\mathcal{W}$}. 
	\item The DeNB has access to all $N$ popular files of $\mathcal{W}$.  
	\item Global CSI for \alert{the single-antenna setting} at time instant $t$ is summarized by the channel vectors $\mathbf{f}[t]=\{f_{m}[t]\}_{m=1}^{m=M}\in\mathbb{C}^{M}$ and $\mathbf{g}[t]=\{g_{k}[t]\}_{k=1}^{k=K}\in\mathbb{C}^{K}$ and the channel matrix $\mathbf{H}[t]=\{h_{km}[t]\}_{k=1,m=1}^{k=K,m=M}\in\mathbb{C}^{K\times M}$. Here, $f_m$ and $g_k$ represent the complex channel coefficients from DeNB to RN$_m$ and UE$_k$, respectively, while $h_{km}$ is the channel from RN$_m$ to UE$_k$. We assume that all channel coefficients are drawn i.i.d. from a continuous \alert{random} distribution.
\end{itemize}

Communication over the wireless channel occurs in two consecutive phases, (a) \emph{placement phase} followed by (b) \emph{delivery phase}. These are detailed next, along with the key performance metric termed as \emph{normalized delivery time per bit} (NDT). 
\vspace{.5em}

\paragraph{Placement phase} 
%During this phase, every RN is given full access to the database of $N$ files. The cached content at RN$_m$ is generated through its individual caching function. 
\alert{During this phase, each RN caches content from the library $\mathcal{W}$ by requesting this content from the DeNB using the caching function defined next.}

\begin{definition}(Caching function)\label{def_cache_fct} 
	RN$_m$, $\forall m=1,\ldots,M$, maps each file $W_n\in\mathcal{W}$ to its local \emph{file cache content} \alert{as} 
	\begin{equation}
	S_{m,n}=\phi_{m,n}(W_{n}),\qquad\forall n=1,\ldots,N\nonumber,
	\end{equation} where $\phi_{m,n}(\cdot)$ is the caching function. 
	All $S_{m,n}$ are concatenated to form the total cache content 
	\begin{equation}
	S_m=(S_{m,1},S_{m,2},\ldots,S_{m,N})\nonumber
	\end{equation} 
	at RN$_m$.
\end{definition}
\vspace{.5em}	
\alert{Hereby, we assume symmetry in caching, i.e., each file $W_n,\forall n\in[N],$ is cached with at most $\mu L$ number of bits.\footnote{For instance, in cases, where the files are requested in a non-uniform fashion \cite{Niesen17}, asymmetric caching across files is of relevance.} In consequence,} the entropy $H(S_{m,n})$ of each component $S_{m,n}$, $n=1,\ldots,N$, is upper bounded by $\nicefrac{\mu NL}{N}=\mu L$. The definition of the caching function presumes that every file $W_i$ is subjected to individual caching functions. Thus, permissible caching policies allow for intra-file coding but avoid coding across files known as inter-file coding. Moreover, the caching policy is typically kept fixed over long transmission intervals. Thus, it is indifferent to the UEs request pattern and of channel realizations.     
\vspace{.5em}	
\paragraph{Delivery phase} 
In this phase, a transmission policy at DeNB and all RNs is applied to satisfy the given requests $\mathbf{d}$ under the current channel realizations $\mathbf{f},\mathbf{g}$ and $\mathbf{H}$. Throughout the remaining definitions, we denote the number of channel uses required to satisfy all file demands by $T$. \alert{This time depends on the demand vector $\mathbf{d}$ and the channel realizations $\mathbf{f},\mathbf{g}$ and $\mathbf{H}$, i.e., $T=T(\mathbf{d},\mathbf{f},\mathbf{g},\mathbf{H})$.} \alert{In the sequel, we exploit the lowercase subscripts $s,r$ and $u$ for notations concerning DeNB, RNs and UEs, respectively.} 
\vspace{.5em}
\begin{definition}(Encoding functions)\label{def_enc_fct} The DeNB encoding function at time instant $t\in[T]$
	\begin{equation}
	\psi_{s}^{[t]}:[2^{NL}]\times [N]^{M+K}\times\mathbb{C}^{Mt}\times\mathbb{C}^{Kt}\times\mathbb{C}^{Kt\times M}\rightarrow \mathbb{C}\nonumber
	\end{equation} 
	determines the DeNBs transmission signal %$x_{s}[t]=\psi_{s}^{[t]}(\mathcal{W},\mathbf{d},\mathbf{f}_{t=1}^{t},\mathbf{g}_{t=1}^{t},\mathbf{H}_{t=1}^{t})$ 
	\alert{$x_{s}[t]=\psi_{s}^{[t]}(\mathcal{W},\mathbf{d},\mathbf{f}^{t},\mathbf{g}^{t},\mathbf{H}^{t})$} subjected to an average power constraint of $P$. The encoding function of the causal \emph{full-duplex} RN$_m$ at time instant $t\in[T]$ is defined by
	\begin{align}
	\psi_{r,m}^{[t]}:[2^{\mu NL}]\times \mathbb{C}^{t-1}\times [N]^{M+K}\times\mathbb{C}^{Mt}\times\mathbb{C}^{Kt}\times\mathbb{C}^{Kt\times M}\rightarrow \mathbb{C},\nonumber
	\end{align} 
	which determines the \alert{symbol $x_{r,m}[t]=\psi_{r,m}^{[t]}(S_m,\mathbf{y}_{r,m}^{t-1},\mathbf{d},\mathbf{f}^{t},\mathbf{g}^{t},\mathbf{H}^{t})$} while satisfying the average power constraint given by the parameter $P$.
\end{definition}
\vspace{.5em}
Hereby, the \alert{symbols} $x_{s}[t]$ and $x_{r,m}[t]$ are transmitted over $t\in[T]$ channel uses. For any time instant $t$, $\psi_{r,m}^{[t]}$ accounts for the simultaneous reception and transmission through incoming and outgoing wireless links at RN$_m$. To be specific, at the $t$--th channel use the encoding function $\psi_{r,m}^{[t]}$ maps the cached content $S_m$, the received signal $\mathbf{y}_{r,m}^{t-1}$ (see Eq. \eqref{eq:Gaus_mod_RN}), the demand vector $\mathbf{d}$ and global CSI to the \alert{symbol} $x_{r,m}[t]$.

After transmission, the received signals at UE$_k$ is given by
\begin{equation}\label{eq:Gaus_mod}
y_{u,k}[t]=
g_{k}[t]x_{s}[t]+\sum_{m=1}^{M}h_{km}[t]x_{r,m}[t]+z_{u,k}[t],\forall t\in[T],
\end{equation}
where $z_{u,k}[t]$ denotes complex i.i.d. Gaussian noise of zero mean and unit power. The received signal at RN$_{m}$ is given by
\begin{equation}\label{eq:Gaus_mod_RN}
y_{r,m}[t]=
f_{m}[t]x_{s}[t]+z_{r,m}[t],\forall t\in[T],
\end{equation}	
where $z_{r,m}[t]$ is additive zero mean, unit-power i.i.d. Gaussian noise. The desired files are decoded using the following functions.
\vspace{.5em}
\begin{definition}(Decoding functions)\label{def_dec_fct} 
	The decoding operation at UE$_k$ follows the mapping
	\begin{equation}
	\eta_{u,k}:\mathbb{C}^{T}\times [N]^{K+M}\times\mathbb{C}^{MT}\times\mathbb{C}^{KT}\times\mathbb{C}^{KT\times M}\rightarrow [2^{L}]\nonumber 
	\end{equation} 
	to provide an estimate \alert{$\hat{W}_{d_k}=\eta_{u,k}(\mathbf{y}_{u,k}^{T},\mathbf{d},\mathbf{f}^{T},\mathbf{g}^{T},\mathbf{H}^{T})$} of the requested file $W_{d_k}$.
	%The decoding function $\eta_{u,k}$ takes as its arguments $\mathbf{f}[t]$, $\mathbf{g}[t]$, $\mathbf{H}[t]$, $\forall t\in[T]$, the available demand pattern $\mathbf{d}$ and the channel outputs $\mathbf{y}_{u,k}^{T}$
	%=\eta_{u,k}\big(\mathbf{y}_{u,k}^{T},\mathbf{d},\mathbf{h}\big)$ 
	In contrast to decoding at UE$_k$, all RNs explicitly leverage their cached content according to
	\begin{align}
	\eta_{r,m}:\mathbb{C}^{T}\times [2^{\mu NL}]\times [N]^{K+M}\times\mathbb{C}^{MT}\times\mathbb{C}^{KT}\times\mathbb{C}^{KT\times M}\rightarrow [2^{L}]\nonumber 
	\end{align} 
	to generate \alert{$\hat{W}_{d_r}=\eta_{r,m}(\mathbf{y}_{r,m}^{T},S_m,\mathbf{d},\mathbf{f}^{T},\mathbf{g}^{T},\mathbf{H}^{T})$} as an estimate of the requested file $W_{d_r}$.
\end{definition}

\alert{The reliability measure of a jointly proposed placement and delivery scheme is governed by its worst-case error probability defined as
\begin{equation}\label{eq:error_prob}
P_e=\max_{\mathbf{d}\in [N]^{K+M}}\max_{j\in[K+M]}\mathbb{P}(\hat{W}_{d_j}\neq W_{d_j})
\end{equation} which is taken over error probabilities of $M$ RNs and $K$ UEs for all possible demands.} A proper choice of caching, encoding and decoding functions that satisfy the reliability condition; that is, the worst-case error probability $P_e$  
%\begin{equation}\label{eq:error_prob}
%P_e=\max_{\mathbf{d}\in [N]^{K+M}}\max_{j\in[K+M]}\mathbb{P}(\hat{W}_{d_j}\neq W_{d_j})
%\end{equation} 
approaches $0$ as \alert{$T\rightarrow\infty$}, is called a \emph{feasible policy}. \alert{For strictly positive rates $T\rightarrow\infty$ is congruent with $L\rightarrow\infty$}.\footnote{\alert{This is due to the fact that $N=2^{L}$ files are chosen uniformly at random from the index set $[2^{TR}]$ with $R$ being the rate. In consequence, $L=TR$ which shows the equivalence of $T\rightarrow\infty$ and $L\rightarrow\infty$.}} Now we are ready to define the delivery time per bit and its normalized version.
\vspace{.5em}
\begin{definition}(Delivery time per bit \cite{avik}) 
	The delivery time per bit (DTB) %for a given request pattern $\mathbf{d}$ and channel realization $\mathbf{f},\mathbf{g}$ and $\mathbf{H}$ 
	is defined as 
	\begin{equation}\label{eq:DTB}
	\Delta(\mu,P)=\max_{\mathbf{d}\in [N]^{K+M}}\limsup_{L\rightarrow\infty}\frac{\mathbb{E}[T(\mathbf{d},\mathbf{f},\mathbf{g},\mathbf{H})]}{L},
	\end{equation} 
	where the expectation is over the channel realizations $\mathbf{f},\mathbf{g}$ and $\mathbf{H}$.
\end{definition}  
\vspace{.5em}

In the definition above, $T$ represents the \alert{delivery} time \cite{Liu2011}. The normalization of the expected delivery time by the file size $L$ gives insight about the \alert{per-bit latency}. In this context, the DTB measures the per-bit latency, i.e., the latency incurred per-bit when transmitting the requested files through the wireless channel, within a single transmission interval for the \emph{worst-case} request pattern of RNs and UEs as $L\rightarrow\infty$. The DTB depends on the fractional cache size $\mu$ and the power level $P$.   

In analogy to the degrees-of-freedom metric \cite{Etkin08}, %,KakarMDPI}, 
the normalized delivery time per bit (NDT) is a high-SNR metric that relates the DTB to that of a point-to-point reference system. 
\vspace{.5em}
\begin{definition} (Normalized delivery time \cite{avik}) 
	The NDT is defined as 
	\begin{equation}\label{eq:NDT}
	\delta(\mu)=\lim_{P\rightarrow\infty}\frac{\Delta(\mu,P)}{1/\log(P)}.
	\end{equation} 
	The minimum NDT $\delta^{\star}(\mu)$ is \alert{the infimum of $\delta(\mu)$ over all feasible policies}.
\end{definition}  
\vspace{.5em}
The NDT compares the \emph{delivery time per bit} achieved by the feasible coding scheme for the worst-case demand scenario to that of a baseline interference-free system in the high SNR regime. The \alert{feasible} scheme, on the one hand, allows for reliable transmission of one file of $L$ bits to each Rx on average in $\mathbb{E}[T(\mathbf{f},\mathbf{g},\mathbf{H})]$ channel uses, i.e., $1$ bit in $\mathbb{E}[T(\mathbf{f},\mathbf{g},\mathbf{H})]/L$ channel uses. The baseline system (e.g., a point-to-point channel), on the other hand, can transmit $\log(P)$ bits to a single Rx in one channel use, i.e., $1$ bit in $1/\log(P)$ channel uses \alert{in the worst case}. Therefore, the resulting NDT $\delta(\mu)$ indicates that the worst-case delivery time for one bit of the cache-aided network at fractional cache size $\mu$ is $\delta(\mu)$ times larger \alert{than} the time needed by the baseline system.   

From \cite[Lemma 1]{KakarArxiv}, it readily follows that the NDT is a convex function in $\mu$. This means that a cache-aided network shown in Fig. \ref{fig:HetNet} operating at fractional cache size $\mu=\alpha\mu_1+(1-\alpha)\mu_2$ for any $\alpha\in [0,1]$ achieves \alert{less (or equal) NDT than the \emph{convex combination} $\alpha\delta(\mu_1)+(1-\alpha)\delta(\mu_2)$} through applying known feasible schemes applicable at fractional cache sizes $\mu_1$ and $\mu_2$ on distinct $\alpha$ and $1-\alpha$-fractions of the files, respectively. This strategy is known as \emph{memory sharing}.
\section{Main Results}
\label{sec:main_res}

In this section, we state our main results on the minimum NDT for the cache-enabled \alert{broadcast-relay} wireless network of Fig. \ref{fig:HetNet} for $M$ RNs and $K$ UEs. Hereby, our main results are presented in Theorems \ref{theorem_lower_bound}--\ref{th:opt_NDT_tradeoff}. \alert{They include, respectively, a novel lower bound, an upper bound (achievability) on the NDT and a complete NDT-tradeoff characterization for $K+M\leq 4$.} Further, we formulate multiple corollaries that evaluate the performance of the scheme presented in Theorem \ref{th:one_shot_ach_NDT} with respect to the lower bound of Theorem \ref{theorem_lower_bound} \alert{in terms of a multiplicative gap given by
\begin{equation*}
\frac{\delta_\text{ach}(\mu)}{\delta_{\text{LB}}(\mu)}.
\end{equation*} Hereby, $\delta_\text{ach}(\mu)$ and $\delta_{\text{LB}}(\mu)$ denote, respectively, an upper and lower bound on the NDT.}     

\begin{theorem}[Lower bound on NDT]\label{theorem_lower_bound}
	For the transceiver cache-aided network with one DeNB, $M$ RNs each endowed with a cache of fractional cache size $\mu\in[0,1]$, $K$ UEs and a file library of $N\geq K+M$ files, the optimal NDT is lower bounded under perfect CSI at all nodes by
	\begin{align}\label{eq:NDT_lw_bound}
	\delta^{\star}&(\mu)\geq\max\Big\{1,\max_{\substack{\ell\in[\bar{s}:M],\\s\in[\min\{M+1,K\}]}}\delta_{\text{LB}}(\mu,\ell,s)\Big\},
	\end{align} where $\bar{s}=M+1-s$ and 
	\begin{align}\label{eq:NDT_lw_bound_inner_comp}
	&\hspace{-.25cm}\delta_{\text{LB}}(\mu,\ell,s)=\frac{K+\ell-\mu(\bar{s}\big(K-s+\frac{(\bar{s}-1)}{2}\big)+\frac{\ell}{2}(\ell+1))}{s}.
	\end{align}
\end{theorem}
\begin{proof}
	The proof of Theorem \ref{theorem_lower_bound} will be given in Section \ref{sec:lw_bd}. \alert{To provide some insight into the lower bound presented in Theorem \ref{theorem_lower_bound}, however, we outline a short sketch of the proof. Particularly, we summarize the ideas when deriving the two terms in Eq. \eqref{eq:NDT_lw_bound}.} 
		
	\alert{First, we find the bound $\delta_{\text{LB}}(\mu,\ell,s)$ by exploiting the following main observation in the high SNR regime (where noise becomes negligible). That is, given the channel outputs of any $s$ UEs (e.g., of UE$_1$, UE$_2$, $\ldots$, UE$_s$ denoted by $\mathbf{y}_{u,[1:s]}^{T}$),} in addition to the cached content of $\ell$ RNs \alert{(e.g, cached contents of RN$_1$, RN$_2$, $\ldots$, RN$_\ell$ represented by $S_{[1:\ell]}$)} such that $s+\ell\geq M+1$ enables the decoding of all $K$ files requested by the UEs as well as $\ell$ files desired by the RNs. This is due to the fact that with this information set, all $M+1$ transmit signals \alert{consisting of} the DeNB signal $x_s$ and the RNs transmit signals $x_{r,m},\forall m\in[M]$, can be reproduced. This in turn, allows the reconstruction of the following channel outputs: On the one hand, the remaining $K-s$ channel outputs of the UEs and on the other hand $\ell$ outputs of the RNs. With the availability of $K$ UE channel ouputs as well $\ell$ RN channel outputs and cached contents, $K+\ell$ files in total become decodable. 
	
	\alert{Second, the unity lower bound follows from the fact that the NDT is bounded from below by the performance of the reference interference-free system with an NDT of $1$}. The maximum over these two lower bounds concludes the proof of Theorem \ref{theorem_lower_bound}.                 
\end{proof}

Before establishing the achievability at fractional cache sizes in the range $\mu\in(0,1)$, we consider two special corner points at fractional cache sizes $\mu=0$ and $\mu=1$ for \emph{arbitrary} $M$ and $K$. These are the cases where the RN has either \emph{zero-cache} ($\mu=0$) or \emph{full-cache} ($\mu=1$) capabilities. In the following lemma, we expound the \emph{optimal NDT} for these two points.  
\vspace{.5em}
\begin{lemma}\label{corr_mu_0_and_1}
	For the transceiver cache-aided network with one DeNB, $M$ RNs each endowed with a cache of fractional cache size $\mu$, $K$ UEs and a file library of $N\geq M+K$ files, the optimal NDT is
	\begin{equation}\label{eq:opt_NDT_mu_0}
	\delta^{\star}(\mu)=K+M\:\:\text{ for }\mu=0,
	\end{equation} 
	achievable via DeNB broadcasting to $M$ RNs and $K$ UEs, and 
	\begin{equation}
	\label{eq:opt_NDT_mu_1}
	\delta^{\star}(\mu)=\max\Bigg\{1,\frac{K}{M+1}\Bigg\}\:\:\text{ for }\mu=1,
	\end{equation} 
	achievable via zero-forcing beamforming for an $(M+1,K)$ MISO\footnote{In MISO broadcast channels, we use the notation, $(a,b)$ for integers $a$ and $b$ to denote a broadcast channel with $a$ transmit antennas and $b$ single antenna receivers.} broadcast channel. 
\end{lemma}
\begin{proof}
	For the proof, it suffices to find a cache transmission policy that matches the lower bound in Thoerem \ref{theorem_lower_bound} for $\mu=0$ and $\mu=1$, respectively. On the one hand, if $\mu=0$, we note that $\delta_{\text{LB}}(0,M,1)=K+M$. On the other hand, if $\mu=1$, we observe that $\delta_{\text{LB}}(1,0,M+1)=\nicefrac{K}{(M+1)}$ if $M+1\leq K$ and $\delta_{\text{LB}}(1,\ell,s)<1$ if $M+1>K$. Next, we consider the achievability at $\mu=0$ and $\mu=1$. For these two fractional cache sizes, the network in Fig. \ref{fig:HetNet} reduces to a SISO broadcast channel (BC) with $K+M$ users for $\mu=0$ and an $(M+1,K)$ MISO broadcast channel for $\mu=1$. The approximate \emph{per-user} rate (neglecting $o(\log(P))$ bits) for these two channels are known to be $\frac{1}{(K+M)}\log(P)$ \cite{Jafar07} (achievable through unicasting each user's message) \alert{for $\mu=0$} and $\frac{1}{K}\min\{M+1,K\}\log(P)$ (achievable through zero-forcing beamfoming) \cite{Weingarten06} \alert{for $\mu=1$}, respectively. Equivalently, each user needs the reciprocal per-user rate of signaling dimensions (e.g., channel uses in time) %or frequency) 
	to retrieve one desired bit \alert{reliably}. Thus, the approximate DTB becomes, respectively, $\frac{(K+M)}{\log(P)}$ and $\frac{K}{\min\{M+1,K\}\log(P)}$. Normalizing the delivery time per bit by the point-to-point reference DTB $\frac{1}{\log(P)}$ generates the NDTs $K+M$ and $\max\{1,\nicefrac{K}{(M+1)}\}$. This establishes the NDT-optimality at these fractional cache sizes.
\end{proof}
%\vspace{.5em}
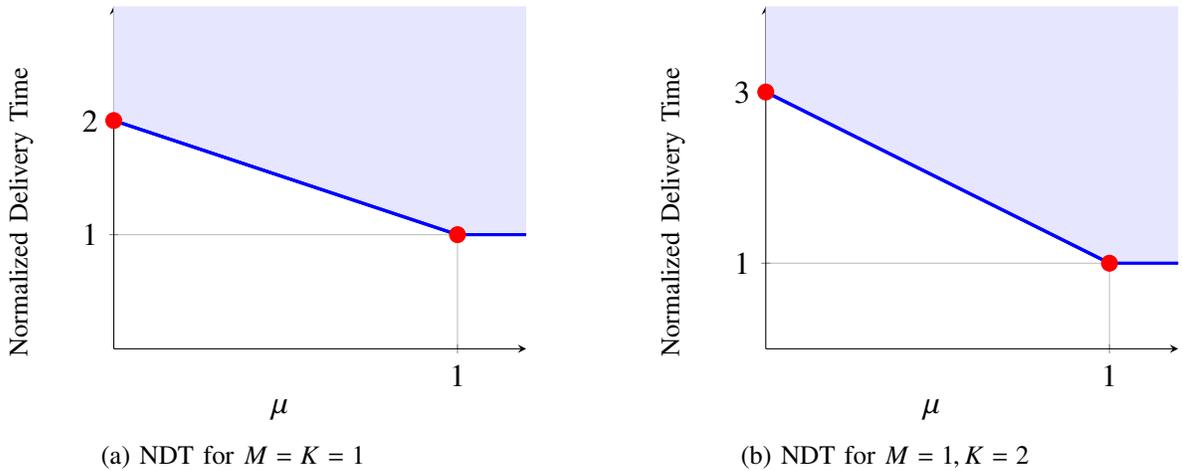
\begin{figure*}
	\centering
	\begin{subfigure}[b]{0.475\textwidth}
		\centering
		\begin{tikzpicture}[scale=0.8]
		\PlotNDTMOneKOne
		\end{tikzpicture}
		\caption[NDT for $M=K=1$]{\small NDT for $M=K=1$}
		\label{fig:NDT_M_K_1_1}
	\end{subfigure}
	\hfill
	\begin{subfigure}[b]{0.475\textwidth}  
		\centering 
		\begin{tikzpicture}[scale=0.8]
		\PlotNDTMOneKTwo
		\end{tikzpicture}
		\caption[NDT for $M=1,K=2$]{\small NDT for $M=1, K=2$} 
		\label{fig:NDT_M_K_1_2}       
	\end{subfigure}
	\caption[NDT as a function of $\mu$ for $M=1$ and $K\leq 2$]
	{\small Optimal NDT as a function of $\mu$ for $M=1$ and $K\leq 2$} 
	\label{fig:NDT_M_K_1_LEQ_2}
\end{figure*} 
\begin{remark}
	From Lemma \ref{corr_mu_0_and_1}, we infer that the caching problem for the system illustrated in Fig. \ref{fig:HetNet} establishes the behavior of the network in terms of delivery time between the two extremes -- SISO BC with $K+M$ users and an $(M+1,K)$ MISO BC. This analysis will reveal what kind of schemes other than simple unicasting and zero-forcing will be optimal for $0<\mu<1$.
\end{remark}
\begin{remark}[\alert{Optimality of Memory Sharing of Zero-cache and Full-cache Schemes}]\label{remark:memory_sharing_opt}\alert{Memory sharing of optimal schemes at extreme points of zero-cache ($\mu=0$) and full-cache ($\mu=1$) may be optimal. Or in other words, the optimal NDT of intermediate points at fractional cache sizes $0<\mu<1$ may be achievable through successively time-sharing between unicasting and zero-forcing on $(1-\mu)$ and $\mu$ fractions of the files, respectively. This implies that treating uncached and cached file fractions \emph{independently} by applying two \emph{separate} delivery schemes -- unicasting and zero-forcing -- can be delivery time optimal. However, as shown in Fig. \ref{fig:NDT_M_K_1_LEQ_2}, this only happens for cases where $M=1$ and $K\leq 2$. In these cases, the lower bound $\delta_{\text{LB}}(\mu,1,1)$ coincides with the achievable NDT.} 
\end{remark}

\alert{Intuitively, the observation of Remark \ref{remark:memory_sharing_opt} makes sense because of the following two reasons. On the one hand, when $M=1$, there are (with respect to $(1-\mu)$ fractions of uncached information on each file) (i) no multicasting opportunities on the DeNB-RN broadcast channel and (ii) no zero-forcing opportunities on the RN-UE channel. On the other hand, when $K\leq 2$ joint DeNB-RN zero-forcing beamforming guarantees (with respect to the cached fractions of the files all $K$ UEs desire) that every UE receives in \emph{each} channel use desired information on its requested file since $K\leq M+1=2$. Thus, in conclusion, for these cases applying unicasting and joint DeNB-RN beamforming successively on $(1-\mu)$ and $\mu$ file fractions is NDT-optimal. However, this is in general not true for arbitrary instances of $K$ and $M$ as we shall see next.}            

\alert{So far we have discussed cases where successively applying unicasting and joint DeNB-RN beamforming on uncached and cached file fractions is NDT-optimal at fractional cache sizes $0<\mu<1$. However, in general, (with the exception of $M=1, K\leq 2$) delivery schemes that treat uncached and cached file fractions independently are \emph{suboptimal} for arbitrary $K$ and $M$. To this end, we propose a general \emph{one-shot} (OS) scheme that treats uncached and cached file fractions \emph{jointly} by exploiting multicasting and zero-forcing opportunities. Our proposed one-shot scheme is with the exception of the full-duplex requirement at the RNs simple in implementation because all receiving nodes are able to decode their \emph{desired} symbols on a \emph{single} channel use basis. In other words, these schemes explicitly preclude symbol decoding over multiple channel uses. The next theorem specifies the achievable NDT for these one-shot schemes.}       
\vspace{.5em}
\begin{theorem}[Achievable One-Shot NDT]\label{th:one_shot_ach_NDT}
	For $N\geq K+M$ files, $K$ UEs and $M$ RNs each with a cache of (fractional) size $\mu\in\{0,\nicefrac{1}{M},\nicefrac{2}{M},\nicefrac{3}{M},\ldots,\nicefrac{(M-1)}{M},1\}$,
	%\begin{align}\label{eq:NDT_OS_ach}
	%\delta^{\star}(\mu)\leq\delta_{\text{OS}}(\mu)\triangleq\max\Bigg\{\delta_{\text{MAN}}(\mu),\frac{K+\delta_{\text{MAN}}(\mu)\boldsymbol{1}_{K>\mu M}}{\min\{K,1+\mu M\}}\Bigg\}
	%\end{align} 
	\begin{align}\label{eq:NDT_OS_ach}
	\delta_{\text{OS}}(\mu)\triangleq\max\Bigg\{\delta_{\text{MAN}}(\mu),\frac{K+\delta_{\text{MAN}}(\mu)\boldsymbol{1}_{K>\mu M}}{\min\{K,1+\mu M\}}\Bigg\}
	\end{align} is achievable, where $\boldsymbol{1}_{K>\mu M}$ is the indicator function and $\delta_{\text{MAN}}(\mu)$ is the achievable Maddah-Ali Niesen (MAN) NDT given by
	\begin{equation*}
	\delta_{\text{MAN}}(\mu)=M\cdot(1-\mu)\cdot\frac{1}{1+\mu M} 
	\end{equation*} such that $\delta^{\star}(\mu)\leq\delta_{\text{OS}}(\mu)$. For arbitrary $\mu\in[0,1]$, the lower convex envelope of these points is achievable.  
\end{theorem}
\begin{figure*}
	\centering
	\begin{subfigure}[b]{0.475\textwidth}
		\centering
		\vspace{10.5em}
		\hspace{-2em}
		\begin{tikzpicture}[scale=1]
		\SymModSchemePhaseOne
		\end{tikzpicture}
		\vspace{-14em}
		\caption{\small Scheme for the first phase ($t\in[T_1]$) with per-phase NDT $\delta_1(\mu)=\delta_{\text{MAN}}(\mu)$.} 
		\label{fig:OS_Phase1}
	\end{subfigure}
	\hfill
	\begin{subfigure}[b]{0.475\textwidth}  
		\centering 
		\vspace{10.5em}
		\hspace{-2em}
		\begin{tikzpicture}[scale=1]
		\SymModSchemePhaseTwo
		\end{tikzpicture}
		\vspace{-14em}
		\caption{\small Scheme for the second phase ($t\in[T_1+1:T_1+T_2]$) with per-phase NDT $\delta_2(\mu)=\frac{1}{\min\{K,1+\mu M\}}$.}  
		\label{fig:OS_Phase2}       
	\end{subfigure}
	\caption
	{\small Illustration of the proposed one-shot scheme with $M=4$ RNs, $K=2$ UEs and $\mu M=2$ for the worst-case demand scenario. On the one hand, in each channel use of the first phase [cf. (a)], the MAN scheme is used on the SISO DeNB-RN broadcast channel to convey desired symbols of (any combination of) $1+\mu M$ RNs. In the worst-case scenario, where UEs request other files, these symbols represent interference which are zero-forced through cooperative DeNB-RN interference cancelation at (any combination) of $\min\{K,\mu M\}$ UEs. Simultaneously, the scheme exploits RN caches by providing the same UEs with their desired symbols. After $T_1$ channel uses, of the first phase, the demand of the RNs is satisfied. On the other hand, the second phase [cf. (b)] is devoted to communicate, if necessary, the remaining file symbols of the UEs by applying cooperative DeNB-RN zero-forcing beamforming.} 
	\label{fig:OS_Scheme}
\end{figure*}
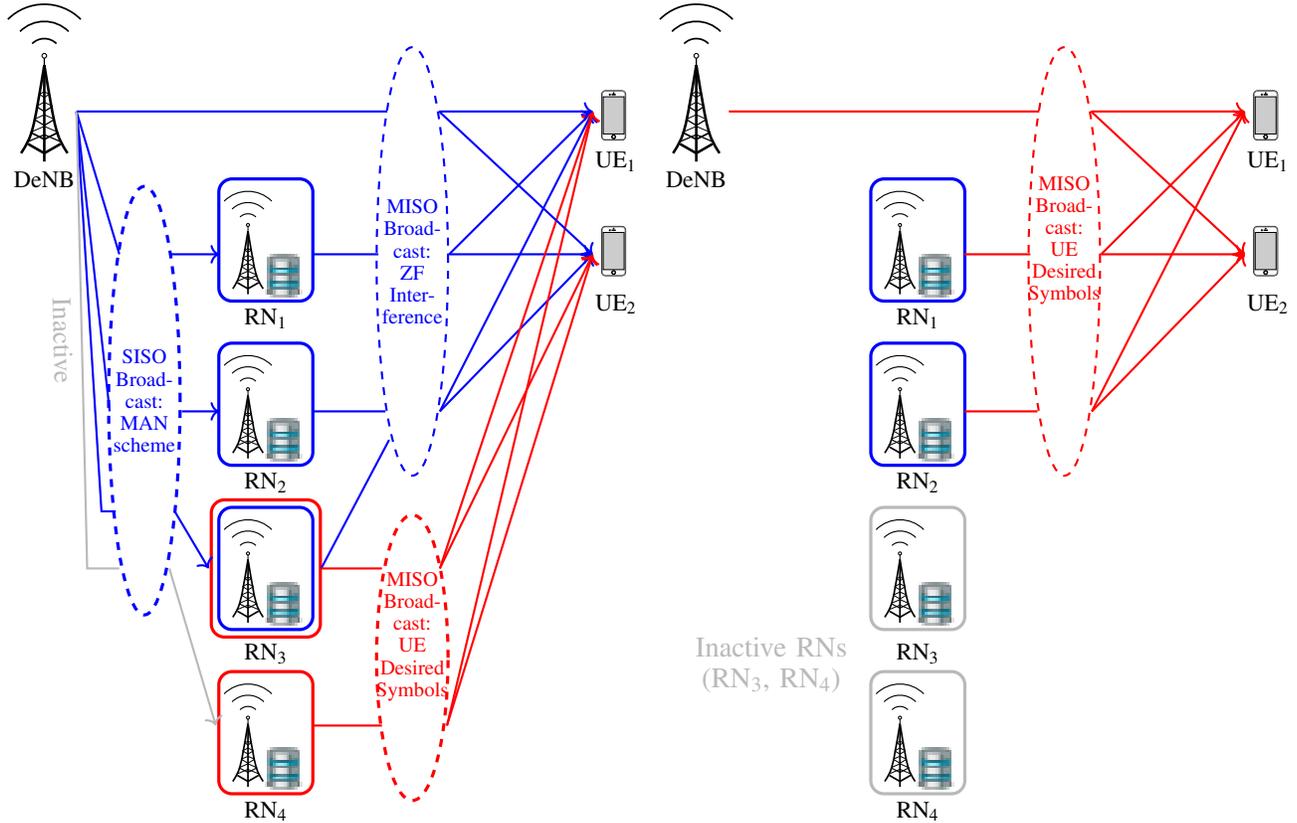 
\begin{proof}
	Details on the scheme are presented in Section \ref{cha_one_shot}. Nevertheless, we use this paragraph to highlight the main idea of the scheme. \alert{Recall that the file length is denoted by $L$. Assuming Gaussian signaling, the file is comprised of $L'=L/\log(P)$ symbols.} \alert{The scheme we develop (potentially) consists of two phases requiring $T_1$ and $T_2$ channel uses, respectively, to send uncached $(1-\mu)L'$ Gaussian symbols (each symbol carrying approximately $\log(P)$ bits) to each RN and also $L'$ symbols to each UE}. 
	
	In every channel use of the \emph{first phase} \alert{depicted in Fig. \ref{fig:OS_Phase1}, beamforming facilitates the integration of the MAN scheme \cite{Maddah-Ali2} with zero-forcing beamforming to (i) pre-cancel interference caused by applying the MAN scheme on the DeNB-RN broadcast channel at the UEs and (ii) convey desired symbols to UEs. Precisely, the MAN scheme is applied on the DeNB-RN broadcast channel} to provide each RN in a subset $\mathcal{S}_R\subset[M]$ \alert{with} $|\mathcal{S}_R|=1+\mu M$ RNs with a desired symbol. \emph{Simultaneously}, the full-duplex capabilities at the RNs are exploited by conveying \alert{to} each UE in the subset $\mathcal{S}_U\subset[K]$ \alert{with} $|\mathcal{S}_U|=\min\{K,\mu M\}$ UEs in total with a desired symbol by zero-forcing the contribution of all interfering symbols that $|\mathcal{S}_R|=1+\mu M$ RNs in $\mathcal{S}_R$ desire. \alert{Recall that the first phase consumes $T_1$ channel uses. We show that $T_1$ channel uses suffice in sending each RN$_m$, $\forall m\in[M]$, the remaining $(1-\mu)L'$ symbols of its requested file. Simultaneously, in $T_1$ channel uses each UE$_k$, $\forall k\in[K]$, receives $\tilde{L}$ symbols of its desired file, with $\tilde{L}$ being proportional to $|\mathcal{S}_U|=\min\{K,\mu M\}$. Thus, we may encounter cases where it is either feasible or infeasible to communicate all $L'$ symbols of each requested file to the respective UEs in $T_1$ channel uses ($\tilde{L}\geq L'$ or $\tilde{L}<L'$).}    
	
	\alert{Only in the case of missing symbols ($\tilde{L}<L'$) that all $K$ UEs still require after $T_1$ channel uses, additional $T_2>0$ channel uses are required in phase two to deliver the remaining desired symbols as shown in Fig. \ref{fig:OS_Phase2}.} To this end, in every channel use \emph{cooperative DeNB-RN zero-forcing beamforming} is deployed to send one symbol in total to $\psi'=\min\{K,1+\mu M\}$ UEs. The decoding at the RNs and UEs does not involve symbol decoding over multiple channel uses. Instead, decoding occurs on a one-shot, or single channel use, basis. In conclusion, the achievable NDT becomes either $\frac{T_1}{L'}$ if \alert{$T_2=0$} or $\frac{T_1+T_2}{L'}$ if \alert{$T_2>0$} with $L'$ being the number of symbols per file.     
\end{proof}
%\begin{table}
\begin{center}
	\begin{tabular}{ |l|l|c|c|l|}
		\hline
		\multirow{2}{*}{Region Name} & \multirow{2}{*}{Definition} & \multicolumn{2}{c|}{Channel limitation} & \multirow{2}{*}{Achievable NDT} \\\cline{3-4} & & RN side & UE side &  \\ \hline  \multirow{2}{*}{Region A} & $K\leq\mu M<M<\frac{1}{1-2\mu},\mu\leq\frac{1}{2}$ & \multirow{2}{*}{--} & \multirow{2}{*}{\checkmark} & \multirow{2}{*}{$\delta_{\text{OS}}^{(\text{A})}(\mu)=1$} \\\cline{2-2}  & $K\leq\mu M\leq M,M>\frac{1}{1-2\mu},\mu>\frac{1}{2}$ & & & \\ \hline
		Region B & $K\leq\mu M,\frac{1}{1-2\mu}\leq M,\mu\leq\frac{1}{2}$ & \multirow{2}{*}{\checkmark} & \multirow{2}{*}{--} & \multirow{2}{*}{$\delta_{\text{OS}}^{(\text{B,E})}(\mu)=\delta_{\text{MAN}}(\mu)$} \\ \cline{1-2} Region E & $\mu M <K\leq\mu M\cdot\delta_{\text{MAN}}(\mu)\leq M$ & & & \\ \hline
		Region C & $\mu M<M<K$ & \multirow{2}{*}{--} & \multirow{2}{*}{\checkmark} & \multirow{2}{*}{$	\delta_{\text{OS}}^{(\text{C,D})}(\mu)=\frac{K+\delta_{\text{MAN}}(\mu)}{1+\mu M}$} \\ \cline{1-2} Region D & $\mu M\cdot\max\Big\{1,\delta_{\text{MAN}}(\mu)\Big\}<K\leq M$ & & & \\ \hline
	\end{tabular}
	\captionof{table}{\small \alert{Definition of $(\mu,K,M)$ region triplets and their achievable one-shot NDT. The achievable one-shot NDT in Region A coincides with the lower bound and is thus NDT-optimal.}}
	%\caption{}
	\label{tab:def_reg_ach_NDT}
\end{center}
%\end{table}
\alert{The delivery time of the proposed one-shot scheme is devoted to both RNs and UEs. It is intuitive to expect cases where the delivery of requested files by the UEs may take longer than the delivery of uncached file fractions by the RNs. For instance, we expect that for $M\ll K$, irrespective of the fractional cache size, the file delivery to UEs through the interference channel represents the % bottleneck 
channel limitation from a delivery time perspective. However, finding the exact areas (including the transition) as a function of $\mu, K$ and $M$ where either RN or UE file delivery through broadcast or interference channel represents the bottleneck from a latency perspective for one-shot schemes is of interest. To this end, we conclude from Theorem \ref{th:one_shot_ach_NDT} that the functional behavior of the achievable one-shot NDT changes for different region triplets $(\mu,K,M)$ as follows.} Specifically, when neglecting the discretization of the fractional cache size $\mu$ to values  $\{0,\nicefrac{1}{M},\nicefrac{2}{M},\nicefrac{3}{M},\ldots,\nicefrac{(M-1)}{M},1\}$, Table \ref{tab:def_reg_ach_NDT} specifies how the one-shot NDT expression \eqref{eq:NDT_OS_ach} simplifies for the given region triplets. %\vspace{0.5em}
%\begin{itemize}
%	\item $\text{Region A}:\begin{cases}
%	K\leq\mu M<M<\frac{1}{1-2\mu},\mu\leq\frac{1}{2}&(\text{Region A}_1)\\
%	K\leq\mu M\leq M,M>\frac{1}{1-2\mu},\mu>\frac{1}{2}&(\text{Region A}_2)
%	\end{cases}\quad\qquad\qquad\text{ with NDT}$ 
%	\begin{align}\label{eq:ach_OS_NDT_RegA}
%	\delta_{\text{OS}}^{(\text{A})}(\mu)=1,
%	\end{align}
%	\item $\text{Region Pair (B,E)}:\begin{cases}
%	K\leq\mu M,\frac{1}{1-2\mu}\leq M,\mu\leq\frac{1}{2}&(\text{Region B})\\\mu M <K\leq\mu M\cdot\delta_{\text{MAN}}(\mu)\leq M&(\text{Region E})
%	\end{cases}\qquad\text{ with NDT}$
%	\begin{align}\label{eq:ach_OS_NDT_RegBE}
%	\delta_{\text{OS}}^{(\text{B,E})}(\mu)=\delta_{\text{MAN}}(\mu),
%	\end{align}
%	\item $\text{Region Pair (C,D)}:\begin{cases}
%	\mu M<M<K&(\text{Region C})\\\mu M\cdot\max\Big\{1,\delta_{\text{MAN}}(\mu)\Big\}<K\leq M&(\text{Region D})
%	\end{cases}\quad\text{ with NDT}$
%	\begin{align}\label{eq:ach_OS_NDT_RegCD}
%	\delta_{\text{OS}}^{(\text{C,D})}(\mu)=\frac{K+\delta_{\text{MAN}}(\mu)}{1+\mu M}.
%	\end{align}
%\end{itemize}
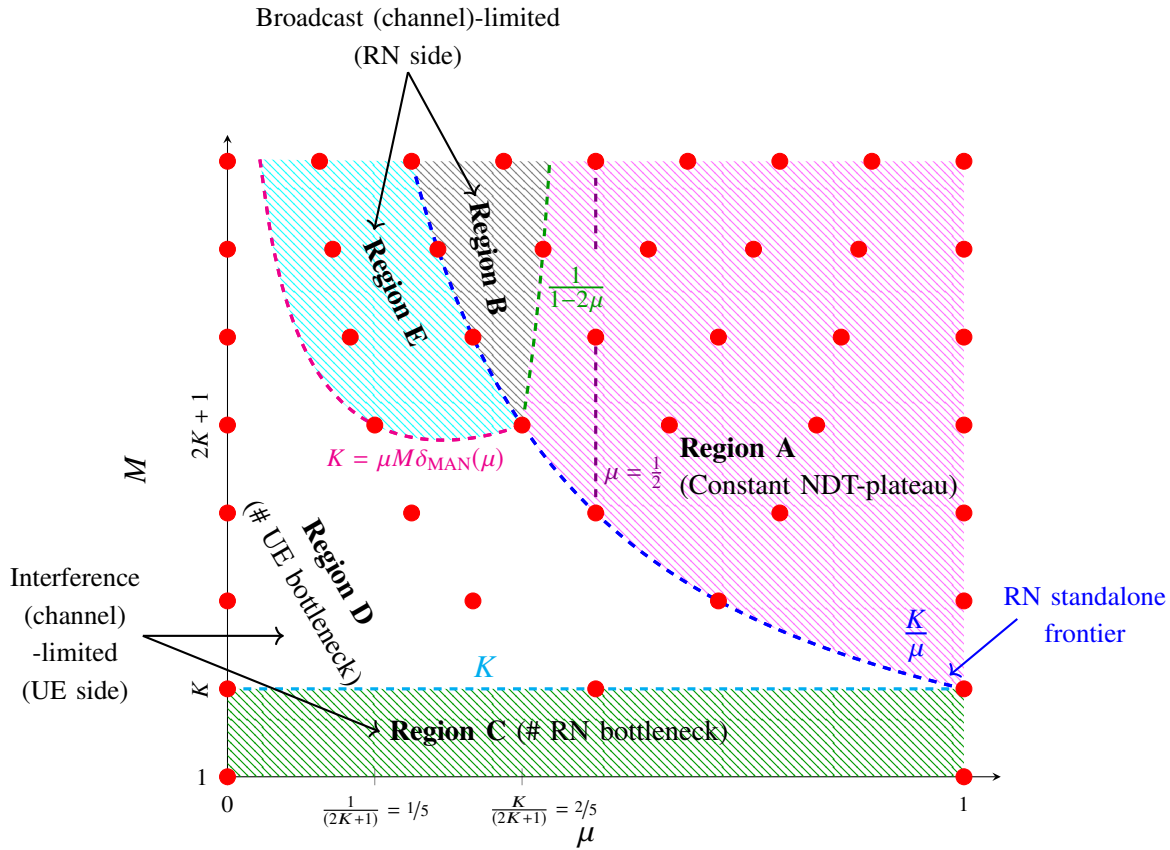
\begin{figure}
	\centering
	\vspace{-10em}
	\begin{tikzpicture}[scale=1.5]
	%\hspace{15em}
	\PlotMMuRegionsTwo
	\end{tikzpicture}
	\caption[Plot of all regions]{\small 2D $(\mu,M)$-plot of all Regions A, B, C, D and E for constant $K$ $(K=2)$. The labels on the graph indicate the functional relationship at the borders of neighboring regions. The discrete points illustrate the fractional cache sizes $\mu\in\Big\{0,\frac{1}{M},\ldots,\frac{M-1}{M},1\Big\}$ for which the achievable one-shot NDT expression $\delta_{\text{OS}}(\mu)$ in Eq. \eqref{eq:NDT_OS_ach} actually hold. \alert{The annotations to the regions specify the main characteristics of the respective region. The channel limitations specify which channel -- broadcast or interference channel -- is characteristic for the delivery time overhead. The RN standalone frontier, where $\mu M=K$ holds, represents scenarios for which all $K$ UEs can be served by any subset of $\mu M$ RNs \emph{without} the need of the DeNB.}}
	\label{fig:Region_plot_advanced}
\end{figure} \alert{The regions of Table \ref{tab:def_reg_ach_NDT} are illustrated in Fig. \ref{fig:Region_plot_advanced} for constant $K$ ($K=2$).} We state two interesting observations on these regions in the following remarks.

\begin{remark}
	Interestingly, when $M\geq 2K+1$, we see that for $\mu\geq\nicefrac{1}{M}$, the achievable one-shot NDT does not depend on $K$, i.e., the number of UEs. Instead, the NDT is solely dependent on the number of RNs $M$. Consequently, our one-shot NDT behaves identical to the achievable NDT of the MAN scheme. This is due to the fact that %the bottleneck in 
	the delivery time is governed by the delivery of uncached file fragments to the RNs through the broadcast channel. Then, this problem reduces to the initial receiver-based, single server coded caching problem of Maddah-Ali and Niesen. For this setting, we recall that the MAN NDT for $M$ receivers consists of the local and global caching gains (captured by the factors $1-\mu$ and $\frac{1}{1+\mu M}$ in $\delta_{\text{MAN}}(\mu)$, respectively) \cite{Maddah-Ali2}. In conclusion, we refer, respectively, to the attainable NDT as the DeNB-to-RN) \emph{(broadcast-limited NDT} and \alert{cases where $M/K>2$ as cases of \emph{high} $M/K$ (with respect to the number of UEs)}.  
\end{remark}
\begin{remark}
	As opposed to the previous remark, we observe that in the one-shot scheme the interference channel to the UEs functions as the bottleneck from a delivery time perspective as long as $M<2K+1$. \alert{We call these instances, all with respect to the number of UEs, as cases of \emph{moderate} $M/K$ when $1\leq M/K\leq 2$ and \emph{low} $M/K$ when $M/K<1$.} We then name the attainable NDT of $\frac{K+\delta_{\text{MAN}}(\mu)}{1+\mu M}$ the \emph{interference-limited NDT}.    
\end{remark}

In the following four corollaries, we state the relations of the NDT lower bound from Theorem \ref{theorem_lower_bound} and the one-shot scheme upper bound of Theorem \ref{th:one_shot_ach_NDT} of the aforementioned regions for discretized $\mu$.
\begin{corollary}[One-Shot NDT Optimality]\label{corr:OS_NDT_Opt} The one-shot scheme %presented in Section \ref{cha_one_shot} 
	is optimal (i.e., it coincides with the lower bound) \alert{achieving the minimum NDT given by} $\delta^{\star}(\mu)=1$, when the triplet $(\mu,K,M)$ satisfies any of the following conditions:
	\begin{enumerate}
		\item[1a)] $K<M,\mu\geq\frac{K}{M},\mu M\geq \Bigl\lceil{\frac{M-1}{2}}\Bigr\rceil,\mu\leq\frac{1}{2}$; or
		\item[1b)] $K\leq M,\mu\geq\frac{K}{M}$ and $\mu>\frac{1}{2}$.
	\end{enumerate} 
\end{corollary}	
\begin{proof} The proof follows from the fact that all discretized $\mu$ values inside Region A attain an NDT $\delta_{\text{OS}}^{(A)}(\mu)=1$ that matches the lower bound. 
\end{proof}
%\vspace{0.5em}
\begin{remark}
	With this corollary, we make the observation that at \emph{high} $M/K$ a fractional cache size of \emph{approximately} $\nicefrac{1}{2}$ is sufficient of achieving the lowest attainable NDT of $1$. \alert{In other words, this shows that when $M/K>2$ prefetching half of each file (from the entire library of files) and applying the one-shot scheme is \emph{delivery time optimal}.} Caching more than that will \emph{not} reduce the delivery time any further.         
\end{remark}
The remaining three corollaries state the multiplicative gap of the one-shot scheme with respect to the lower bound. 
%The proofs of these corollaries will be given in multiple subsection of \ref{cha_gap}. 
\vspace{0.5em}
\begin{corollary}[\alert{Gap of NDT for High $M/K$ ($M/K>2$)}]\label{corr:Gap_OS_NDT_High_RN_Quantity} The multiplicative gap between the one-shot scheme presented in Section \ref{cha_one_shot} and the lower bound on the NDT for $M\geq 2K+1$ is \alert{upper bounded by} %within 
	\begin{enumerate}
		\item[1a)] $\frac{M-1}{2}$ for $\mu\leq\frac{1}{M}$,
		\item[1b)] $\frac{M-\theta}{1+\theta}$ for any $\mu\in\Big[\mu'(\theta),\frac{\ceil{\nicefrac{(M-1)}{2}}}{M}\Big]$, where $\mu'(\theta)=\frac{\ceil{\theta}}{M}$ and $\theta\in\Big[1,\frac{M-3}{2}\Big]$.
	\end{enumerate}	
	\begin{proof} \alert{The proof of this corollary is given in the appendix. Specifically, details of the proof can be found in subsections \ref{subsec:be_reg_pair_gap} and \ref{subsec:de_reg_pair_gap}.}  
	\end{proof}
\end{corollary}	
\vspace{0.5em}
\begin{corollary}[\alert{Gap of NDT for Moderate $M/K$ ($1\leq M/K\leq 2$)}]\label{corr:Gap_OS_NDT_Moderate_RN_Quantity} The multiplicative gap between the one-shot scheme presented in Section \ref{cha_one_shot} and the lower bound on the NDT at \emph{moderate} $M/K$ is 
	\begin{enumerate}
		\item $1$ for $K=1$
		and within (or upper bounded by)
		\item[2a)] $\frac{K}{2}+\frac{M-1}{4}$ for $\mu\leq\frac{1}{M}$ and $K\geq 2$,
		\item[2b)] $\frac{dK}{M+1}+\frac{d(d-1)}{M+1}$ for any $\mu\in\Big[\mu'(\kappa_d),\frac{K}{M}\Big]$, where $\mu'(\kappa_d)=\frac{\ceil{\kappa_d}}{M}$, $\kappa_d=\frac{M+1-d}{d}$ and $d\in\Big[\frac{M+1}{K},\frac{M+1}{2}\Big]$ for $K\geq 2$.
	\end{enumerate}	
\end{corollary}
\begin{proof} \alert{The proof of this corollary is given in subsection \ref{subsec:cd_reg_pair_gap}.}  
\end{proof}	
\vspace{0.5em}
\begin{corollary}[\alert{Gap of NDT for Low $M/K$ ($M/K<1$)}]\label{corr:Gap_OS_NDT_Low_RN_Quantity} The multiplicative gap between the one-shot scheme presented in Section \ref{cha_one_shot} and the lower bound on the NDT at \emph{low} $M/K$ is 
	\begin{enumerate}
		\item $1$ for $(K,M)=(1,2)$
		and upper bounded by 
		\item $1+\frac{1}{K}\Big(\frac{K}{2}-\frac{2}{K}\Big)$ for $M=1,K>2$,
		\item[3a)] $1+\Big(\frac{K}{2}+\frac{M-5}{4}\Big)\cdot\frac{M}{K}\cdot\frac{(K+M+1)}{(K+M-1)}$ for $\mu\leq\frac{1}{M}$ and $M\geq 2$,
		\item[3b)] $d+\frac{d(d-1)}{K}$ for any $\mu\in\Big[\mu'(\kappa_d),1\Big]$, where $\mu'(\kappa_d)=\frac{\ceil{\kappa_d}}{M}$, $\kappa_d=\frac{M+1-d}{d}$ and $d\in\Big[\frac{M+1}{M},\frac{M+1}{2}\Big]$ for $M\geq 2$.
	\end{enumerate}	
\end{corollary}	
\begin{proof} \alert{The proof of this corollary is given in subsection \ref{subsec:cd_reg_pair_gap} of the appendix.}  
\end{proof}	
We use the last three corollaries to show that a fractional cache size of approximately $\nicefrac{1}{2}$ generates a \emph{constant} multiplicative gap less than $3$ for arbitrary $K$ and $M$. This is stated in the following corollary. 

\begin{corollary}[Constant Gap]\label{corr:const_gap} For fractional cache sizes $\mu\geq\mu_{C}\triangleq\frac{\ceil{\nicefrac{(M-1)}{2}}}{M}$ the multiplicative gap between the one-shot scheme presented in Section \ref{cha_one_shot} and the optimal NDT is constant \alert{versus $\mu$} and is \alert{bounded by %within 
		a factor of $\nicefrac{8}{3}$.} 
\end{corollary}	  
\begin{proof}
	At high $M/K$, the NDT optimality at $\mu_{C}$ is shown in Corollary \ref{corr:OS_NDT_Opt}. For this case, the gap is per definition of optimality $1$. \alert{Next, at moderate $M/K$, the gap is, respectively, for $K=1$ and $(K,M)=(2,2)$ $1$ and $\frac{5}{4}$ as stated by Corollary \ref{corr:Gap_OS_NDT_Moderate_RN_Quantity} (points 1 and 2b) with $d=\frac{3}{2}$). Also at moderate $M/K$, when $K\geq 2, M\geq 3$, we use point 2b) with $d=2$ to find the multiplicative gap} 
	\begin{align*}
	\frac{dK}{M+1}+\frac{d(d-1)}{M+1}\Bigg|_{d=2}=\frac{2(K+1)}{M+1}\leq 2
	\end{align*} for fractional cache sizes $\mu\geq\mu'(\kappa_2)=\mu_C=\frac{\ceil{\nicefrac{(M-1)}{2}}}{M}$. Finally, at low $M/K$, the gap for the special cases $(K,M)=(2,1)$ and $M=1,K>2$ are according to 1) and 2) of Corollary \ref{corr:Gap_OS_NDT_Low_RN_Quantity} $1$ and
	\begin{align*}
	1+\frac{1}{K}\Big(\frac{K}{2}-\frac{2}{K}\Big)\leq\frac{3}{2},
	\end{align*} respectively. Additionally, when $M\geq 3$, 2b) of Corollary \ref{corr:Gap_OS_NDT_Low_RN_Quantity} with $d=2$ generates the gap \begin{align*}
	d+\frac{d(d-1)}{K}\Bigg|_{d=2}=2\bigg(1+\frac{1}{K}\bigg)\leq2\bigg(1+\frac{1}{M}\bigg)\leq\frac{8}{3}.
	\end{align*} Combining all cases, we conclude that the gap is bounded from above by $\nicefrac{8}{3}$.   
\end{proof} In conclusion, this corollary shows that one-shot schemes at fractional cache sizes of $\frac{1}{2}$ $\Big(\frac{1}{2}\geq\frac{\ceil{\nicefrac{(M-1)}{2}}}{M}\Big)$ are optimal within a constant \alert{(with respect to the lower bound) multiplicative gap of optimality}. 
%These schemes, apart from the full-duplex capabilities of the RNs, are simpler in implementation and thus practically more applicable than alignment schemes presented in Section \ref{sec:ach_special}.
\alert{The main disadvantage of this scheme becomes visible when considering regions C and D. For these regions, the first and second phase provide (per channel use) only a subset of $\mu M$ and $\min\{K,1+\mu M\}$ UEs with their desired symbols while the remaining UEs observe interference. Interference alignment provides the opportunity of alleviating the effect of undesired symbols. To this end, we establish achievability schemes that involve a novel beamforming design that facilitates (i) multicasting opportunities (when $M\geq 2$) on the DeNB-RN broadcast channel, (ii) (joint) zero-forcing opportunities and (iii) subspace interference alignment. The notion of subspace alignment was first introduced in \cite{Suh08}. The idea is to align interferences into a multi-dimensional subspace instead of a single dimension \cite{MaddahAliAlignment}. Through this beamforming design and using the insights from Lemma \ref{corr_mu_0_and_1} and Theorem \ref{th:one_shot_ach_NDT}, we are able to establish the complete NDT-tradeoff for $K+M\leq 4$. The following theorem specifies this tradeoff.}   
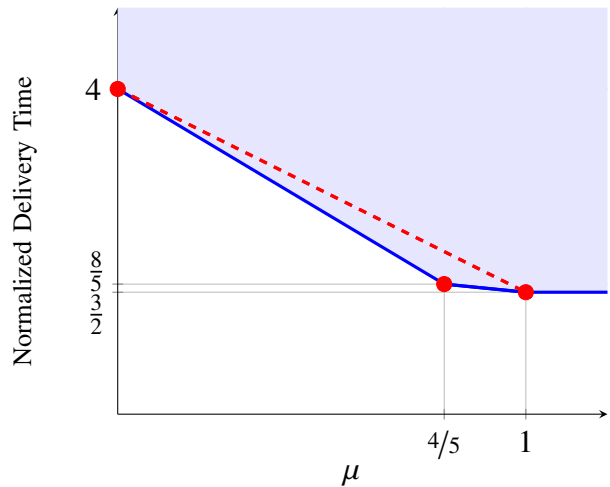
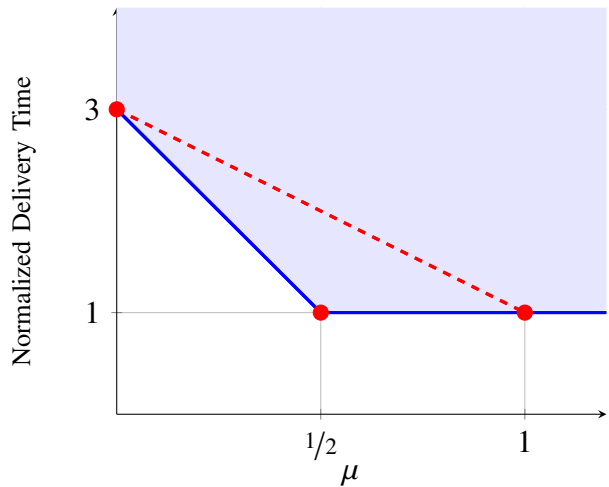
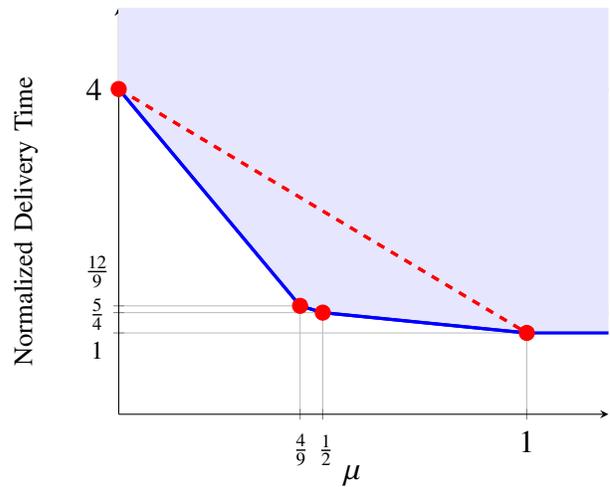
\begin{figure*}
	\centering
	\begin{subfigure}[b]{\textwidth}
		%\hspace{2em}
		\centering
		\begin{tikzpicture}[scale=0.95]
		\PlotNDTMOneKThree
		\end{tikzpicture}
		\caption[NDT for $M=1,K=3$]{\small NDT for $(K,M)=(3,1)$} 
		\label{fig:NDT_M_K_1_3_edit}     
	\end{subfigure} %\newline
	%\hspace{3em}
	\\[2ex]
	\begin{subfigure}[b]{0.475\textwidth}
		\centering
		\vspace{1em}
		\begin{tikzpicture}[scale=0.95]
		\PlotNDTMTwoKOne
		\end{tikzpicture}
		\caption[NDT for $M=2,K=1$]{\small NDT for $(K,M)=(1,2)$}
		\label{fig:NDT_M_K_2_1}
	\end{subfigure}
	\hfill
	\begin{subfigure}[b]{0.475\textwidth}  
		\centering 
		\begin{tikzpicture}[scale=0.95]
		\PlotNDTMTwoKTwo
		\end{tikzpicture}
		\caption[NDT for $M=K=2$]{\small NDT for $(K,M)=(2,2)$} 
		\label{fig:NDT_M_K_2_2}       
	\end{subfigure}
	\caption
	{\small Optimal NDT as a function of $\mu$ for (a) $(K,M)=(3,1)$, (b) $(K,M)=(1,2)$ and (c) $(K,M)=(2,2)$. The dashed line shows the achievable NDT of a time-sharing based unicasting-zero-forcing scheme.} 
	\label{fig:NDT_M_K_2_LEQ_2}
\end{figure*} 
\begin{figure*}
	\centering
	\begin{subfigure}[b]{\textwidth}
		%\hspace{2em}
		\centering
		\begin{tikzpicture}[scale=0.95]
		\PlotDoFMOneKThree
		\end{tikzpicture}
		\caption[DoF for $M=1,K=3$]{\small Achievable DoF for $(K,M)=(3,1)$} 
		\label{fig:DOF_M_K_1_3_edit}     
	\end{subfigure} %\newline
	%\hspace{3em}
	\\[2ex]
	\begin{subfigure}[b]{0.475\textwidth}
		\centering
		\vspace{1em}
		\begin{tikzpicture}[scale=0.95]
		\PlotDoFMTwoKOne
		\end{tikzpicture}
		\caption[DoF for $M=2,K=1$]{\small Achievable DoF for $(K,M)=(1,2)$}
		\label{fig:DOF_M_K_2_1}
	\end{subfigure}
	\hfill
	\begin{subfigure}[b]{0.475\textwidth}  
		\centering 
		\begin{tikzpicture}[scale=0.95]
		\PlotDoFMTwoKTwo
		\end{tikzpicture}
		\caption[DoF for $M=K=2$]{\small Achievable DoF for $(K,M)=(2,2)$} 
		\label{fig:DOF_M_K_2_2}       
	\end{subfigure}
	\caption
	{\small Achievable sum DoF as a function of $\mu$ for (a) $(K,M)=(3,1)$, (b) $(K,M)=(1,2)$ and (c) $(K,M)=(2,2)$. The dashed line shows the DoF of a time-sharing based unicasting-zero-forcing scheme.} 
	\label{fig:DOF_M_K_2_LEQ_2}
\end{figure*}
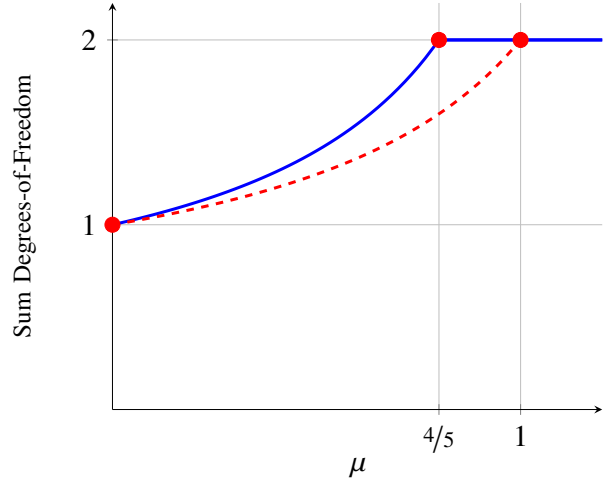
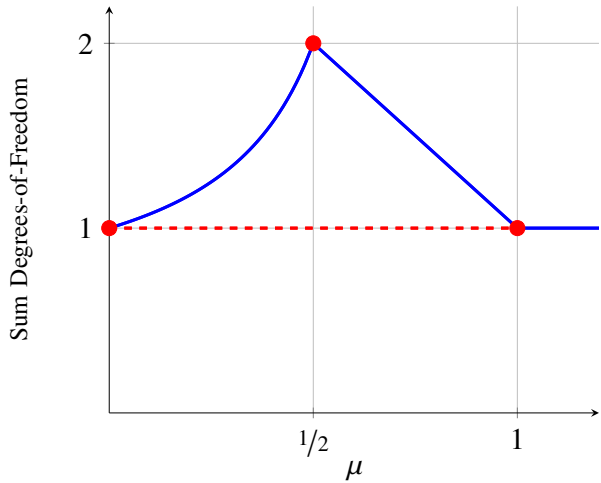
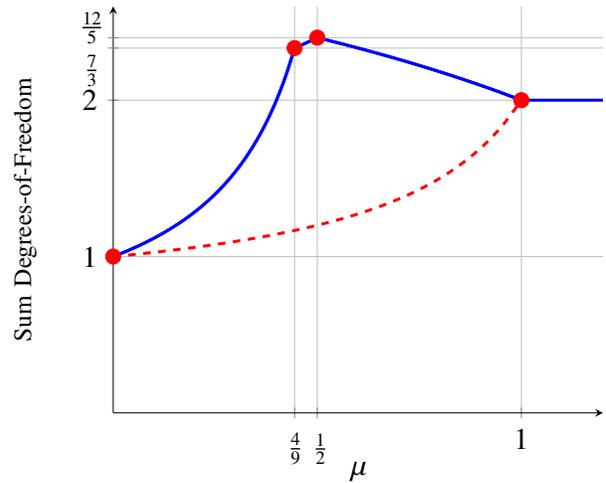 
\begin{theorem}[Optimal NDT Tradeoff]\label{th:opt_NDT_tradeoff}
	The optimal NDT tradeoff for the transceiver cache-aided network with one DeNB, $M$ RNs each endowed with a cache of fractional cache size $\mu\in[0,1]$, $K$ UEs satisfying $K+M\leq 4$ with $N\geq K+M$ \alert{number of files}, is given as 
%\begin{align}\label{eq:NDT_opt_trade_off}
%\delta^{\star}(\mu)=\max\Bigg\{&1,K+M-\mu M(K+M-1),\frac{K+M-\mu\Big(M^2+(K-3)(M-1)\Big)}{2},\nonumber\\&\qquad\frac{K+M-1-\mu\Big((M-1)(K+M-3)\Big)}{2}\Bigg\}.
%\end{align}
\begin{align}\label{eq:NDT_opt_trade_off}
\delta^{\star}(\mu)=\max\Bigg\{&1,K+M-\mu M(K+M-1),\frac{K+M-\mu\Big(M^2+(K-3)(M-1)\Big)}{2},\nonumber\\&\qquad\frac{K+M-1-\mu\Big(\alert{M^{2}+(K-3)(M-1)-M}\Big)}{2}\Bigg\}.
\end{align}
\end{theorem}
\begin{proof}
The lower bound on the NDT for this setting readily follows from Theorem \ref{theorem_lower_bound}. Concretely, the last three terms inside the max-expression of \eqref{eq:NDT_opt_trade_off} correspond to %$\delta_{\text{LB}}(\mu,\ell=M,s=1)$, $\delta_{\text{LB}}(\mu,\ell=M,s=2)$ and $\delta_{\text{LB}}(\mu,\ell=M-1,s=1)$
\alert{$\delta_{\text{LB}}(\mu,M,1)$, $\delta_{\text{LB}}(\mu,M,2)$ and $\delta_{\text{LB}}(\mu,M-1,2)$}, respectively. The outer bound (achievability) on the NDT is presented in Section \ref{sec:ach_special}. Shortly, we establish the achievability for (at most) four corner points at \alert{zero-cache and full-cache} fractional cache sizes $\mu=0$, $\mu=1$ as well as \alert{intermediate} cache sizes $\mu'$ in the interval $0<\mu'<1$. The first two corner points are achievable through DeNB broadcasting and cooperative DeNB-RN zero-forcing beamforming (cf. Lemma \ref{corr_mu_0_and_1}). \alert{In the achievability scheme at intermediate fractional cache size $\mu=\nicefrac{1}{M}$, on the other hand, the one-shot scheme of Theorem \ref{th:one_shot_ach_NDT} is optimal for $(K,M)\in\{(1,2),(2,2),(1,3)\}$. This is shown in Figs. \ref{fig:NDT_M_K_2_1} and \ref{fig:NDT_M_K_2_2} for the cases $(K,M)\in\{(1,2),(2,2)\}$. For the cases $(K,M)\in\{(3,1),(2,2)\}$, respectively, we establish the optimal NDT at fractional cache sizes $\nicefrac{4}{5}$ and $\nicefrac{4}{9}$ (cf. Figs. \ref{fig:NDT_M_K_1_3_edit} and \ref{fig:NDT_M_K_2_2}) through optimized precoding design that synergistically integrates subspace alignment %and MAN multicasting 
with zero-forcing beamforming consuming \emph{finite} channel uses.}     
\end{proof}

\alert{In the following remarks, we discuss the results of Theorem \ref{th:opt_NDT_tradeoff} in further detail. In the discussion, we assume that RNs and UEs all request \emph{distinct} files. This represents the worst-case demand scenario.}  

\begin{remark}[\alert{Subspace Interference Alignment}] \alert{It is of interest to discuss why aligning all interferences into \emph{one dimension} is not feasible when $\mu<\frac{1}{M}$. In cases, where $\mu<\frac{1}{M}$ the collection of all $M$ caches cannot hold the entire library of files. Thus, under the placement strategy when each RN caches $\mu L'$ ($L'=L/\log(P)$) independent Gaussian symbols (i.e., no overlaps in file chunks being cached) of a file, we observe that $(1-\mu M)L$ symbols are \emph{only} available at the DeNB. In the following, we argue why multiple subspace dimensions have to be allocated for the interference from both RN and UE perspective.}    

\alert{First, we consider the delivery of file content to all $M$ RNs. We observe that each RN is interested in $(1-\mu)L'$ uncached symbols of its requested file; out of which $(1-\mu M)L'$ symbols are only available at the DeNB. In the worst-case scenario, this amounts to $M(1-\mu M)L$ symbols in total since there are $M$ distinct files that $M$ RNs request. These symbols have to be broadcast to the RNs. However, they represent interference to \emph{all} $K$ UEs. Aligning these symbols at all $K$ UEs to a single signaling dimension is not feasible because alignment at the UEs would make these symbols \emph{indistinguishable} at all $M$ RNs. To preclude this, all $M(1-\mu M)L'$ symbols have to be aligned at \emph{distinct} subspaces.} %(i.e, at least $M(1-\mu)L$ interference dimensions are required). 

\alert{Second, subspace alignment is also necessary when focusing on the delivery of $K$ distinct files with one single file being desired by one UE. A similar line of argument as in the previous paragraph suggests that $K(1-\mu M)L'$ symbols are only available at the DeNB (and not at the $M$ RNs) and have to be conveyed from the DeNB to $K$ UEs. For the sake of reliable decodability at the UEs, all these symbols have to be distinguishable from each other. Since only $(K-1)(1-\mu M)L'$ of those symbols represent interference at a single UE, the interference dimension increases by $(K-1)(1-\mu M)L'$.}

\alert{In conclusion, at least $(K+M-1)(1-\mu M)L'$ interference dimensions are required. For the cases $(K,M)\in\{(3,1),(2,2)\}$, at fractional cache sizes $\nicefrac{4}{5}$ ($L'=5$) and $\nicefrac{4}{9}$ ($L'=9$), the number of interference dimensions is in agreement with $(K+M-1)(1-\mu M)L'=3$.} 
  
\end{remark}

\begin{remark}[\alert{Feasibility for Constant Channels}]
\alert{For completeness, we would like to emphasize that the complete NDT tradeoff of Theorem \ref{th:opt_NDT_tradeoff} is applicable to constant channels as well. This is due the fact that both the one-shot scheme as well as the alignment scheme are feasible for time-invariant channels. In particular, the synergistic beamforming design is feasible for constant channels under the umbrella of \emph{real interference alignment} \cite{Motahari14,Maddah-Ali10}. Whether a two-phase precoding design with constant channels (similar to previous work on relay-aided X-channels \cite{Frank2014}) attains close-to-optimal performance is an interesting extension to work on. However, it is beyond the scope of this paper.}
\end{remark}
\begin{remark}[\alert{Inverse Sum DoF vs. NDT}] \alert{From the optimal NDT tradeoff of Theorem \ref{th:opt_NDT_tradeoff}, we may compute the resulting \emph{achievable} (sum) DoF as follows
	\begin{equation*}
	\text{DoF}\geq K\cdot\underbrace{\frac{1}{\delta^{\star}(\mu)}}_{\text{per-UE DoF}}+M\cdot\underbrace{\frac{(1-\mu)}{\delta^{\star}(\mu)}}_{\text{per-RN DoF}}.
	\end{equation*}
The achievable DoF for $(K,M)\in\{(3,1),(1,2),(2,2)\}$ is shown in Fig. \ref{fig:DOF_M_K_2_LEQ_2}. When comparing NDT and DoF (Fig. \ref{fig:NDT_M_K_2_LEQ_2} vs. Fig. \ref{fig:DOF_M_K_2_LEQ_2}), we clearly see that an increase in sum rate (measured by the DoF) is not necessarily equivalent to a decrease in delivery time (measured by the NDT). In other words, we observe that the NDT metric is not necessarily proportional to the inverse of the sum DoF for fractional cache sizes exceeding $\tilde{\mu}=\frac{4}{5}$ ($\tilde{\mu}=\frac{1}{2}$) when $(K,M)=(3,1)$ ($(K,M)\in\{(1,2),(2,2)\}$). In fact, interestingly there are cases when both NDT and achievable DoF decrease for increasing $\mu$. At first glance, this may seem counterintuitive. However, a closer look reveals that an increasing $\mu$ has the advantage of elevating the per-UE DoF at the cost of a declining per-RN DoF. In consequence, the overall achievable sum DoF may drop, but the increase in the per-UE DoF comes along with a decay in NDT. Ultimately, this observation suggests that the inverse sum DoF can be a misleading metric with respect to the delivery time. This is mainly due to the fact that the DoF metric loses information on the per-user DoF for asymmetric rate allocation scenarios.}        
\end{remark}

\section{Lower Bound (Converse) of the Minimum NDT}
\label{sec:lw_bd}

In this section, we present the proof of the lower bound of the minimum NDT in Theorem \ref{theorem_lower_bound}. The method of the proof extends on the approaches of \cite{avik} and \cite{conference214}.  

In the following, we expound the key idea of the proof. To this end, we introduce the following worst-case considerations. \alert{First, we presume that all $K$ UEs request \emph{distinct} files $W_{d_j}$ ($d_j\neq d_\ell,j,\ell\in[K],j\neq \ell$). Second, for the sake of notational simplicity and without loss of generality, we set the requested files by the $K$ UEs to $W_{[1:K]}=(W_1,W_2,\ldots,W_K)$, or shortly $\mathbf{d}[1:K]=[K]$. Third, in our proof we assume that there at least $K+M$ distinct files $W_{d_k}$ available, i.e., $N\geq K+M$.} Under this UE demand pattern and given channel realizations $\mathbf{f},\mathbf{g}$ and $\mathbf{H}$, we obtain a lower bound on the delivery time $T=T(\mathbf{d},\mathbf{f},\mathbf{g},\mathbf{H})$, and therefore ultimately on the NDT, of any \emph{feasible} scheme. \alert{The UE request pattern is one possible (out of ${{N}\choose{K}}K!$ requests) worst-case scenario due to the following line of argument. Redundancy in the UE request pattern can, if anything, only decrease the NDT further. This is due to the fact that the NDT of a smaller network (with respect to the number of UEs) is in general a lower bound on bigger networks. For instance, when we consider redundancy in the file request pattern at the UE side with respect to one single file, we make the following observation. The transceiver cache-aided network under study with $M$ RNs and $K$ UEs with $\tilde{K},\tilde{K}\leq K$, UEs having a \emph{redundant} request (i.e., requesting the same file) at the UE side behaves similarly (with respect to the NDT performance) to a smaller network with $M$ RNs but only $K-\tilde{K}+1$ UEs.} 

\begin{figure*}%[h]
	\begin{center}
		\vspace{15em}
		\hspace{-6em}
		\begin{tikzpicture}[scale=0.8]
		\SymModLB
		\end{tikzpicture}
		\vspace{-16em}
		\caption{\small Illustration of the proof of converse. The top part of the figure shows how $K+\ell$ unique files become (reliably) decodable in the high SNR regime when a receiver is aware of the information subset $\{\mathbf{y}_{u,[1:s]}^{T},S_{[1:\ell]}\}$ with $s+\ell\geq M+1$.}	
		\label{fig:lw_bound}
	\end{center}
\end{figure*}
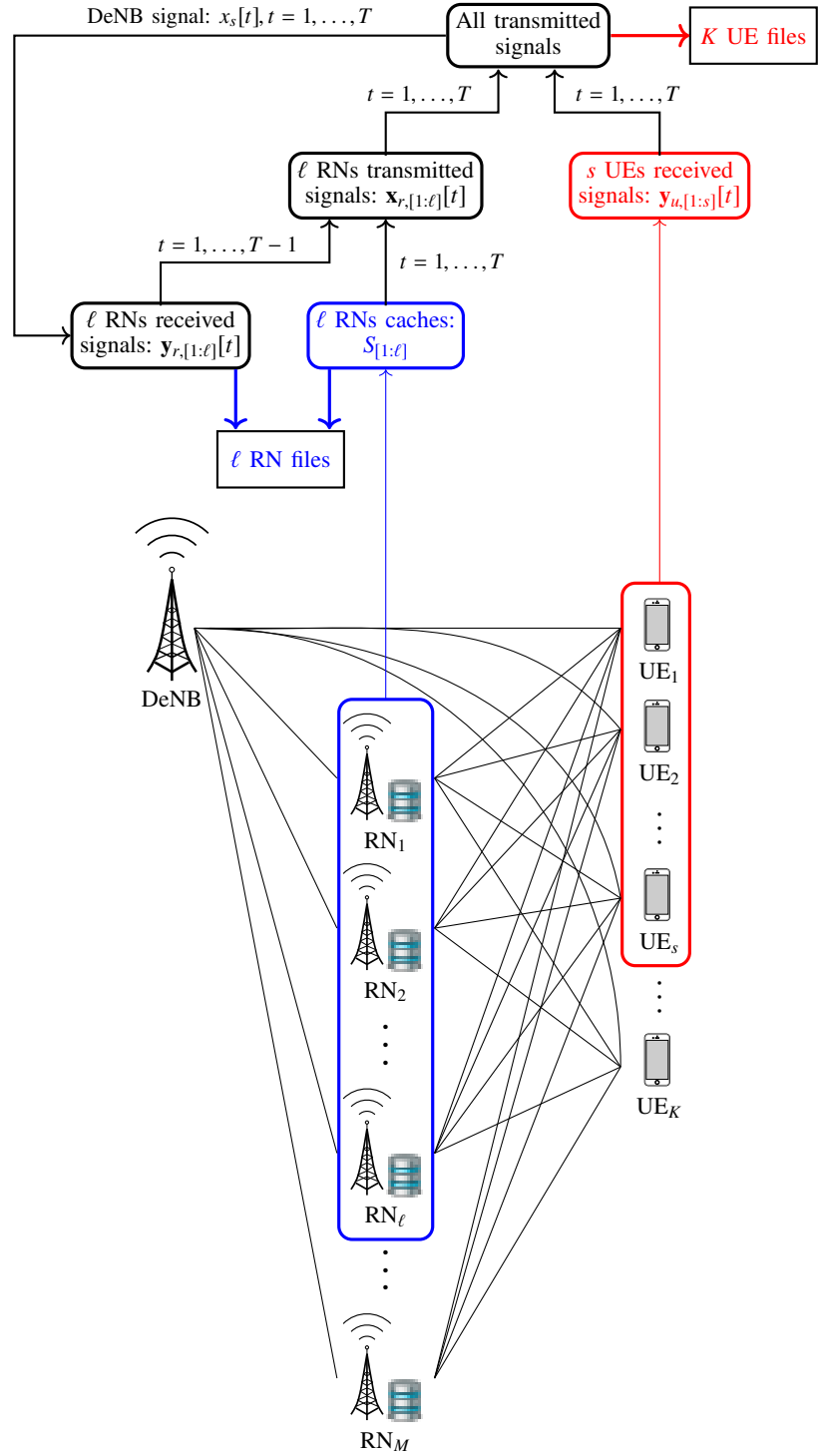
The key idea in establishing the lower bound on the NDT is that $K+\ell$ \alert{\emph{unique}} files, comprising of all $K$ files $W_{[1:K]}$ requested by the UEs and $\ell$ files desired by \alert{at least a} subset of $\ell$ RNs (out of $M$ RNs), e.g., \alert{$W_{[K+1:K+\ell]}$\footnote{%For ease of notation, we used $W_{d_{[K+1:K+\ell]}}$ to denote $W_{d_{K+1}},\ldots,W_{d_{K+\ell}}$. 
For ease of notation, we assume that the $\ell$ files retrieved by the RNs are $W_{[K+1:K+\ell]}$. 
}, can be retrieved in the \emph{high SNR regime} from} 

\begin{itemize}[leftmargin=*]
	\item $s$ output signals of the UEs, e.g., $\mathbf{y}^{T}_{u,[1:s]}$ for $1\leq s\leq \min\{M+1,K\}$, and
	\item $\ell$ cached contents of $\ell$ RNs, e.g., $S_{[1:\ell]}$, where $\bar{s}\leq\ell\leq M$ and $\bar{s}=M+1-s$. 
\end{itemize} We note that since $s+\ell\geq M+1$ holds, we are able to reconstruct all $M+1$ transmit signals ($x_s[t]$ and $x_{r,m}[t],m\in[M]$) at all $T$ time instants of the delivery phase within bounded noise. \alert{This is shown in Fig. \ref{fig:lw_bound}. We would like to emphasize to the reader that through the choice of $\ell$, the bounds account for distinct RN request patterns.}
%\alert{In the converse, we use $\widetilde{W}\triangleq\{W_{[1:K]},W_{d_{[K+1:K+\ell]}}\}$}. 
We start the converse as follows:
%\begin{eqnarray}\label{eq:conv_1}
%(K+\ell)L &= & H\big(W_{[1:K+\ell]}\big) \stackrel{(a)}= H\big(W_{[1:K+\ell]}|W_{[K+\ell+1:N]}\big) \nonumber \\
% &=&I\big(W_{[1:K+\ell]};\mathbf{y}^{T}_{[1:s]},S_{[1:\ell]}|W_{[K+\ell+1:N]}\big)+H\big(W_{[1:K+\ell]}|\mathbf{y}^{T}_{[1:s]},S_{[1:\ell]},W_{[K+\ell+1:N]}\big).
%\end{eqnarray} 
\begin{eqnarray}\label{eq:conv_1}
\hspace{-1em}(K+\ell)L &= & H\big(W_{[1:K+\ell]}\big)= H\big(W_{[1:K+\ell]}|W_{[K+\ell+1:N]}\big) \nonumber \\
&=&I\big(W_{[1:K+\ell]};\mathbf{y}^{T}_{u,[1:s]},S_{[1:\ell]}|W_{[K+\ell+1:N]}\big)+H\big(W_{[1:K+\ell]}|\mathbf{y}^{T}_{u,[1:s]},S_{[1:\ell]},W_{[K+\ell+1:N]}\big)
\end{eqnarray} 
%\begin{eqnarray}\label{eq:conv_1}
%H\big(\widetilde{W}\big)&=& H\big(\widetilde{W}|W_{[K+\ell+1:N]}\big) \nonumber \\
%&=&I\big(W_{[1:K+\ell]};\mathbf{y}^{T}_{u,[1:s]},S_{[1:\ell]}|W_{[K+\ell+1:N]}\big)+H\big(W_{[1:K+\ell]}|\mathbf{y}^{T}_{u,[1:s]},S_{[1:\ell]},W_{[K+\ell+1:N]}\big)
%\end{eqnarray} 
In what follows, we upper bound each summand in Eq. \eqref{eq:conv_1} individually. 
% The first summand is upper bounded according to:
\alert{Using the chain rule of mutual information, the first term in \eqref{eq:conv_1} can be rewritten as shown
below
\begin{align}
\label{eq:conv_term_intermediate_step1}
&I\big(W_{[1:K+\ell]};\mathbf{y}^{T}_{u,[1:s]},S_{[1:\ell]}|W_{[K+\ell+1:N]}\big)\nonumber\\&=I\big(W_{[1:K+\ell]};\mathbf{y}^{T}_{u,[1:s]}|W_{[K+\ell+1:N]}\big)+I\big(W_{[1:K+\ell]};S_{[1:\ell]}|W_{[K+\ell+1:N]},\mathbf{y}^{T}_{u,[1:s]}\big). 
\end{align} Next, we use the non-negativity of mutual information ($I\big(W_{[1:K+\ell]};W_{[1:s]}|W_{[K+\ell+1:N]},\mathbf{y}^{T}_{u,[1:s]},S_{[1:\ell]}\big)$) followed by the chain rule of mutual information to upper bound \eqref{eq:conv_term_intermediate_step1} by
\begin{align}
\label{eq:conv_term_intermediate_step2}
&I\big(W_{[1:K+\ell]};\mathbf{y}^{T}_{u,[1:s]}|W_{[K+\ell+1:N]}\big)+I\big(W_{[1:K+\ell]};S_{[1:\ell]},W_{[1:s]}|W_{[K+\ell+1:N]},\mathbf{y}^{T}_{u,[1:s]}\big)\nonumber\\&= I\big(W_{[1:K+\ell]};\mathbf{y}^{T}_{u,[1:s]}|W_{[K+\ell+1:N]}\big)+I\big(W_{[1:K+\ell]};W_{[1:s]}|W_{[K+\ell+1:N]},\mathbf{y}^{T}_{u,[1:s]}\big)\nonumber\\&\qquad+I\big(W_{[1:K+\ell]};S_{[1:\ell]}|W_{[1:s]\cup[K+\ell+1:N]},\mathbf{y}^{T}_{u,[1:s]}\big). 
\end{align}}
\noindent \alert{Rewriting the three mutual information terms of \eqref{eq:conv_term_intermediate_step2} by their respective differential or discrete entropy terms and bounding them results in the following chain of inequalities.} 
\begin{flalign}
\label{eq:conv_term1}
%&I\big(W_{[1:K+\ell]};\mathbf{y}^{T}_{u,[1:s]},S_{[1:\ell]}|W_{[K+\ell+1:N]}\big)\nonumber\\&=I\big(W_{[1:K+\ell]};\mathbf{y}^{T}_{u,[1:s]}|W_{[K+\ell+1:N]}\big)+I\big(W_{[1:K+\ell]};S_{[1:\ell]}|W_{[K+\ell+1:N]},\mathbf{y}^{T}_{u,[1:s]}\big)\nonumber\\&\leq I\big(W_{[1:K+\ell]};\mathbf{y}^{T}_{u,[1:s]}|W_{[K+\ell+1:N]}\big)+I\big(W_{[1:K+\ell]};S_{[1:\ell]},W_{[1:s]}|W_{[K+\ell+1:N]},\mathbf{y}^{T}_{u,[1:s]}\big)\nonumber\\&= I\big(W_{[1:K+\ell]};\mathbf{y}^{T}_{u,[1:s]}|W_{[K+\ell+1:N]}\big)+I\big(W_{[1:K+\ell]};W_{[1:s]}|W_{[K+\ell+1:N]},\mathbf{y}^{T}_{u,[1:s]}\big)\nonumber\\&\qquad+I\big(W_{[1:K+\ell]};S_{[1:\ell]}|W_{[1:s]\cup[K+\ell+1:N]},\mathbf{y}^{T}_{u,[1:s]}\big)\nonumber\\&=
&h\big(\mathbf{y}^{T}_{u,[1:s]}|W_{[K+\ell+1:N]}\big)-h\big(\mathbf{y}^{T}_{u,[1:s]}|W_{[1:N]}\big)+ H\big(W_{[1:s]}|W_{[K+\ell+1:N]},\mathbf{y}^{T}_{u,[1:s]}\big)\nonumber\\&\qquad-H\big(W_{[1:s]}|W_{[1:N]},\mathbf{y}^{T}_{u,[1:s]}\big)+ H\big(S_{[1:\ell]}|W_{[1:s]\cup [K+\ell+1:N]},\mathbf{y}^{T}_{u,[1:s]}\big)- H\big(S_{[1:\ell]}|W_{[1:N]},\mathbf{y}^{T}_{u,[1:s]}\big)\nonumber\\
&\stackrel{(a)}\leq h\big(\mathbf{y}^{T}_{u,[1:s]}\big)-h\big(\mathbf{z}^{T}_{u,[1:s]}\big)+L\epsilon_L+ H\big(S_{[1:\ell]}|W_{[1:s]\cup [K+\ell+1:N]},\mathbf{y}^{T}_{u,[1:s]}\big)\nonumber\\&\stackrel{(b)}\leq sT\log(2\pi e(cP+1))-h\big(\mathbf{z}^{T}_{u,[1:s]}\big)+L\epsilon_L+\sum_{i=1}^{\ell}H\big(S_{i}|W_{[1:s]\cup [K+\ell+1:N]},\mathbf{y}^{T}_{u,[1:s]},S_{[1:i-1]}\big)\nonumber\\&\stackrel{(c)}= sT\log(cP+1)+L\epsilon_L+\sum_{i=1}^{\bar{s}}H\big(S_{i}|W_{[1:s]\cup [K+\ell+1:N]},\mathbf{y}^{T}_{u,[1:s]}, S_{[1:i-1]}\big)\nonumber\\&\quad+\sum_{i=\bar{s}+1}^{\ell}H\big(S_{i}|W_{[1:s]\cup [K+\ell+1:N]},\mathbf{y}^{T}_{u,[1:s]}, S_{[1:i-1]}\big)\nonumber\\&\leq sT\log(cP+1)+L\epsilon_L+\sum_{i=1}^{\bar{s}}H\big(S_{i}|W_{[1:s]\cup [K+\ell+1:N]}\big)+\sum_{i=\bar{s}+1}^{\ell}H\big(S_{i}|W_{[1:s]\cup [K+\ell+1:N]},\mathbf{y}^{T}_{u,[1:s]}, S_{[1:i-1]}\big)\nonumber\\&\stackrel{(d)}\leq sT\log(cP+1)+L\epsilon_L+\sum_{i=1}^{\bar{s}}\sum_{j=s+1}^{K+\ell}H\big(S_{i,j}\big)+\sum_{i=\bar{s}+1}^{\ell}H\big(S_{i}|W_{[1:s]\cup [K+\ell+1:N]},\mathbf{y}^{T}_{u,[1:s]}, S_{[1:i-1]}\big)\nonumber\\&\stackrel{(e)}\leq sT\log(cP+1)+L\epsilon_L+\bar{s}(K+\ell-s)\mu L+\sum_{i=\bar{s}+1}^{\ell}H\big(S_{i}|W_{[1:s]\cup [K+\ell+1:N]},\mathbf{y}^{T}_{u,[1:s]}, S_{[1:i-1]}\big)\nonumber\\
&\stackrel{(f)}= sT\log(cP+1)+L\epsilon_L+\bar{s}(K+\ell-s)\mu L\nonumber\\&\quad+\sum_{i=\bar{s}+1}^{\ell}H\big(S_{i}|W_{[1:s]\cup [K+\ell+1:N]},\mathbf{y}^{T}_{u,[1:s]}, S_{[1:i-1]},\mathbf{x}_{r,[1:i-1]}[1],\mathbf{z}^{T}_{r,[1:i-1]}\big)\nonumber
%&\stackrel{(f)}= sT\log(cP+1)+L\epsilon_L+\bar{s}(K+\ell-s)\mu L\nonumber\\&\quad+\sum_{i=\bar{s}+1}^{\ell}H\big(S_{i}|W_{[1:s]\cup [K+\ell+1:N]},\mathbf{y}^{T}_{u,[1:s]}, S_{[1:i-1]},\mathbf{x}_{r,[1:i-1]}[1],\mathbf{z}^{T}_{u,[s+1:K]},\mathbf{z}^{T}_{r,[1:i-1]}\big)\nonumber
\\&\leq sT\log(cP+1)+L\epsilon_L+\bar{s}(K+\ell-s)\mu L\nonumber\\&+\sum_{i=\bar{s}+1}^{\ell}H\big(S_{i},W_{[s+1:K+i-1]},\mathbf{z}^{T}_{u,[1:s]}|W_{[1:s]\cup [K+\ell+1:N]},\mathbf{y}^{T}_{u,[1:s]}, S_{[1:i-1]},\mathbf{x}_{r,[1:i-1]}[1],\mathbf{z}^{T}_{r,[1:i-1]}\big)\nonumber
%\\&\leq sT\log(cP+1)+L\epsilon_L+\bar{s}(K+\ell-s)\mu L\nonumber\\&+\sum_{i=\bar{s}+1}^{\ell}H\big(S_{i},W_{[s+1:K+i-1]}|W_{[1:s]\cup [K+\ell+1:N]},\mathbf{y}^{T}_{u,[1:s]}, S_{[1:i-1]},\mathbf{x}_{r,[1:i-1]}[1],\mathbf{z}^{T}_{u,[s+1:K]},\mathbf{z}^{T}_{r,[1:i-1]}\big)\nonumber
\\&\stackrel{(g)}\leq sT\log(cP+1)+L\epsilon_L+\bar{s}(K+\ell-s)\mu L+T\epsilon_P\log(P)\nonumber\\&\quad+\sum_{i=\bar{s}+1}^{\ell}H\big(S_i|W_{[1:K+i-1]\cup [K+\ell+1:N]},\mathbf{y}^{T}_{u,[1:s]}, S_{[1:i-1]},\mathbf{x}_{r,[1:i-1]}[1],\mathbf{z}^{T}_{u,[s+1:K]},\mathbf{z}^{T}_{r,[1:i-1]}\big)\nonumber
%&\stackrel{(g)}\leq sT\log(cP+1)+L\epsilon_L+\bar{s}(K+\ell-s)\mu L+To(\log(P))\nonumber\\&\quad+\sum_{i=\bar{s}+1}^{\ell}H\big(S_i|W_{[1:K+i-1]\cup [K+\ell+1:N]},\mathbf{y}^{T}_{u,[1:s]}, S_{[1:i-1]},\mathbf{x}_{r,[1:i]}[1],\mathbf{z}^{T}_{u,[s+1:K]},\mathbf{z}^{T}_{r,[1:i-1]}\big)\nonumber
\\&\stackrel{(d)}\leq sT\log(cP+1)+L\epsilon_L+\bar{s}(K+\ell-s)\mu L+T\epsilon_P\log(P)+\sum_{i=\bar{s}+1}^{\ell}\sum_{j=K+i}^{K+\ell}H\big(S_{i,j}\big)\nonumber\\&\stackrel{(e)}\leq sT\log(cP+1)+L\epsilon_L+\bar{s}(K+\ell-s)\mu L+T\epsilon_P\log(P)+\sum_{i=\bar{s}+1}^{\ell}(\ell-i+1)\mu L\nonumber\\&=sT\log(P)\Bigg[1+\frac{\log(c+1/P)}{\log(P)} +\frac{\epsilon_P}{s}\Bigg]+\Bigg[\bar{s}\bigg(K-s+\frac{\bar{s}-1}{2}\bigg)+\frac{\ell(\ell+1)}{2}\Bigg]\mu L+L\epsilon_L,
%\\&\stackrel{(d)}\leq sT\log(cP+1)+L\epsilon_L+\bar{s}(K+\ell-s)\mu L+To(\log(P))+\sum_{i=\bar{s}+1}^{\ell}\sum_{j=K+i}^{K+\ell}H\big(S_{i,j}\big)\nonumber\\&\stackrel{(e)}\leq sT\log(cP+1)+L\epsilon_L+\bar{s}(K+\ell-s)\mu L+To(\log(P))+\sum_{i=\bar{s}+1}^{\ell}(\ell-i+1)\mu L\nonumber\\&= sT\log(cP+1)+L\epsilon_L+To(\log(P))+\Bigg[\bar{s}\bigg(K-s+\frac{\bar{s}-1}{2}\bigg)+\frac{\ell(\ell+1)}{2}\Bigg]\mu L,
\end{flalign} 
where the steps in \eqref{eq:conv_term1} are explained as follows:
\begin{itemize}
	\item Step (a) follows from the fact that dropping the conditioning on the first term does not increase entropy. Further, we apply Fano's inequality to the third term. Hereby, $\epsilon_L$ is a function, independent of $P$ which vanishes in the limit as $L\rightarrow\infty$. \alert{The fourth and sixth term are zero, respectively, because $W_{[1:s]}\subseteq W_{[1:N]}$ and $S_{[1:\ell]}$ is a deterministic function of $W_{[1:N]}$.}    
	\item Step (b) follows by applying \cite[Lemma 1]{avik} on the first differential entropy term. Additionally, we use the chain rule to rewrite the fourth term. 
	\item Step (c) is due to the fact that the channel noise $\mathbf{z}^{T}_{u,[1:s]}$ is i.i.d. across time and has a Gaussian distribution with zero mean and unit variance. 
	\item \alert{Step (d) follows from the fact that in a discrete entropy of the form $H(S_{i}|\tilde{\mathcal{W}}),\tilde{\mathcal{W}}\subseteq\mathcal{W}$ with $\mathcal{W}=\{W_1,\ldots,W_N\}$, only the files $\mathcal{W}\setminus\tilde{\mathcal{W}}$ have non-zero contribution. This is because the cached content at the $i$-th RN of the $j$-th file is solely a function of $W_j$, i.e., $S_{i,j}=\phi_{i,j}(W_j)$.}
	\item Step (e) is since the entropy of each local file cache content is upper bounded according to $H(S_{i,j})\leq\mu L$.   
	\item In step (f), we use the fact that in the first channel use ($t=1$), the transmit signal $x_{r,m}[1]$ at RN$_m$ depends (apart from CSI) solely on the cached content $S_m$. Further, we note that the 
	%noise terms $\mathbf{z}^{T}_{u,[s+1:K]}$ and $\mathbf{z}^{T}_{r,[1:i-1]}$ 
	\alert{noise term $\mathbf{z}^{T}_{r,[1:i-1]}$ is independent of all the remaining random variables in the conditional entropy term.} 
	\item \alert{In step (g), we upper bound $H\big(W_{[s+1:K+i-1]},\mathbf{z}^{T}_{u,[1:s]}|W_{[1:s]\cup [K+\ell+1:N]},\mathbf{y}^{T}_{u,[1:s]}, S_{[1:i-1]},\mathbf{x}_{r,[1:i-1]}[1],$\newline $\mathbf{z}^{T}_{r,[1:i-1]}\big)$ by $T\epsilon_P\log(P)+L\epsilon_L$, where $\epsilon_P$ is any function in $P$ which satisfies $\lim_{P\rightarrow\infty}\epsilon_P=0$. This is because 
	\begin{align*}
	h(\mathbf{z}^{T}_{u,[1:s]}|W_{[1:s]\cup [K+\ell+1:N]},\mathbf{y}^{T}_{u,[1:s]}, S_{[1:i-1]},\mathbf{x}_{r,[1:i-1]}[1],\mathbf{z}^{T}_{r,[1:i-1]}\big))\leq h(\mathbf{z}^{T}_{u,[1:s]})=sT\log(2\pi e)%T\epsilon_P\log(P)\nonumber
	\end{align*} and 
	\begin{align*}
	H\big(W_{[s+1:K+i-1]}|W_{[1:s]\cup [K+\ell+1:N]},\mathbf{y}^{T}_{u,[1:s]}, \mathbf{z}^{T}_{u,[1:s]} S_{[1:i-1]},\mathbf{x}_{r,[1:i-1]}[1],\mathbf{z}^{T}_{r,[1:i-1]}\big)\leq L\epsilon_L
	\end{align*} due to Fano's inequality. In the nutshell, this bound states that the files $W_{[s+1:K+i-1]}$ can be resolved reliably when knowing $\mathbf{z}^{T}_{u,[1:s]}$ in addition to $S_{[1:i-1]}$ and $\mathbf{y}^{T}_{u,[1:s]}$ (if $i\geq M+2-s=\bar{s}+1$) (cf. Fig. \ref{fig:lw_bound}). This works since the Markov chain $\mathbf{z}^{T}_{u,[1:s]}\rightarrow\Big(\mathbf{y}^{T}_{u,[s+1:K]},\mathbf{x}^{T}_{s},S_{[1:i-1]},W_{[s+1:K]}\Big)\rightarrow W_{[K+1:K+i-1]}$ is applicable to the network.}   
%	\item Step (g) follows from \cite[Lemma 2]{avik}. In the nutshell, this lemma states that the files $W_{[s+1:K+i-1]}$ can be resolved from $S_{[1:i-1]}$ and $\mathbf{y}^{T}_{u,[1:s]}$ (if $i-1+s\geq M+1$) within bounded noise distortion $To(\log(P))$, where $\frac{o(\log(P))}{\log(P)}\rightarrow 0$ as $P\rightarrow\infty$. This works since the following Markov chains are applicable to the network:
%	\newline
%	\vspace{-2em}
%	\begin{enumerate}[leftmargin=3\parindent] 
%		\item[1a)] $S_{[1:i-1]}\rightarrow\mathbf{x}_{r,[1:i-1]}[1]$  
%		\item[1b)] $\big(S_{[1:i-1]},\mathbf{y}_{r,[1:i-1]}[t]\big)\rightarrow\mathbf{x}_{r,[1:i-1]}[t],\forall t\in[2:T]$ 
%		\item[2)] $\big(\mathbf{x}_{r,[1:i-1]}[t],\mathbf{y}_{u,[1:s]}[t],\mathbf{z}_{u,[1:K]}[t],\mathbf{z}_{r,[1:i-1]}[t]\big)\rightarrow\big(\mathbf{y}_{u,[s+1:K]}[t],\mathbf{y}_{r,[1:i-1]}[t]\big)$ if $i-1\geq\bar{s}\triangleq M+1-s,\forall t\in[T]$
%		\item[3a)]
%		$\mathbf{y}^{T}_{u,[s+1:K]}\rightarrow W_{[s+1:K]}$
%		\item[3b)]
%		$\Big(\mathbf{y}^{T}_{r,[1:i-1]},S_{[1:i-1]}\Big)\rightarrow W_{[K+1:K+i-1]}$   	
% \end{enumerate}	  
\end{itemize}
\alert{Now we consider the second term of \eqref{eq:conv_1}. Using the chain rule of discrete entropies for this term leads to
\begin{align}
\label{eq:conv_term2_intermediate_step1}
&H\big(W_{[1:K+\ell]}|\mathbf{y}^{T}_{u,[1:s]},S_{[1:\ell]},W_{[K+\ell+1:N]}\big)\nonumber\\&\qquad=H\big(W_{[1:s]}|\mathbf{y}^{T}_{u,[1:s]},S_{[1:\ell]},W_{[K+\ell+1:N]}\big)+H\big(W_{[s+1:K+\ell]}|\mathbf{y}^{T}_{u,[1:s]},S_{[1:\ell]},W_{[1:s]\cup[K+\ell+1:N]}\big)
\end{align} 
Next, we use the aforementioned steps $(a),(f)$ and $(g)$, respectively, to upper bound \eqref{eq:conv_term2_intermediate_step1}}.
\begin{flalign}
\label{eq:conv_term2}
&H\big(W_{[1:s]}|\mathbf{y}^{T}_{u,[1:s]},S_{[1:\ell]},W_{[K+\ell+1:N]}\big)+H\big(W_{[s+1:K+\ell]}|\mathbf{y}^{T}_{u,[1:s]},S_{[1:\ell]},W_{[1:s]\cup[K+\ell+1:N]}\big)\nonumber\\&\stackrel{(a)}\leq L\epsilon_L+H\big(W_{[s+1:K+\ell]}|\mathbf{y}^{T}_{u,[1:s]},S_{[1:\ell]},W_{[1:s]\cup[K+\ell+1:N]}\big)\nonumber\\&\stackrel{(f)}= L\epsilon_L+H\big(W_{[s+1:K+\ell]}|\mathbf{y}^{T}_{u,[1:s]},\mathbf{z}^{T}_{r,[1:\ell]},S_{[1:\ell]},\mathbf{x}_{r,[1:\ell]}[1],W_{[1:s]\cup[K+\ell+1:N]}\big)\nonumber\\&\leq L\epsilon_L+H\big(W_{[s+1:K+\ell]},\mathbf{z}^{T}_{u,[1:s]}|\mathbf{y}^{T}_{u,[1:s]},\mathbf{z}^{T}_{r,[1:\ell]},S_{[1:\ell]},\mathbf{x}_{r,[1:\ell]}[1],W_{[1:s]\cup[K+\ell+1:N]}\big)\nonumber\\&\stackrel{(g)}\leq L\epsilon_L+T\epsilon_P\log(P) 
\end{flalign} 
%\begin{flalign}
%\label{eq:conv_term2}
%&H\big(W_{[1:K+\ell]}|\mathbf{y}^{T}_{u,[1:s]},S_{[1:\ell]},W_{[K+\ell+1:N]}\big)\nonumber\\&=H\big(W_{[1:s]}|\mathbf{y}^{T}_{u,[1:s]}\big)+H\big(W_{[s+1:K]}|\mathbf{y}^{T}_{u,[1:s]},S_{[1:\ell]},W_{[1:s]\cup[K+\ell+1:N]}\big)\nonumber\\&\quad+H\big(W_{[K+1:K+\ell]}|\mathbf{y}^{T}_{u,[1:s]},S_{[1:\ell]},W_{[1:K]\cup[K+\ell+1:N]}\big)\nonumber\\&\stackrel{(a)}\leq L\epsilon_L+H\big(W_{[s+1:K]}|\mathbf{y}^{T}_{u,[1:s]},S_{[1:\ell]},W_{[1:s]\cup[K+\ell+1:N]}\big)+H\big(W_{[K+1:K+\ell]}|\mathbf{y}^{T}_{u,[1:s]},S_{[1:\ell]},W_{[1:K]\cup[K+\ell+1:N]}\big)\nonumber\\&\stackrel{(e)}= L\epsilon_L+H\big(W_{[s+1:K]}|\mathbf{y}^{T}_{u,[1:s]},\mathbf{z}^{T}_{u,[s+1:K]},\mathbf{z}^{T}_{r,[1:i-1]},S_{[1:\ell]},\mathbf{x}_{r,[1:\ell]}[1],W_{[1:s]\cup[K+\ell+1:N]}\big)\nonumber\\&\quad+H\big(W_{[K+1:K+\ell]}|\mathbf{y}^{T}_{u,[1:s]},\mathbf{z}^{T}_{u,[s+1:K]},\mathbf{z}^{T}_{r,[1:i-1]},S_{[1:\ell]},\mathbf{x}_{r,[1:\ell]}[1],W_{[1:K]\cup[K+\ell+1:N]}\big)\nonumber\\&\stackrel{(f)}\leq L\epsilon_L+T o(\log(P)) 
%\end{flalign} 
Combining \eqref{eq:conv_term1} and \eqref{eq:conv_term2} in \eqref{eq:conv_1}, we get the following inequality
\begin{align}
\label{eq:ineq_bound}
\hspace{-0.75em}(K+\ell)L\leq sT\log(P)\Bigg[1+\frac{\log(c+1/P)}{\log(P)} +\frac{\epsilon_P}{s}\Bigg]+\Bigg[\bar{s}\bigg(K-s+\frac{\bar{s}-1}{2}\bigg)+\frac{\ell(\ell+1)}{2}\Bigg]\mu L+L\epsilon_L.
\end{align}
Dividing both sides of \eqref{eq:ineq_bound} by $L$, letting $L\rightarrow\infty$ and $P\rightarrow\infty$, we obtain the following class of lower bounds on the minimum NDT
\begin{align}\label{eq:delta_lb_bound}
\delta^{\star}(\mu)\geq\delta_{\text{LB}}(\mu,\ell,s)\triangleq\frac{K+\ell-\mu(\bar{s}\big(K-s+\frac{(\bar{s}-1)}{2}\big)+\frac{\ell}{2}(\ell+1))}{s}.
\end{align} By optimizing the bound $\delta_{\text{LB}}(\mu,\ell,s)$ in \eqref{eq:delta_lb_bound} with respect to $\ell\in[\bar{s}:M]$ and $s\in[\min\{M+1,K\}]$, we determine one term of the NDT lower bound in \eqref{eq:NDT_lw_bound}. Additionally, we recall that the NDT is also bounded from below by the performance of the reference interference-free system which has an NDT of $1$. The maximum over these two lower bounds concludes the proof of Theorem \ref{theorem_lower_bound}. \alert{We note that the lower bound simplifies to $\delta_{\text{LB}}(0,M,1)=K+M$ for $\mu=0$, while it reduces to $\max\big\{1,\delta_{\text{LB}}(1,0,M+1)\big\}=\max\big\{1,\nicefrac{K}{(M+1)}\big\}$ for $\mu=1$. These NDT lower bounds coincide with the upper bounds one would intuitively assume to be optimal in the worst-case scenario. These are respectively, unicasting $K+M$ files for $\mu=0$ from the DeNB and zero-forcing beamforming of $K$ files to $K$ UEs from $M+1$ identical transmitters ($M$ RNs and DeNB) at $\mu=1$.}

\section{NDT-One Shot Scheme: \alert{Synergistic Integration} of Multicasting and Zero-Forcing Schemes}
\label{cha_one_shot}

We now develop a general \emph{one-shot} scheme. We refer to a scheme to be one-shot if all receiving nodes are able to decode their \emph{desired} symbols on a \emph{single} channel use basis. Such schemes explicitly preclude symbol decoding over multiple channel uses. % In each channel use, 
Hereby, our one-shot scheme uses a combination of the Maddah-Ali Niesen (MAN) scheme \cite{Maddah-Ali2} and ZF to simultaneously convey a subset of RNs and UEs with their desired symbols. The achievability for the extreme cases of zero-cache ($\mu=0$) and full-cache ($\mu=1$) are based on one-shot schemes (These are, respectively, unicasting and zero-forcing.) and readily follow from Lemma \ref{corr_mu_0_and_1}. Thus in the sequel, we consider fractional cache sizes that are strictly larger than zero and strictly less than one. Henceforth, we explicitly assume $M\geq 2$.  

To this end, let us consider the achievability at fractional cache sizes
$\mu\in\{\nicefrac{1}{M},\nicefrac{2}{M},\nicefrac{3}{M},\ldots,\nicefrac{(M-1)}{M}\}$. We split each file $W_n,\forall n\in[N],$ into $\Gamma{{M}\choose{\mu M}}$ symbols, where $\Gamma={{K}\choose{\psi}}$ and $\psi=\min\{K,\mu M\}$. Symbols of every file are labeled according to
\begin{equation*}
W_n=\big(\eta_{n,\mathcal{T},\mathcal{U}}:\mathcal{T}\subset [M],|\mathcal{T}|=\mu M,\mathcal{U}\subset [K],|\mathcal{U}|=\psi\big).
\end{equation*} In the placement phase, for each $n\in[N]$, symbol $\eta_{n,\mathcal{T},\mathcal{U}}$ is prefetched at $\text{RN}_m$ if $m\in\mathcal{T}$. Thus, each RN caches a total of $N\Gamma{{M-1}\choose{\mu M-1}}$ symbols. It is easy to see that the cache constraint at each RN is satisfied since \begin{equation*}
\frac{\#\text{ of cached symbols}}{\#\text{ of total symbols}}=\frac{N\Gamma{{M-1}\choose{\mu M-1}}}{N\Gamma{{M}\choose{\mu M}}}=\mu.
\end{equation*} 
Next, we describe the delivery phase. Consider for the sake of notational simplicity, the worst-case request vector %\begin{equation*}
%\big(\underbrace{1,2,\ldots,K}_{\text{UEs requested files}},\underbrace{K+1,K+2,\ldots,K+M}_{\text{RNs requested files}}\big)
%\end{equation*}
\begin{equation*}
\big(1,2,\ldots,K,K+1,K+2,\ldots,K+M\big)
\end{equation*} for $N\geq K+M$. At each time instant $t$, we focus on the delivery of desired symbols for a subset $\mathcal{S}_R\subset[M]$ of $|\mathcal{S}_R|=1+\mu M$ RNs and $\mathcal{S}_U\subset[K]$ of $|\mathcal{S}_U|=\psi$ UEs. Hereby, the UEs are served by a subset of $|\mathcal{S}_R'|=\psi$ RNs belonging to the set $\mathcal{S}_R'\subset[M]$. To this end, DeNB and RN transmit the following signals:
\begin{align}\label{eq:Tx_sig_DeNB_one_shot}
x_s[t]=&\sum_{m\in\mathcal{S}_R}\nu_{\eta_{K+m,\mathcal{S}_R\setminus\{m\},\mathcal{S}_U}}[t]\eta_{K+m,\mathcal{S}_R\setminus\{m\},\mathcal{S}_U}\\\label{eq:Tx_sig_RN_one_shot} x_{r,m'}[t]=&\begin{cases}\sum_{\substack{m\in\mathcal{S}_R\\m\neq m'}}\beta^{(m')}_{\eta_{K+m,\mathcal{S}_R\setminus\{m\},\mathcal{S}_U}}[t]\eta_{K+m,\mathcal{S}_R\setminus\{m\},\mathcal{S}_U}+\sum_{\substack{p\in \mathcal{S}_U}}\beta^{(m')}_{\eta_{p,\mathcal{S}_R',\mathcal{S}_U}}[t]\eta_{p,\mathcal{S}_R',\mathcal{S}_U}\qquad&\text{ if }m'\in\big(\mathcal{S}_R\cap\mathcal{S}_R'\big)\\ \sum_{\substack{m\in\mathcal{S}_R\\m\neq m'}}\beta^{(m')}_{\eta_{K+m,\mathcal{S}_R\setminus\{m\},\mathcal{S}_U}}[t]\eta_{K+m,\mathcal{S}_R\setminus\{m\},\mathcal{S}_U}&\text{ if }m'\in\big(\mathcal{S}_R\setminus\mathcal{S}_R'\big)\\ \sum_{\substack{p\in \mathcal{S}_U}}\beta^{(m')}_{\eta_{p,\mathcal{S}_R',\mathcal{S}_U}}[t]\eta_{p,\mathcal{S}_R',\mathcal{S}_U}&\text{ if }m'\in\big(\mathcal{S}_R'\setminus\mathcal{S}_R\big)\\ 0&\text{ otherwise}\end{cases}
\end{align} In these equations, the precoders for symbol $\eta_{n,\mathcal{T},\mathcal{S}_U}$ originating from the DeNB and $\text{RN}_m$ are denoted by $\nu_{\eta_{n,\mathcal{T},\mathcal{S}_U}}$ and $\beta^{(m)}_{\eta_{n,\mathcal{T},\mathcal{S}_U}}$, respectively. Decoding at the RNs in $\mathcal{S}_R$ follows along the standard MAN manner. That is, each RN $m'\in\mathcal{S}_R$ exploits its knowledge of symbols $\eta_{K+m,\mathcal{S}_R\setminus\{m\},\mathcal{S}_U}$ for all $m\in\mathcal{S}_R\setminus\{m'\}$ to recover its desired symbol $\eta_{K+m',\mathcal{S}_R\setminus\{m'\},\mathcal{S}_U}$. Thus, we shift our focus to the UEs. The received signal at $\text{UE}_k$ is specified by the equation
\begin{align}
\label{eq:Rx_sig_UEk_one_shot}y_{u,k}[t]=&\sum_{m\in\mathcal{S}_R}\eta_{K+m,\mathcal{S}_R\setminus\{m\},\mathcal{S}_U}\Big(g_k[t]\nu_{\eta_{K+m,\mathcal{S}_R\setminus\{m\},\mathcal{S}_U}}[t]+\sum_{\substack{m'\in\mathcal{S}_R\\m'\neq m}}h_{km'}[t]\beta^{(m')}_{\eta_{K+m,\mathcal{S}_R\setminus\{m\},\mathcal{S}_U}}[t]\Big)+\nonumber\\&+\sum_{p\in\mathcal{S}_U}\sum_{m'\in\mathcal{S}_R'}h_{km'}[t]\beta^{(m')}_{\eta_{p,\mathcal{S}_R',\mathcal{S}_U}}[t]\eta_{p,\mathcal{S}_R',\mathcal{S}_U}\\=&\sum_{m\in\mathcal{S}_R}\eta_{K+m,\mathcal{S}_R\setminus\{m\},\mathcal{S}_U}[t]\:\big(g_k[t],\mathbf{h}^{\dagger}_{k,\mathcal{S}_R\setminus\{m\}}[t]\big)\:\big(\nu_{\eta_{K+m,\mathcal{S}_R\setminus\{m\},\mathcal{S}_U}}[t],\boldsymbol{\beta}^{\dagger}_{\eta_{K+m,\mathcal{S}_R\setminus\{m\},\mathcal{S}_U}}[t]\big)^{\dagger}+\nonumber\\\label{eq:Rx_sig_UEk_vec_form_one_shot}&+\sum_{p\in\mathcal{S}_U}\mathbf{h}^{\dagger}_{k,\mathcal{S}_R'}[t]\boldsymbol{\beta}_{\eta_{p,\mathcal{S}_R',\mathcal{S}_U}}[t]\eta_{p,\mathcal{S}_R',\mathcal{S}_U}.
\end{align} This equation can be rewritten in a compact form by making the following definitions. 

First, we define the collection of \emph{channel coefficients} from RNs in the set $\mathcal{W}$ to $\text{UE}_k$ as the \emph{vector}
\begin{equation*}
\mathbf{h}_{k,\mathcal{W}}[t]\triangleq\{h_{km}[t]\}_{m\in\mathcal{W}},\quad\mathbf{h}_{k,\mathcal{W}}[t]\in\mathbb{C}^{|\mathcal{W}|}.
\end{equation*} Similarly, channel coefficients from RNs in the set $\mathcal{W}$ to UEs in the set $\mathcal{U}$ are denoted by the channel \emph{matrix}    
\begin{equation*}
\mathbf{H}_{\mathcal{U},\mathcal{W}}[t]\triangleq\{\mathbf{h}_{k,\mathcal{W}}^{\dagger}[t]\}_{k\in\mathcal{U}},\quad\mathbf{H}_{\mathcal{U},\mathcal{W}}[t]\in\mathbb{C}^{|\mathcal{U}|\times|\mathcal{W}|}.
\end{equation*} Second, we concatenate the precoders of the RNs for symbol $\eta_{n,\mathcal{T},\mathcal{U}}$ to the \emph{vector}
\begin{equation*}
\boldsymbol{\beta}^{\mathcal{B}}_{\eta_{n,\mathcal{T},\mathcal{U}}}[t]\triangleq\big\{\beta^{(m)}_{\eta_{n,\mathcal{T},\mathcal{U}}}[t]\big\}_{m\in\mathcal{B}},\quad\boldsymbol{\beta}^{\mathcal{B}}_{\eta_{n,\mathcal{T},\mathcal{U}}}[t]\in\mathbb{C}^{|\mathcal{B}|}.
\end{equation*} When $\mathcal{B}=\mathcal{T}$, we simply write $\boldsymbol{\beta}{\eta_{n,\mathcal{T},\mathcal{U}}}[t]$ instead of $\boldsymbol{\beta}^{\mathcal{B}}_{\eta_{n,\mathcal{T},\mathcal{U}}}[t]$. These definitions are also applicable to the channel coefficients $g_k[t]$ as well. They allow us to rewrite Eq. \eqref{eq:Rx_sig_UEk_one_shot} to \eqref{eq:Rx_sig_UEk_vec_form_one_shot}. Recall that $\text{UE}_k$ is only provided with its desired symbol $\eta_{k,\mathcal{S}_R',\mathcal{S}_U}$ as long as $k\in\mathcal{S}_U$. To this end, \emph{all} interferences in the concatenated vector $\mathbf{y}_{u,\mathcal{S}_U}[t]=\{y_{u,k}\}_{k\in\mathcal{S}_U}$ given by \begin{align}\mathbf{y}_{u,\mathcal{S}_U}[t]=&\underbrace{\sum_{m\in\mathcal{S}_R}\eta_{K+m,\mathcal{S}_R\setminus\{m\},\mathcal{S}_U}[t]\:\big(\mathbf{g}_{\mathcal{S}_U}[t],\mathbf{H}_{\mathcal{S}_U,\mathcal{S}_R\setminus\{m\}}[t]\big)\:\big(\nu_{\eta_{K+m,\mathcal{S}_R\setminus\{m\},\mathcal{S}_U}}[t],\boldsymbol{\beta}^{\dagger}_{\eta_{K+m,\mathcal{S}_R\setminus\{m\},\mathcal{S}_U}}[t]\big)^{\dagger}}_{\text{Interferences from DeNB}}+\nonumber\\\label{eq:Rx_sig_UEs_matvecform_one_shot}&+\underbrace{\sum_{p\in\mathcal{S}_U}\mathbf{H}_{\mathcal{S}_U,\mathcal{S}_R'}[t]\boldsymbol{\beta}_{\eta_{p,\mathcal{S}_R',\mathcal{S}_U}}[t]\eta_{p,\mathcal{S}_R',\mathcal{S}_U}}_{\substack{{\psi-1\text{ interferences and 1 desired }}\\\text{ component for UE}_k,k\in\mathcal{S}_U}}
\end{align} have to be zero-forced. This is equivalent to
\begin{equation}\label{eq:first_ZF_condition_one_shot_scheme}
\big(\mathbf{g}_{\mathcal{S}_U}[t],\mathbf{H}_{\mathcal{S}_U,\mathcal{S}_R\setminus\{m\}}[t]\big)\boldsymbol{\tilde{\beta}}_{\eta_{K+m,\mathcal{S}_R\setminus\{m\},\mathcal{S}_U}}[t]=\boldsymbol{0}_{|\mathcal{S}_U|},
\end{equation} $\forall m\in\mathcal{S}_R$ and
\begin{equation}
\label{eq:second_ZF_condition_one_shot_scheme}
\mathbf{H}_{\mathcal{S}_U\setminus\{p\},\mathcal{S}_R'}[t]\boldsymbol{\beta}_{\eta_{p,\mathcal{S}_R',\mathcal{S}_U}}[t]=\boldsymbol{0}_{|\mathcal{S}_U|-1},
\end{equation} $\forall p\in\mathcal{S}_U$. Note that we used \begin{align*}
\boldsymbol{\tilde{\beta}}_{\eta_{K+m,\mathcal{S}_R\setminus\{m\},\mathcal{S}_U}}[t]&=\big(\nu_{\eta_{K+m,\mathcal{S}_R\setminus\{m\},\mathcal{S}_U}}[t],\boldsymbol{\beta}^{\dagger}_{\eta_{K+m,\mathcal{S}_R\setminus\{m\},\mathcal{S}_U}}[t]\big)^{\dagger}
\end{align*} in Eq. \eqref{eq:first_ZF_condition_one_shot_scheme} for reasons of compactness. It is easy to see that
\begin{subequations}
	\begin{alignat}{2}
	\rank\Big(\big(\mathbf{g}_{\mathcal{S}_U}[t],\mathbf{H}_{\mathcal{S}_U,\mathcal{S}_R\setminus\{m\}}[t]\big)\Big)&=\psi,\\
	\rank\Big(\mathbf{H}_{\mathcal{S}_U\setminus\{p\},\mathcal{S}_R'}[t]\Big)&=\psi-1,
	\end{alignat}
\end{subequations}$\forall m\in\mathcal{S}_R,\forall p\in\mathcal{S}_U.$ The nullspace dimension for these two matrices thus become \begin{subequations}
\begin{alignat}{2}
|\mathcal{S}_R|-\psi&=1+(\mu M-K)^{+},\\
|\mathcal{S}_R'|-(\psi-1)&=1\text{ for }\psi\geq 2. 
\end{alignat}
\end{subequations} (When $\psi=1$, we note that there are no interference terms ($\psi-1=0$) in the second sum of Eq. \eqref{eq:Rx_sig_UEs_matvecform_one_shot} which makes zero-forcing for this component obsolete.) We choose the precoding vectors in Eqs. \eqref{eq:first_ZF_condition_one_shot_scheme} and \eqref{eq:second_ZF_condition_one_shot_scheme} such that   
\begin{subequations}
	\begin{alignat}{2}
	\boldsymbol{\tilde{\beta}}_{\eta_{K+m,\mathcal{S}_R\setminus\{m\},\mathcal{S}_U}}[t]&\in\mathcal{N}\Big(\big(\mathbf{g}_{\mathcal{S}_U}[t],\mathbf{H}_{\mathcal{S}_U,\mathcal{S}_R\setminus\{m\}}[t]\big)\Big),\\
	\boldsymbol{\beta}_{\eta_{p,\mathcal{S}_R',\mathcal{S}_U}}[t]&\in\mathcal{N}\Big(\mathbf{H}_{\mathcal{S}_U\setminus\{p\},\mathcal{S}_R'}[t]\Big), 
	\end{alignat}
\end{subequations} where $\mathcal{N}(\mathbf{A})$ denotes the (right) nullspace of $\mathbf{A}$. Consequently, all UEs in the subset $\mathcal{S}_U$ will be free from interference. We conclude that in a single channel use $1+\mu M$ RNs in $\mathcal{S}_R$ and $\psi$ UEs in $\mathcal{S}_U$ were able to decode their desired symbols through a combination of ZF and the MAN scheme.

Assume that we deploy this scheme for 
\begin{equation*}
T_1=\Gamma{{M}\choose{1+\mu M}}
\end{equation*} channel uses. The probability that $\text{RN}_m$ is served at the $t$-th channel use, $t\in[T_1]$, with its desired symbol $\eta_{K+m,\mathcal{S}_R\setminus\{m\},\mathcal{S}_U}$ is determined whether $m\in\mathcal{S}_R$. In $T_1$ channel uses this happens
$\nicefrac{{{M-1}\choose{\mu M}}}{{{M}\choose{1+\mu M}}}$-fraction\footnote{This fraction can be interpreted as the probability $p(m\in\mathcal{S}_R)$.} of the time $T_1$. Thus, $\text{RN}_m$, $\forall m\in[M]$, receives
\begin{equation*}
N_{\text{RN}}^{\bar{c}}=T_1\frac{{{M-1}\choose{\mu M}}}{{{M}\choose{1+\mu M}}}=\Gamma{{M-1}\choose{\mu M}}
\end{equation*} \emph{uncached} symbols $\eta_{K+m,\mathcal{S}_R\setminus\{m\},\mathcal{S}_U}$ of the requested file $W_{K+m}$. Recall that $\text{RN}_m$ has also
\begin{equation*}
N_{\text{RN}}^{c}=\Gamma{{M-1}\choose{\mu M-1}}
\end{equation*} symbols $\eta_{K+m,\mathcal{S}_R\setminus\{m'\},\mathcal{S}_U}$, $m'\in\mathcal{S}_R\setminus\{m\}$, of the requested file $W_{K+m}$ available it its cache. Using uncached and cached symbols, $\text{RN}_m$ can reconstruct file $W_{K+m}$ of size 
\begin{equation*}
|W_{K+m}|=N_{\text{RN}}^{\bar{c}}+N_{\text{RN}}^{c}=\Gamma{{M}\choose{\mu M}}
\end{equation*} symbols. With respect to the RNs, we conclude that $T_1$ channel uses suffice to allow them to retrieve their desired files $W_{K+m},\forall m\in[M]$. We now shift our focus to the $K$ UEs. In $\widetilde{T}_1\leq T_1$ channel uses, where \begin{equation*}
\widetilde{T}_1=\min\Bigg\{T_1,\Gamma\frac{K}{\psi}{{M}\choose{\mu M}}\Bigg\},
\end{equation*}		
we try to provide each UE with the \emph{same} number of desired symbols. To this end, $\text{UE}_k$ obtains its desired symbols $\eta_{k,\mathcal{S}_R',\mathcal{U}}$ $\nicefrac{{{K-1}\choose{\psi-1}}}{{{K}\choose{\psi}}}$-fraction of time duration $\widetilde{T}_1$. Hereby, $\text{RN}_m,\forall m\in[M]$, is used
\begin{equation*}
N_{\text{RN,Tx}}=\widetilde{T}_1\frac{{{M-1}\choose{\psi-1}}}{{{M}\choose{\psi}}}=\Gamma\min\Bigg\{\frac{\psi}{M}{{M}\choose{1+\mu M}},\frac{K}{M}{{M}\choose{\mu M}}\Bigg\}
\end{equation*} times for transmission\footnote{If $N_{\text{RN,Tx}}$ is not an integer, file symbols require further fragmentation and rate splitting ought to be applied.} of these symbols such that every UE receives
\begin{equation}
\label{eq:num_sym_UE_after_T1_cus}
N_{\text{UE}}=\widetilde{T}_1\frac{{{K-1}\choose{\psi-1}}}{{{K}\choose{\psi}}}=\min\Bigg\{{{M}\choose{\mu M+1}}{{K-1}\choose{\psi-1}},\Gamma{{M}\choose{\mu M}}\Bigg\}
\end{equation} symbols. Depending on whether
\begin{subequations}
	\label{eq:num_sym_UE_overall_comparison}
	\begin{equation}
	\label{eq:num_sym_UE_comparison}
	|W_k|-N_{\text{UE}}\geqq 0 ,
	\end{equation} or 
	\begin{equation}
	\frac{K}{\psi}=\frac{{{K}\choose{\psi}}}{{{K-1}\choose{\psi-1}}}\greatleq\frac{{{M}\choose{1+\mu M}}}{{{M}\choose{\mu M}}}=\frac{M-\mu M}{1+\mu M},
	\end{equation}
\end{subequations} $\widetilde{T}_1$ channel uses are sufficient or insufficient for the delivery of files $W_k$ to $\text{UE}_k,\forall k\in[K]$. We state the following conditions (see also Fig. \ref{fig:Region_plot}) with respect to the inequalities in \eqref{eq:num_sym_UE_overall_comparison}. 
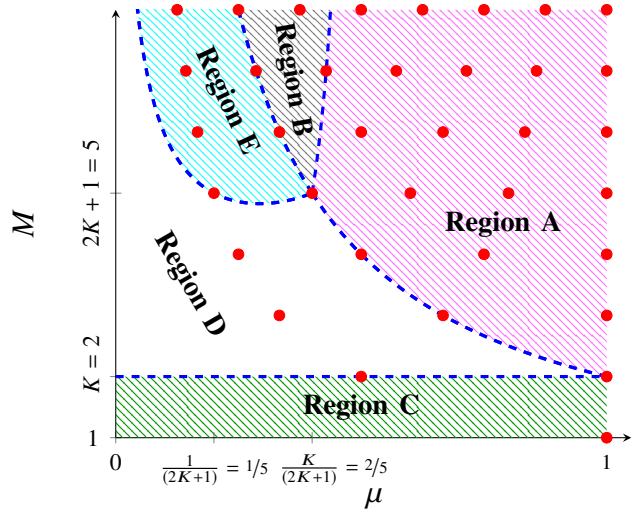
\begin{figure}
	\centering
	\begin{tikzpicture}[scale=1]
	\PlotMMuRegions
	\end{tikzpicture}
	\caption[Plot of all regions]{\small 2D $(\mu,M)$-plot of all regions described by inequalities \eqref{eq:condition_1_T1_insufficient}--\eqref{eq:condition_2_T1_sufficient} for $K=2$. Specifically, the regions described by inequalities \eqref{eq:condition_1_T1_insufficient},\eqref{eq:condition_1_T1_sufficient} and \eqref{eq:condition_2_T1_sufficient} are termed Region $A$, $B$ and $E$, respectively. The two inequalities in \eqref{eq:condition_2_T1_insufficient} are represented by Region $C$ and $D$. The discrete points illustrate the fractional cache sizes $\mu=\nicefrac{m}{M}$, $m\in[M]$ for an integer $M$.}
	\label{fig:Region_plot}
\end{figure} First, when $\psi=K$, $|W_k|>N_{\text{UE}}$ holds if
\begin{equation}\label{eq:condition_1_T1_insufficient}
\begin{cases}
K\leq\mu M<M<\nicefrac{1}{(1-2\mu)},\:\mu\leq\nicefrac{1}{2}\quad&(\text{Region A}_1)\\K\leq\mu M\leq M,\:M>\nicefrac{1}{(1-2\mu)},\:\mu>\nicefrac{1}{2}\quad&(\text{Region A}_2)
\end{cases},
\end{equation} while $|W_k|=N_{\text{UE}}$ is satisfied if
\begin{equation}\label{eq:condition_1_T1_sufficient}K\leq\mu M,\nicefrac{1}{(1-2\mu)}\leq M,\:\mu\leq\nicefrac{1}{2}\quad\text{(Region B)}.
\end{equation} We denote the union of regions A$_1$ and A$_2$ as A. Second, when $\psi=\mu M$, $|W_k|> N_{\text{UE}}$ is met if \begin{align}\label{eq:condition_2_T1_insufficient}
\begin{cases}
\mu M< M<K\quad&\text{(Region C)}\\ \max\Big\{\frac{\mu M^{2}(1-\mu)}{1+\mu M},\mu M\Big\}<K\leq M\quad&\text{(Region D)}
\end{cases},
\end{align} 
whereas on the other hand $|W_k|=N_{\text{UE}}$ if \begin{equation}\label{eq:condition_2_T1_sufficient}
\mu M<K\leq\frac{\mu M^{2}(1-\mu)}{1+\mu M}\leq M\quad\text{(Region E)}.
\end{equation} 
Thus, for cases where $|W_k|= N_{\text{UE}}$ (under conditions \eqref{eq:condition_1_T1_sufficient} or \eqref{eq:condition_2_T1_sufficient}), $T_1$ channel uses suffice to meet the demands of all $M$ RNs and $K$ UEs. The achievable NDT under theses cases results in the NDT of the MAN scheme given by
\begin{equation}\label{eq:sufficing_T1_achievable_NDT}
\delta_{\text{MAN}}(\mu)=M\cdot(1-\mu)\cdot\frac{1}{1+\mu M}
\end{equation} under conditions \eqref{eq:condition_1_T1_sufficient} or \eqref{eq:condition_2_T1_sufficient} at fractional cache sizes $\mu\in\{\nicefrac{1}{M},\nicefrac{2}{M},\ldots,\nicefrac{(M-1)}{M}\}$. 
For the cases where $|W_k|>N_{\text{UE}}$ (i.e., \eqref{eq:condition_1_T1_insufficient} or \eqref{eq:condition_2_T1_insufficient}), however, $T_2$ \emph{additional} channel uses are required to convey each $\text{UE}$ with its remaining $|W_k|-N_{\text{UE}}$ desired symbols. Hereby, in each channel use, $\psi'=\min\{K,1+\mu M\}$ UEs can be provided with their desired symbols through applying ZF beamforming for a $(\psi',K)$ MISO broadcast channel such that
\begin{equation}\label{eq:time_T2}
T_2=\frac{K\big(|W_k|-N_{\text{UE}}\big)}{\psi'}.
\end{equation} As opposed to the first block of $T_1$ channel uses, both RNs and DeNB are now involved in providing the UEs with their desired symbols. This is due to the fact that the delivery of the RNs requested files is terminated. The achievable NDT for the cases \eqref{eq:condition_1_T1_insufficient} and \eqref{eq:condition_2_T1_insufficient} are given by
\begin{equation*}
\delta(\mu)=\frac{T_1+T_2}{\Gamma{{M}\choose{\psi}}}.
\end{equation*} Simplification of this term leads to the following expression 
\begin{align}\label{eq:not_sufficing_T1_achievable_NDT}
\delta(\mu)=\begin{cases}
\delta_{\text{MISO-BC}}=1\qquad&\text{ for }\eqref{eq:condition_1_T1_insufficient}\\\frac{K+\delta_{\text{MAN}}(\mu)}{\psi'}\qquad&\text{ for }\eqref{eq:condition_2_T1_insufficient}
\end{cases},
\end{align} where $\delta_{\text{MAN}}(\mu)$ is the NDT expression given in Eq. \eqref{eq:sufficing_T1_achievable_NDT} and $\delta_{\text{MISO-BC}}$ is the NDT of the ($\psi,K$) MISO BC if $\psi=K$. Combining the two NDT expressions in Eqs. \eqref{eq:sufficing_T1_achievable_NDT} and \eqref{eq:not_sufficing_T1_achievable_NDT} to a single one generates the achievable one-shot NDT $\delta_{\text{OS}}(\mu)$ of Theorem \ref{th:one_shot_ach_NDT}.        
\section{NDT-Optimal Schemes for Special Instances: Integration of Subspace Alignment, Multicasting and Zero-Forcing}
\label{sec:ach_special}

In this section, we present our novel achievability schemes that combine the well-known schemes -- interference alignment, multicasting and zero-forcing -- for special instances of $M\in\{1,2,3\}$ and $K\in\{1,2,3\}$ that satisfy $K+M\leq 4$. For these special instances, these schemes allow us to fully characterize the optimal NDT cache-memory tradeoff for any $\mu\in[0,1]$. 

%Now, we move to the aforementioned special cases of the system, where $M\in\{1,2,3\}$, to provide a \emph{complete} characterization of the NDT-memory trade-off. 

We start with the case where $M=1$.  

\subsection{Achievability for $M=1$}\label{subsec:ach_M1}

The NDT-optimal scheme for $M=1$ and $K\leq 3$ is presented. Hereby, the following proposition quantifies the achievable NDT. 

\begin{proposition}\label{prop_M1}
	The achievable NDT of the network under study for $M=1$ RNs, $K\in\{1,2,3\}$ UEs and $\mu\in[0,1]$ is given by 
	\begin{equation}
	\label{eq:ach_NDT_M_1_K_123}
	\delta(\mu)=\begin{cases}K+1-\mu K\quad&\text{ if } K\leq 2\\\max\bigg\{K+1-\mu K,\frac{K+1-\mu}{2}\bigg\}\quad&\text{ if }K=3\end{cases}.
	\end{equation}
\end{proposition}
In what follows the proof of this proposition is outlined. As shown in Fig. \ref{fig:NDT_M_K_1_LEQ_2}, in the case, where $K\leq 2$, the lower bound \alert{$\delta_{\text{LB}}(\mu,1,1)$} coincides with the achievable NDT. Lemma \ref{corr_mu_0_and_1} establishes the achievability for the corner points %(marked by circles in Fig. \ref{fig:NDT_M_K_1_LEQ_2}) 
at $\mu=0$ and $\mu=1$. Hereby, due to arguments of \emph{memory sharing}, \alert{intermediate} points at fractional cache sizes $0<\mu<1$ become achievable through successively time-sharing between unicasting and zero-forcing for $\mu$ and $(1-\mu)$ file fractions, respectively.

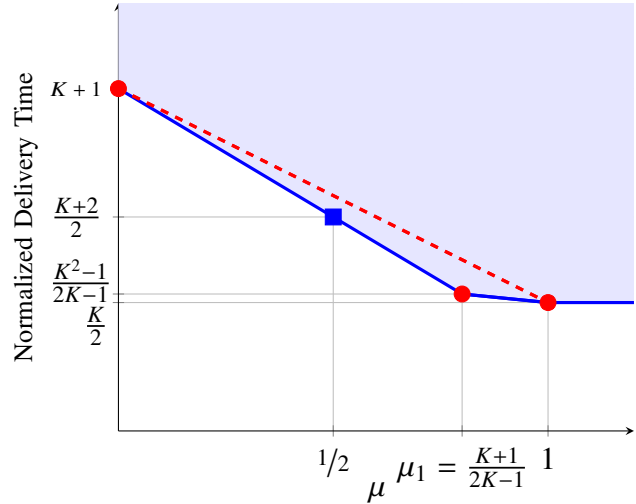
\begin{figure}
	\centering
	\begin{tikzpicture}[scale=1]
	\PlotNDTMOneKGreatTwo
	\end{tikzpicture}
	\caption[NDT lower bound for $M=1$ and $K\geq 3$]{\small The NDT lower bound for $M=1$ and $K\geq 3$ is shown by the solid line. For $K=3$, this line is in fact achievable. The dashed line shows the achievable NDT of a time-sharing based unicasting-zero-forcing scheme.} 
	\label{fig:NDT_M_K_1_3}
\end{figure}   

However, for $K=3$ such type of memory sharing scheme is in fact \emph{suboptimal}. \alert{This is indicated when comparing the achievable NDT of memory sharing between the zero-cache and full-cache scheme with the coinciding lower bound for $K=3$ illustrated in Fig. \ref{fig:NDT_M_K_1_3}. In fact,} we will show that \emph{simultaneously} combining subspace alignment with zero-forcing through appropriate precoder design allows us to show the achievability of the corner point $\Big(\frac{K+1}{2K-1},\frac{K^{2}-1}{2K-1}\Big)=\Big(\frac{4}{5},\frac{8}{5}\Big)$ for $K=3$. To this end, we describe, respectively, \begin{enumerate*}[label=(\roman*)]
	\item the \emph{RN cache placement} for $\mu=\nicefrac{4}{5}$,
	\item the \emph{encoding} at DeNB and RN ($M=1$),
	\item and finally the \emph{decoding} at RN and the three UEs ($K=3$).
\end{enumerate*} We begin with the cache placement.  
\vspace{.5em}
\subsubsection*{RN Cache Placement}

\begin{figure}
	\FilesKThree
	\caption[Requested files by $K=3$ users and $M=1$ RNs]{\small \alert{Requested files by $K=3$ users and $M=1$ RNs and the availability exclusively at the DeNB or both at the DeNB and the RN. The highlighted symbols are zero-forcing symbols, i.e., the (spatial) availability at both the DeNB and the RN is used to zero-force these symbols at one of the $K=3$ UEs.}}
	\label{fig:files_K_three}
\end{figure}

Assume without loss of generality \alert{one possible worst-case demand scenario for $N=4$. That is,} the UEs request files $W_1,W_2$ and $W_3$ while the RN is interested in file $W_4$. According to Fig. \ref{fig:files_K_three}, all files are \alert{broken into parts of $5$ symbols each.} \alert{For the sake of simplicity, it suffices to focus on single file parts. Thus, the remaining discussion focuses on $5$ symbols per file denoted by $\eta_{i,1},\eta_{i,2},\eta_{i,3},\eta_{i,4}$ and $\eta_{i,5}$ with respect to the $i$-th file. For notational simplicity, we stack these symbols to the vector $\boldsymbol{\eta}_i=(\eta_{i,1},\eta_{i,2},\ldots,\eta_{i,5})^{\dagger}$. The vector $\boldsymbol{\eta}$ on the other hand, concatenates all $\boldsymbol{\eta}_i,i=1,\ldots,4,$ to a single column vector, i.e., $\boldsymbol{\eta}=\big(\boldsymbol{\eta}^{\dagger}_{1},\boldsymbol{\eta}^{\dagger}_{2},\boldsymbol{\eta}^{\dagger}_{3},\boldsymbol{\eta}^{\dagger}_{4}\big)^{\dagger}$.} \emph{All} symbols \alert{in $\boldsymbol{\eta}$} are available at the DeNB. However, as far as the RN is concerned, only the first four symbols of all files are locally available in its cache. Since, the RN is interested in file $W_4$ and it knows $\eta_{4,1},\eta_{4,2},\eta_{4,3}$ and $\eta_{4,4}$ \alert{already}, the only missing symbol it desires is $\eta_{4,5}$. Thus, the transmission policy has to be designed \alert{such that \emph{all} symbols of files $W_1,W_2$ and $W_3$ as well as $\eta_{4,5}$ are conveyed jointly or exclusively by DeNB and/or RN.} These are in total $16$ information symbols. Next, we describe the encoding strategy for both DeNB and RN. 
\vspace{.5em}
\subsubsection*{Encoding at DeNB and RN}

\begin{figure}[h]
	\ZFKThree
	\caption[Zero-forcing map]{\small Map that assigns which symbol to zero-force at which UE.} 
	\label{fig:ZF_K_three}
\end{figure} 

The transmission strategy will exploit the correlation that arises between the availability of shared symbols at RN and DeNB by leveraging zero-forcing (ZF) opportunities while \emph{simultaneously} facilitating (subspace) interference alignment (IA) at the UEs. This is why our scheme (as shown in Fig. \ref{fig:files_K_three}) only zero-forces symbols $\eta_{1,1},\eta_{1,2},\eta_{1,3}$, $\eta_{2,1},\eta_{2,2},\eta_{2,3}$ and $\eta_{3,1},\eta_{3,2},\eta_{3,3}$. Symbols $\eta_{1,4},\eta_{2,4}$ and $\eta_{3,4}$ are \emph{not} zero-forced but are instead used to enable alignment\footnote{IA is facilitated by the fact that the DeNB does \emph{not} transmit these symbols (even though it knows them). Thus, effectively, the DeNB does not need to be aware of $\eta_{i,4},i\in[N]$. \alert{This is accounted for in Eq. \eqref{eq:some_precod_elements2} for $N=4$.}} amongst others with $\eta_{4,5}$ at the UEs. The map that assigns which symbol is zero-forced at which UE is given in Figure \ref{fig:ZF_K_three}. To this end, DeNB and RN form their transmit signals according to
\alert{
\begin{align}
\label{eq:tx_sig_DeNB}
x_s[t]&=\boldsymbol{\nu}^{\dagger}[t]\boldsymbol{\eta} \\
\label{eq:tx_sig_RN}
x_r[t]&=\boldsymbol{\beta}^{\dagger}[t]\boldsymbol{C}_{\text{RN}}\boldsymbol{\eta}
\end{align}}
%\begin{align}
%\label{eq:tx_sig_DeNB}
%x_s[t]&=\begin{bmatrix}
%1 & 1 & 1 & 0 & 1
%\end{bmatrix}\begin{bmatrix}
%\nu_{\eta_{1,1}}[t]\eta_{1,1} & \nu_{\eta_{1,2}}[t]\eta_{1,2} & \nu_{\eta_{1,3}}[t]\eta_{1,3} & \nu_{\eta_{1,4}}[t]\eta_{1,4} & \nu_{\eta_{1,5}}[t]\eta_{1,5} \\ \nu_{\eta_{2,1}}[t]\eta_{2,1} & \nu_{\eta_{2,2}}[t]\eta_{2,2} & \nu_{\eta_{2,3}}[t]\eta_{2,3} & \nu_{\eta_{2,4}}[t]\eta_{2,4} & \nu_{\eta_{2,5}}[t]\eta_{2,5} \\ \nu_{\eta_{3,1}}[t]\eta_{3,1} & \nu_{\eta_{3,2}}[t]\eta_{3,2} & \nu_{\eta_{3,3}}[t]\eta_{3,3} & \nu_{\eta_{3,4}}[t]\eta_{3,4} & \nu_{\eta_{3,5}}[t]\eta_{3,5} \\ 0 & 0 & 0 & 0 & \nu_{4,5}[t]\eta_{4,5}
%\end{bmatrix}^{\dagger}\begin{bmatrix}
%1 \\ 1 \\ 1 \\ 1
%\end{bmatrix} \\ \label{eq:tx_sig_RN}
%x_r[t]&=\begin{bmatrix}
%1 & 1 & 1 & 1 
%\end{bmatrix}\begin{bmatrix}
%\beta_{\eta_{1,1}}[t]\eta_{1,1} & \beta_{\eta_{1,2}}[t]\eta_{1,2} & \beta_{\eta_{1,3}}[t]\eta_{1,3} & \beta_{\eta_{1,4}}[t]\eta_{1,4} \\ \beta_{\eta_{2,1}}[t]\eta_{2,1} & \beta_{\eta_{2,2}}[t]\eta_{2,2} & \beta_{\eta_{2,3}}[t]\eta_{2,3} & \beta_{\eta_{2,4}}[t]\eta_{2,4} \\ \beta_{\eta_{3,1}}[t]\eta_{3,1} & \beta_{\eta_{3,2}}[t]\eta_{3,2} & \beta_{\eta_{3,3}}[t]\eta_{3,3} & \beta_{\eta_{3,4}}[t]\eta_{3,4} \\ 0 & 0 & 0 & 0
%\end{bmatrix}^{\dagger}\begin{bmatrix}
%1 \\ 1 \\ 1 \\ 1
%\end{bmatrix}
%\end{align}}
%\begin{align}
%\label{eq:tx_sig_DeNB}
%x_s[t]&=\sum_{i=1}^{3}\sum_{\substack{j=1,\\j\neq 4}}^{5}\nu_{\eta_{i,j}}[t]\eta_{i,j}+\nu_{\eta_{4,5}}[t]\eta_{4,5},\\
%\label{eq:tx_sig_RN}
%x_r[t]&=\sum_{i=1}^{3}\sum_{j=1}^{4}\beta_{\eta_{i,j}}[t]\eta_{i,j},
%\end{align} 
$\forall t\in[T]$ for $T=8$, respectively. \alert{In these two equations, the vectors $\boldsymbol{\nu}[t]$ and $\boldsymbol{\beta}[t]$ denote the precoding vectors at time instant $t$ for DeNB and RN with respect to all symbols in $\boldsymbol{\eta}$, whereas $\boldsymbol{C}_{\text{RN}}\in\mathbb{F}_2^{20\times 20}$ is a binary (caching) matrix accounting for the cache placement of the RN. The elements of the precoding vectors $\boldsymbol{\nu}[t]$ and $\boldsymbol{\beta}[t]$ comprise of complex precoding scalars $\nu_{\eta_{i,j}}[t]$ and $\beta_{\eta_{i,j}}[t]$ of symbol $\eta_{i,j}$. These elements are stacked to $\boldsymbol{\nu}_i[t]$ and $\boldsymbol{\beta}_i[t]$ which themselves are then concatenated to $\boldsymbol{\nu}[t]$ and $\boldsymbol{\beta}[t]$ in the exact same fashion as $\boldsymbol{\eta}_{i}$ in $\boldsymbol{\eta}$. In \eqref{eq:tx_sig_DeNB} and \eqref{eq:tx_sig_RN}, we have implicitly fixed the following elements of $\boldsymbol{\nu}[t]$ and $\boldsymbol{\beta}[t]$ to be
\begin{align}
\label{eq:some_precod_elements1}
\nu_{\eta_{4,j}}[t]=\beta_{\eta_{4,j}}[t]=0,\quad j=1,\ldots,4\\
\label{eq:some_precod_elements2}
\nu_{\eta_{i,4}}[t]=0,\quad i=1,\ldots,4
\end{align} $\forall t\in[T]$. Further, we fix the RN caching matrix $\boldsymbol{C}_{\text{RN}}$ in accordance to Fig. \ref{fig:files_K_three} and the definition of $\boldsymbol{\eta}$ to
\begin{align}
\label{eq:RN_caching_matrix}
\boldsymbol{C}_{\text{RN}}=
%\begin{pmatrix}
%\mathbf{I}_4 & & & & & & & \\[-0.25cm]
%& 0 & & & & & \bigzero &\\[-0.25cm]
%& & \mathbf{I}_4 & & & & & & \\[-0.25cm] 
%& & & 0 & & & & \\[-0.25cm]
%& & & & \mathbf{I}_4 & & & & \\[-0.25cm]
%& \bigzero & & & & 0 & & \\[-0.25cm]
%& & & & & & \mathbf{I}_4 & \\[-0.25cm]
%& & & & & & & 0  
%\end{pmatrix}=
\mathbf{I}_4\:\otimes\:\begin{pmatrix}
\mathbf{I}_4 & \mathbf{0} \\[-0.25cm] \mathbf{0}^{\dagger} & 0 
\end{pmatrix}.
\end{align} From Eq. \eqref{eq:some_precod_elements2}, we infer that the elements $\nu_{\eta_{i,4}}[t]$ are 0, i.e., irrelevant in $\boldsymbol{\nu}[t]$, which is equivalent to introduce an \emph{effective} caching matrix $\boldsymbol{C}_{\text{DeNB}}^{\text{eff}}$
\begin{align}
\label{eq:DeNB_eff_caching_matrix}
\boldsymbol{C}_{\text{DeNB}}^{\text{eff}}=\mathbf{I}_4\:\otimes\:\begin{pmatrix}
\mathbf{I}_3 & \mathbf{0} & \mathbf{0} \\[-0.25cm] \mathbf{0}^{\dagger} & 0 & 0 \\[-0.25cm] 
\mathbf{0}^{\dagger} & 0 & 1
\end{pmatrix}
\end{align} that ignores the availability of symbols $\eta_{i,4},i=1,\ldots,4,$ at the DeNB}.\footnote{\alert{From this argument, we infer that the \emph{effective} cache size of the DeNB is $\mu=\frac{4}{5}$.}} \alert{In consequence, we may rewrite \eqref{eq:tx_sig_DeNB} as follows
\begin{equation}
\label{eq:tx_sig_DeNB_eff}
x_s[t]=\boldsymbol{\nu}^{\dagger}[t]C_{\text{DeNB}}^{\text{eff}}\boldsymbol{\eta},
\end{equation} where we now implicitly assumed that $\nu_{\eta_{i,4}}[t]\neq 0$.} \alert{For ease of presentation, we first write the received signal at UE$_k$, $k\in[K]=[3]$ according to
\begin{align}
\label{eq:rx_sig_UE_k_M_1_K_3_vec_mac}
y_{u,k}[t]&=g_{k}[t]x_s[t]+h_{k1}[t]x_r[t]+z_{u,k}[t]\nonumber\\
&=\Big(g_{k}[t]\boldsymbol{\nu}^{\dagger}[t]C_{\text{DeNB}}^{\text{eff}}+h_{k1}[t]\boldsymbol{\beta}^{\dagger}[t]\boldsymbol{C}_{\text{RN}}\Big)\boldsymbol{\eta}+z_{u,k}[t]\nonumber\\
&=\sum_{i=1}^{3}\bigg[\Big(g_{k}[t]\boldsymbol{\nu}_{i,[1:3]}^{\dagger}[t]+h_{k1}[t]\boldsymbol{\beta}_{i,[1:3]}^{\dagger}[t]\Big)\mathbf{I}_3\boldsymbol{\eta}_{i,[1:3]}+h_{k1}[t]\beta_{\eta_{i,4}}[t]\eta_{i,4}\bigg]\nonumber\\&\qquad+\sum_{\ell=1}^{4}g_{k}[t]\nu_{\eta_{\ell,5}}[t]\eta_{\ell,5}+z_{u,k}[t].
%&=\begin{bmatrix}
%1 & 1 & 1 & 1 & 1
%\end{bmatrix}\nonumber\\&\quad\begin{bmatrix}
%e_{\eta_{1,1},k}^{\text{ZF}}[t]\eta_{1,1} & e_{\eta_{1,2},k}^{\text{ZF}}[t]\eta_{1,2} & e_{\eta_{1,3},k}^{\text{ZF}}[t]\eta_{1,3} & \beta_{\eta_{1,4}}[t]h_{k1}[t]\eta_{1,4} & \nu_{\eta_{1,5}}[t]g_{k}[t]\eta_{1,5} \\ e_{\eta_{2,1},k}^{\text{ZF}}[t]\eta_{2,1} & e_{\eta_{2,2},k}^{\text{ZF}}[t]\eta_{2,2} & e_{\eta_{2,3},k}^{\text{ZF}}[t]\eta_{2,3} & \beta_{\eta_{2,4}}[t]h_{k1}[t]\eta_{2,4} & \nu_{\eta_{2,5}}[t]g_{k}[t]\eta_{2,5} \\ e_{\eta_{3,1},k}^{\text{ZF}}[t]\eta_{3,1} & e_{\eta_{3,2},k}^{\text{ZF}}[t]\eta_{3,2} & e_{\eta_{3,3},k}^{\text{ZF}}[t]\eta_{3,3} & \beta_{\eta_{3,4}}[t]h_{k1}[t]\eta_{3,4} & \nu_{\eta_{3,5}}[t]g_{k}[t]\eta_{3,5} \\ 0 & 0 & 0 & 0 & \nu_{\eta_{1,5}}[t]g_{k}[t]\eta_{4,5}
%\end{bmatrix}^{\dagger}\begin{bmatrix}
%1 \\ 1 \\ 1 \\ 1
%\end{bmatrix},
\end{align} In \eqref{eq:rx_sig_UE_k_M_1_K_3_vec_mac}, all components for which $i,\ell\neq k$ represent interference. Further, we observe that the effective channel coefficient of the $j$-th ZF symbol of the $i$-th file at UE$_k$ corresponds to $e_{\eta_{i,j},k}^{\text{ZF}}[t]\triangleq g_{k}[t]\nu_{\eta_{i,j}}[t]+h_{k1}[t]\beta_{\eta_{i,j}}[t]$.} \alert{The vectors $\boldsymbol{\nu}[t]$ and $\boldsymbol{\beta}[t]$} are chosen such that both ZF and IA at the UEs become feasible. According to the ZF map of Fig. \ref{fig:ZF_K_three}, the ZF conditions at UE$_k$ %, $k\in[K]=[3]$, 
become 
%\begin{subequations}\label{eq:ZF_conditions}
%\begin{align}
%	&\nu_{\eta_{(k+1)\Mod{K},1}}[t]g_k[t]+\beta_{\eta_{(k+1)\Mod{K},1}}[t]h_{k1}[t]=0,\\
%	&\nu_{\eta_{(k+1)\Mod{K},2}}[t]g_k[t]+\beta_{\eta_{(k+1)\Mod{K},2}}[t]h_{k1}[t]=0,\\
%	&\nu_{\eta_{(k+2)\Mod{K},3}}[t]g_k[t]+\beta_{\eta_{(k+2)\Mod{K},3}}[t]h_{k1}[t]=0.
%	\end{align}
%\end{subequations} 
\alert{
\begin{subequations}\label{eq:ZF_conditions}
	\begin{align}
	&e_{\eta_{k+1,1},k}^{\text{ZF}}[t]\triangleq g_k[t]\nu_{\eta_{k+1,1}}[t]+h_{k1}[t]\beta_{\eta_{k+1,1}}[t]=0,\\
	&e_{\eta_{k+1,2},k}^{\text{ZF}}[t]\triangleq g_k[t]\nu_{\eta_{k+1,2}}[t]+h_{k1}[t]\beta_{\eta_{k+1,2}}[t]=0,\\
	&e_{\eta_{k+2,3},k}^{\text{ZF}}[t]\triangleq g_k[t]\nu_{\eta_{k+2,3}}[t]+h_{k1}[t]\beta_{\eta_{k+2,3}}[t]=0.
	\end{align}
\end{subequations} For the sake of compact notation, we have used modulo-$K$ indexing with respect to symbols $\eta_{k+1,1},\eta_{k+1,2}$ and $\eta_{k+2,1}$ in the formulation of the ZF conditions \eqref{eq:ZF_conditions}.} Simultaneously, we design the precoding scalars \alert{$\nu_{\eta_{i,j}}[t]$ and $\beta_{\eta_{i,j}}[t]$} such that the interference at each UE is aligned into a three-dimensional signal space. (The remaining $5$ dimensions are reserved for the $5$ symbols of the desired file.) The interference graph in Fig. \ref{fig:IA_Graph} shows which symbols align with each other at which UE. This graph consists of $3$ layers. In the first layer, two symbols, namely $\eta_{4,5}$ and $\eta_{1,4}$, $\eta_{2,4}$ or $\eta_{3,4}$ align at the three UEs. At layers two and three, on the other hand, three symbols align per UE. Symbols $\eta_{1,4},\eta_{2,4}$ and $\eta_{3,4}$ link layers $1$ and $2$, while $\eta_{1,5},\eta_{2,5}$ and $\eta_{3,5}$ connect layers $2$ and $3$. In analogy to the graph in Fig. \ref{fig:IA_Graph}, the alignment conditions at UE$_k$ can be written as 
\begin{figure*}[t]%[h]
	\centering
	\begin{tikzpicture}[scale=1]
	\AlignGraph
	\end{tikzpicture}
	\caption[Interference alignment graph for the achievability at corner point $(\frac{4}{5},\frac{8}{5})$ for $M=1$ and $K=3$]{\small Interference alignment graph for the achievability at corner point $(\frac{4}{5},\frac{8}{5})$ for $M=1$ and $K=3$. The graph consists of three (subspace) alignment chains. By definition, the first alignment chain is the path from the node $\eta_{4,5}$ to $\eta_{3,1}$, the second from $\eta_{4,5}$ to $\eta_{1,1}$ and the third from $\eta_{4,5}$ to $\eta_{2,1}$.} 
	\label{fig:IA_Graph}
\end{figure*}
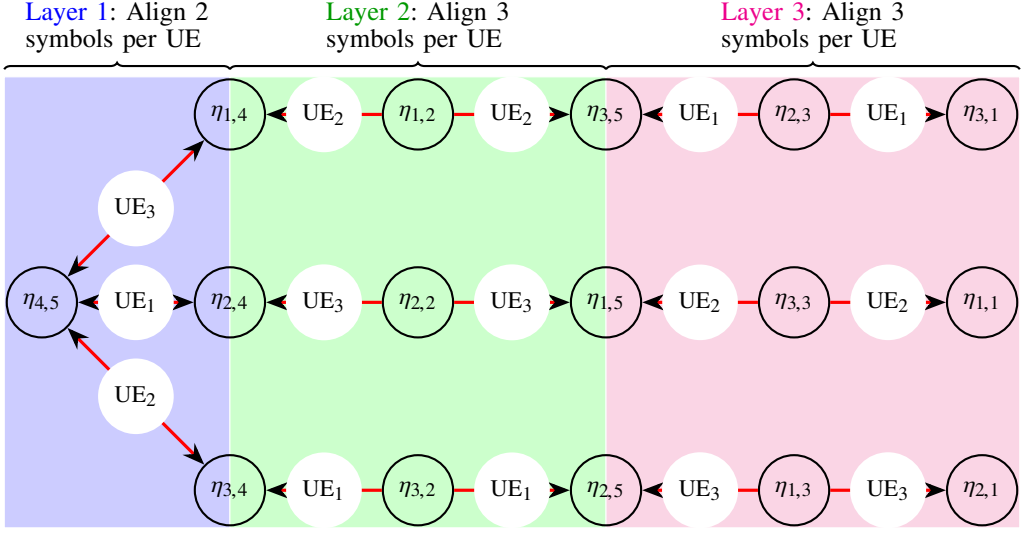 
%\begin{align}
%\label{eq:IA_condition1}
%\nu_{\eta_{4,5}}[t]g_k[t]&=\beta_{\eta_{(k+1)\Mod{K},4}}[t]h_{k1}[t]
%\end{align}
\alert{\begin{align}
\label{eq:IA_condition1}
g_k[t]\nu_{\eta_{4,5}}[t]&=h_{k1}[t]\beta_{k+1,4}[t]
\end{align}}
for Layer 1, 
%\begin{align}
%\label{eq:IA_condition2}
%\hspace{-0.75em}\beta_{\eta_{(k+2)\Mod{K},4}}[t]h_{k1}[t]=\beta_{\eta_{(k+2)\Mod{K},2}}[t]h_{k1}[t]+\nu_{\eta_{(k+2)\Mod{K},2}}[t]g_{k}[t]=\nu_{\eta_{(k+1)\Mod{K},5}}[t]g_{k}[t]
%\end{align}
\alert{\begin{align}
\label{eq:IA_condition2}
h_{k1}[t]\beta_{\eta_{k+2,4}}[t]=e_{\eta_{k+2,2},k}^{\text{ZF}}[t]=g_{k}[t]\nu_{\eta_{k+1,5}}[t]
\end{align}}for Layer 2, and 
%\begin{align}
%\label{eq:IA_condition3}
%&\nu_{\eta_{(k+2)\Mod{K},5}}[t]g_{k}[t]=\beta_{\eta_{(k+2)\Mod{K},1}}[t]h_{k1}[t]+\nu_{\eta_{(k+2)\Mod{K},1}}[t]g_{k}[t]\nonumber\\
%&=\beta_{\eta_{(k+1)\Mod{K},3}}[t]h_{k1}[t]+\nu_{\eta_{(k+1)\Mod{K},3}}[t]g_{k}[t]
%\end{align}    
\alert{\begin{align}
\label{eq:IA_condition3}
&g_{k}[t]\nu_{\eta_{k+2,5}}[t]=e_{\eta_{k+2,1},k}^{\text{ZF}}[t]=e_{\eta_{k+1,3},k}^{\text{ZF}}[t]
\end{align}} for Layer 3. Under the given ZF and IA conditions (cf.  \eqref{eq:ZF_conditions} and \eqref{eq:IA_condition1}--\eqref{eq:IA_condition3}), the precoders are functions of the channels $\mathbf{g}[t]$ and $\mathbf{H}[t]$. We fix the precoder for symbol $\eta_{4,5}$ to
\begin{equation}\label{eq:precoder_d5}
\nu_{\eta_{4,5}}[t]=j_{13}[t]j_{23}[t]j_{33}[t]g_1[t]g_2[t]g_3[t]h_{11}[t]h_{21}[t]h_{31}[t],
\end{equation} where
\begin{subequations}
	\begin{align}
	j_{13}[t]&=g_2[t]h_{31}[t]-g_{3}[t]h_{21}[t],\\
	j_{23}[t]&=g_3[t]h_{11}[t]-g_{1}[t]h_{31}[t],\\
	j_{33}[t]&=g_1[t]h_{21}[t]-g_{2}[t]h_{11}[t].
	\end{align}		
\end{subequations} The remaining precoders scalars depend on $\nu_{\eta_{4,5}}[t]$ and can be computed by using \eqref{eq:precoder_d5} in \eqref{eq:ZF_conditions} and  \eqref{eq:IA_condition1}--\eqref{eq:IA_condition3}. To this end, we compute the precoders along the $r$-th alignment chain \cite{Wang14}, $r\in[3]$, of the graph (the first chain in Fig. \ref{fig:IA_Graph} for instance being the entire path from node $\eta_{4,5}$ to $\eta_{3,1}$) as a function of $\nu_{\eta_{4,5}}[t]$. \alert{The resulting precoders under modulo-$K$ indexing then become:} 
%\begin{subequations}
%\label{eq:all_precoders_function_d5}
%	\begin{align}
%	\beta_{\eta_{r,4}}[t]&=\nu_{\eta_{4,5}}[t]\cdot\frac{g_{(r+2)\Mod{K}}[t]}{h_{(r+2)\Mod{K},1}[t]},\\\mathbf{p}_{\eta_{r,2}}[t]&\triangleq\begin{pmatrix}
%	\nu_{\eta_{r,2}}[t] \\\beta_{\eta_{r,2}}[t]
%	\end{pmatrix}\nonumber\\&=\nu_{\eta_{4,5}}[t]\cdot\frac{g_{(r+2)\Mod{K}}[t]h_{(r+1)\Mod{K},1}[t]}{j_{r,3}[t]}\begin{pmatrix}
%	1 \\ -\frac{g_{(r+2)\Mod{K}}[t]}{h_{(r+2)\Mod{K},1}[t]}
%	\end{pmatrix},\label{eq:concatenated_prec_1}\\
%	\nu_{\eta_{(r+2)\Mod{K},5}}[t]&=\nu_{\eta_{4,5}}[t]\cdot\frac{g_{(r+2)\Mod{K}}[t]h_{(r+1)\Mod{K},1}[t]}{h_{(r+2)\Mod{K},1}[t]g_{(r+1)\Mod{K}}[t]}, \\\mathbf{p}_{\eta_{(r+1)\Mod{K},3}}[t]&\triangleq\begin{pmatrix}
%	\nu_{\eta_{(r+1)\Mod{K},3}}[t] \\\beta_{\eta_{(r+1)\Mod{K},3}}[t]
%	\end{pmatrix}\nonumber\\&=\nu_{\eta_{4,5}}[t]\cdot\frac{g_{(r+2)\Mod{K}}[t]h_{(r+1)\Mod{K},1}[t]g_r[t]}{g_{(r+1)\Mod{K}}[t]j_{(r+1)\Mod{K},3}[t]}\begin{pmatrix}
%	-1 \\ \frac{g_{(r+2)\Mod{K}}[t]}{h_{(r+2)\Mod{K},1}[t]}
%	\end{pmatrix},\label{eq:concatenated_prec_2}\\\mathbf{p}_{\eta_{(r+2)\Mod{K},1}}[t]&\triangleq\begin{pmatrix}
%	\nu_{\eta_{(r+2)\Mod{K},1}}[t] \\\beta_{\eta_{(r+2)\Mod{K},1}}[t]
%	\end{pmatrix}\nonumber\\&=\nu_{\eta_{4,5}}[t]\cdot\frac{g_{(r+2)\Mod{K}}[t]h_{(r+1)\Mod{K},1}[t]g_r[t]}{h_{(r+2)\Mod{K},1}[t]j_{(r+2)\Mod{K},3}[t]}\begin{pmatrix}
%	\frac{h_{(r+1)\Mod{K},1}[t]}{g_{(r+1)\Mod{K}}[t]} \\ -1
%	\end{pmatrix}\label{eq:concatenated_prec_3}.
%	\end{align}		
%\end{subequations}
\alert{ 
\begin{subequations}
	\label{eq:all_precoders_function_d5}
	\begin{align}
	\beta_{\eta_{r,4}}[t]&=\nu_{\eta_{4,5}}[t]\cdot\frac{g_{r+2}[t]}{h_{r+2,1}[t]},\\\mathbf{p}_{\eta_{r,2}}[t]&\triangleq\begin{pmatrix}
	\nu_{\eta_{r,2}}[t] \\\beta_{\eta_{r,2}}[t]
	\end{pmatrix}=\nu_{\eta_{4,5}}[t]\cdot\frac{g_{r+2}[t]h_{r+1,1}[t]}{j_{r,3}[t]}\begin{pmatrix}
	1 \\ -\frac{g_{r+2}[t]}{h_{r+2,1}[t]}
	\end{pmatrix},\label{eq:concatenated_prec_1}\\
	\nu_{\eta_{r+2,5}}[t]&=\nu_{\eta_{4,5}}[t]\cdot\frac{g_{r+2}[t]h_{r+1,1}[t]}{h_{r+2,1}[t]g_{r+1}[t]}, \\\mathbf{p}_{\eta_{r+1,3}}[t]&\triangleq\begin{pmatrix}
	\nu_{\eta_{r+1,3}}[t] \\\beta_{\eta_{r+1,3}}[t]
	\end{pmatrix}=\nu_{\eta_{4,5}}[t]\cdot\frac{g_{r+2}[t]h_{r+1,1}[t]g_r[t]}{g_{r+1}[t]j_{r+1,3}[t]}\begin{pmatrix}
	-1 \\ \frac{g_{r+2}[t]}{h_{r+2,1}[t]}
	\end{pmatrix},\label{eq:concatenated_prec_2}\\\mathbf{p}_{\eta_{r+2,1}}[t]&\triangleq\begin{pmatrix}
	\nu_{\eta_{r+2,1}}[t] \\\beta_{\eta_{r+2,1}}[t]
	\end{pmatrix}=\nu_{\eta_{4,5}}[t]\cdot\frac{g_{r+2}[t]h_{r+1,1}[t]g_r[t]}{h_{r+2,1}[t]j_{r+2,3}[t]}\begin{pmatrix}
	\frac{h_{r+1,1}[t]}{g_{r+1}[t]} \\ -1
	\end{pmatrix}\label{eq:concatenated_prec_3}.
	\end{align}		
\end{subequations}}
Taking a closer look at the \emph{concatenated} precoding vectors in \eqref{eq:concatenated_prec_1}, \eqref{eq:concatenated_prec_2} and \eqref{eq:concatenated_prec_3}, we see that these vectors are orthogonal to 
%\begin{subequations}\label{eq:orthogonal_vectors}
%\begin{align}
%\mathbf{p}^{\perp}_{\eta_{(r+1)\Mod{K},3}}[t]=\mathbf{p}^{\perp}_{\eta_{r,2}}[t]=\mathbf{\tilde{h}}_{(r+2)\Mod{K}}[t]&\triangleq\begin{pmatrix}
%g_{(r+2)\Mod{K}}[t] \\ h_{(r+2)\Mod{K},1}[t]
%\end{pmatrix},\\\mathbf{p}^{\perp}_{\eta_{(r+2)\Mod{K},1}}[t]=\mathbf{\tilde{h}}_{(r+1)\Mod{K}}[t]&\triangleq\begin{pmatrix}
%g_{(r+1)\Mod{K}}[t] \\ h_{(r+1)\Mod{K},1}[t]
%\end{pmatrix}.
%\end{align}
%\end{subequations} 
\alert{
\begin{subequations}\label{eq:orthogonal_vectors}
	\begin{align}
	\mathbf{p}^{\perp}_{\eta_{r+1,3}}[t]=\mathbf{p}^{\perp}_{\eta_{r,2}}[t]=\mathbf{\tilde{h}}_{r+2}[t]&\triangleq\begin{pmatrix}
	g_{r+2}[t] \\ h_{r+2,1}[t]
	\end{pmatrix},\\\mathbf{p}^{\perp}_{\eta_{r+2,1}}[t]=\mathbf{\tilde{h}}_{r+1}[t]&\triangleq\begin{pmatrix}
	g_{r+1}[t] \\ h_{r+1,1}[t]
	\end{pmatrix}.
	\end{align}
\end{subequations}} 
This observation is in agreement with the ZF conditions \eqref{eq:ZF_conditions}. \alert{We may check that the alignment conditions \eqref{eq:IA_condition1}--\eqref{eq:IA_condition3} are also satisfied under the choice of the precoders.} Now we will go through the decoding from the perspective of both the RN and the UEs. 

\subsubsection*{Decoding at RN and the UEs}
\alert{Exploiting both ZF and IA conditions in \eqref{eq:rx_sig_UE_k_M_1_K_3_vec_mac}, we can write the received signal at the $k$-th UE as follows:
\begin{align}
\label{eq:rx_sig_UE_k_M_1_K_3_compact}
y_{u,k}[t]&=D_k(\boldsymbol{\eta}_{k,[1:5]})\nonumber\\&\qquad+I_k(\eta_{4,5}+\eta_{k+1,4},\eta_{k+2,4}+\eta_{k+2,2}+\eta_{k+1,5},\eta_{k+2,5}+\eta_{k+1,3}+\eta_{k+2,1})+z_{u,k}[t],
\end{align} where $D_k$ and $I_k$ are linear combinations of desired symbols $\boldsymbol{\eta}_{k,[1:5]}$ and aligned interference symbols of all three layers at UE$_k$, respectively. These two linear combinations are given by
\begin{align}
D_k(\boldsymbol{\eta}_{k,[1:5]})=\Big(g_{k}[t]\boldsymbol{\nu}_{k,[1:3]}^{\dagger}[t]+h_{k1}[t]\boldsymbol{\beta}_{k,[1:3]}^{\dagger}[t]\Big)\mathbf{I}_3\boldsymbol{\eta}_{k,[1:3]}+h_{k1}[t]\beta_{\eta_{k,4}}[t]\eta_{k,4}+g_{k}[t]\nu_{\eta_{k,5}}[t]\eta_{k,5}
\end{align} and   
\begin{align}
I_k&=e_{\eta_{k+1,3},k}^{\text{ZF}}[t](\eta_{k+2,5}+\eta_{k+1,3}+\eta_{k+2,1})+e_{\eta_{k+2,2},k}^{\text{ZF}}[t](\eta_{k+2,4}+\eta_{k+2,2}+\eta_{k+1,5})\nonumber\\&\qquad+g_{k}[t]\nu_{\eta_{4,5}}[t](\eta_{4,5}+\eta_{k+1,4}). 
\end{align}} 
At the RN, on the other hand, the knowledge of $\eta_{i,j},i,j\in[3],$ as side information prefetched in its cache is exploited to cancel the contribution of these components. Thus, at the $t$-th channel use, the RN observes 
\begin{align}\label{eq:rx_sig_RN_M1_K3}
y_{r,1}'[t]=&f_{1}[t]\bigg(\nu_{\eta_{4,5}}[t]\eta_{4,5}+\sum_{i=1}^{3}\nu_{\eta_{i,5}}[t]\eta_{i,5}\bigg)+z_{r,1}[t].
\end{align} Recall that the scheme spans over $T=8$ channel uses. Thus, due to the time-variant nature of the wireless channel, the UEs and the RN have $8$ (noise-corrupted) linear \emph{independent} observations $\{y_{u,k}[t]\}_{t=1}^{8}$ and $\{y_{r,1}'[t]\}_{t=1}^{8}$ \alert{(in the field of reals $\mathbb{R}$ or complex $\mathbb{C}$)} according to \eqref{eq:rx_sig_UE_k_M_1_K_3_compact} and \eqref{eq:rx_sig_RN_M1_K3}, respectively. In consequence, UE$_k$, $k\in[3]$, on the one hand, is able to decode its 5 \emph{desired} symbols 
\begin{itemize}
%	\item $\eta_{(\tilde{r}_k+2)\Mod{K},1}$, $\eta_{\bar{r}_k,2}$, $\eta_{(r_k'+1)\Mod{K},3}$, $\eta_{\bar{r}_k,4}$, $\eta_{(\tilde{r}_k+2)\Mod{K},5}$
\item $\eta_{k,1}$, $\eta_{k,2}$, $\eta_{k,3}$, $\eta_{k,4}$, $\eta_{k,5}$
\end{itemize} and 3 \emph{aligned} interfering symbols
\alert{\begin{itemize}
%	\item $\eta_{\tilde{r}_k,4}+\eta_{4,5}$,
%	\item $\eta_{r_k',4}+\eta_{r_k',2}+\eta_{(r_k'+2)\Mod{K},5}$
%	\item and $\eta_{(\bar{r}_k+2)\Mod{K},5}+\eta_{(\bar{r}_k+1)\Mod{K},3}+\eta_{(\bar{r}_k+2)\Mod{K},1}$,
	\item $\eta_{4,5}+\eta_{k+1,4}$,
	\item $\eta_{k+2,4}+\eta_{k+2,2}+\eta_{k+1,5}$
	\item and $\eta_{k+2,5}+\eta_{k+1,3}+\eta_{k+2,1}$,
\end{itemize} while the RN, on the other hand, decodes its desired symbol $\eta_{4,5}$ and 3 interfering (but not aligned) symbols $\eta_{1,5}$, $\eta_{2,5}$ and $\eta_{3,5}$.} 
% This is due to the fact that, the time-variant nature of the channel allows both the RN and the UEs to receive $8$ (noise-corrupted) independent observations of its received signal components. 
In fact, as far as the RN is concerned, it actually only requires $4$ channel uses to allow for the decoding of the aforementioned symbols. The UEs are the reason why the scheme spans over $8$, and not $4$, channel uses. Consequently, the achievable NDT becomes $\frac{8}{5}$. Next, we present the NDT-optimal schemes for which $M=2$ and $K\leq 2$. 

\subsection{Achievability for $M=2$}\label{subsec:ach_M2}        

In the following, we will outline the delivery time optimal schemes for $M=2$ and $K\leq 2$. The following proposition states the achievable NDT.

\begin{proposition}\label{prop_M2}
	The achievable NDT of the network under study for $M=2$ RNs, \alert{$K\in\{1,2\}$} UEs and $\mu\in[0,1]$ is given by 
	\begin{equation}
	\label{eq:ach_NDT_M_2_K_12}
	\delta(\mu)=\begin{cases}\max\Big\{1,3-4\mu\Big\}\quad&\text{ if } K= 1\\\max\bigg\{4-6\mu,\frac{4-3\mu}{2},\frac{3-\mu}{2}\bigg\}\quad&\text{ if }K=2\end{cases}.
	\end{equation}
\end{proposition}
Now, we will present the proof of this proposition. Fig. \ref{fig:NDT_M_K_2_LEQ_2} shows that the achievable NDT in Proposition \ref{prop_M2} coincides with the lower bound. Due to Lemma \ref{corr_mu_0_and_1} and arguments of memory sharing, we will only establish the achievability of the corner point(s)
%\begin{itemize}
%	\item $\Big(\frac{1}{2},1\Big)$ 
%\end{itemize} 
$\Big(\frac{1}{2},1\Big)$
for $K=1$ and
%\begin{itemize}
%	\item $\Big(\frac{4}{9},\frac{4}{3}\Big)$, 
%	\item $\Big(\frac{1}{2},\frac{5}{4}\Big)$ 
%\end{itemize} 
$\Big(\frac{4}{9},\frac{4}{3}\Big)$, $\Big(\frac{1}{2},\frac{5}{4}\Big)$ for $K=2$, respectively. In this context, the NDT-optimal schemes for corner points $\Big(\frac{1}{2},1\Big)$ and $\Big(\frac{1}{2},\frac{5}{4}\Big)$ are based on the \emph{one-shot}\footnote{One-shot in this context is equivalent to requiring a \emph{single} channel use for conveying desired symbols to the receiving nodes.} reception of desired symbols, whereas at corner point $\Big(\frac{4}{9},\frac{4}{3}\Big)$ symbol decoding over multiple channel uses (similarly to the achievability for $(K,M)=(3,1)$ in subsection \ref{subsec:ach_M1}) is required. We start the achievability for the latter corner point $\Big(\frac{4}{9},\frac{4}{3}\Big)$.    

\subsubsection*{RNs Cache Placement}

\begin{figure}
	\FilesMTwoKThree
	\caption[Requested files by $K=2$ users and $M=2$ RNs]{\small Requested files by $K=2$ users and $M=2$ RNs and the availability illustrated by the symbols transmitted from the DeNB only or from both at the DeNB and one of the two RNs RN$_1$ and RN$_2$.}
	\label{fig:files_M_two_K_two}
\end{figure}

Again, we assume without loss of generality $N=4$ and \alert{one possible worst-case demand scenario in which} the UEs request files $W_1,W_2$ while RN$_1$ and RN$_2$ are interested in files $W_3$ and $W_4$, respectively. As shown in Fig. \ref{fig:files_M_two_K_two}, all files are comprised of $9$ symbols, i.e., the $i$-th file is composed of symbols \alert{$\boldsymbol{\eta}_i=(\eta_{i,1},\eta_{i,2},\ldots,\eta_{i,9})^{\dagger}$. Similarly to the previously described scheme, $\boldsymbol{\eta}=(\boldsymbol{\eta}_1^{\dagger},\boldsymbol{\eta}_2^{\dagger},\boldsymbol{\eta}_3^{\dagger},\boldsymbol{\eta}_4^{\dagger})^{\dagger}$ represent the column-wise concatenation of all $N=4$ files.} \emph{All} symbols \alert{in $\boldsymbol{\eta}$} are available at the DeNB. In contrast, RN$_1$ and RN$_2$ prefetch symbols $\eta_{i,j},j=1,\ldots,4$ and $\eta_{i,j},j=5,\ldots,8$, respectively in its cache\footnote{We note that at $\mu=\nicefrac{4}{9}$, the joint cache content of the two RNs cannot contain the entire library of popular files.}. \alert{For the sake of compact notation, these cached symbols are denoted by $\boldsymbol{\eta}_{i,[1:4]}$ and $\boldsymbol{\eta}_{i,[5:8]}$.} Since, RN$_1$ (RN$_2$) is, interested in file $W_3$ ($W_4$) and it knows a-priori 
%$\eta_{3,j},j=1,\ldots,4$ ($\eta_{4,l},l=5,\ldots,8$)
\alert{$\boldsymbol{\eta}_{3,[1:4]}$ ($\boldsymbol{\eta}_{4,[5:8]}$), the missing symbols RN$_1$ (RN$_2$) desires are 
% $\eta_{3,5},\eta_{3,6},\ldots,\eta_{3,9}$ ($\eta_{4,1},\ldots,\eta_{4,4}$ and $\eta_{4,9}$). 
$\boldsymbol{\eta}_{3,[5:9]}$ ($\boldsymbol{\eta}_{4,[1:4]}$ and $\eta_{4,9}$).} In consequence, we design the transmission policy such that DeNB, RN$_1$ and RN$_2$ convey all receiving nodes with their desired symbols of files $W_1$ and $W_2$ as well as % $\eta_{3,5},\eta_{3,6},\ldots,\eta_{3,9},\eta_{4,1},\ldots,\eta_{4,4}$ and $\eta_{4,9}$. 
\alert{$\boldsymbol{\eta}_{3,[5:9]}$, $\boldsymbol{\eta}_{4,[1:4]}$ and $\eta_{4,9}$.} These are in total $28$ information symbols. Next, we consider the encoding strategy for DeNB and RNs. 
\vspace{.5em}
\subsubsection*{Encoding at DeNB and RNs} 

\begin{figure}[h]
	\ZFMTwoKTwo
	\caption[Zero-forcing map]{\small Map that assigns which symbol to zero-force at which UE.} 
	\label{fig:ZF_Mtwo_Ktwo}
\end{figure} 

Similarly to the previously described scheme, we use the spatial correlation of shared symbols between RN$_1$ or RN$_2$ with the DeNB to facilitate ZF and IA at the UEs. However, leveraging all ZF opportunities at the UEs has one main drawback from the RNs perspective -- namely, it maximizes the number of interfering symbols imposed on RN$_1$ and RN$_2$. To limit this number, we do \emph{not} send symbols 
%$\eta_{i,3},\eta_{i,4},\eta_{i,7}$ and $\eta_{i,8}$, 
\alert{$\boldsymbol{\eta}_{i,[3:4]}$ and $\boldsymbol{\eta}_{i,[7:8]}$,} $i=1,2$, through the DeNB (at the cost of missing these symbols' ZF opportunities) but rather through RN$_1$ or RN$_2$ directly to the UEs. All in all, this constitutes the transmit signals  
%\begin{align}
%\label{eq:tx_sig_DeNB_M2_K2}
%x_s[t]&=\sum_{i_1=1}^{2}\sum_{\substack{j_1=1,\\j_1\neq \{3,4,7,8\}}}^{8}\nu_{\eta_{i_1,j_1}}[t]\eta_{i_1,j_1}+\sum_{j_2=5}^{8}\nu_{\eta_{3,j_2}}[t]\eta_{3,j_2}+\sum_{j_3=1}^{4}\nu_{\eta_{4,j_3}}[t]\eta_{4,j_3}+\sum_{i_2=1}^{4}\nu_{\eta_{i_2,9}}[t]\eta_{i_2,9},\\
%\label{eq:tx_sig_RN1_M2_K2}
%x_{r,m}[t]&=\sum_{\substack{i=1,\\i\neq m+2}}^{4}\sum_{j=4m-3}^{4m}\beta_{\eta_{i,j}}[t]\eta_{i,j},
%\end{align} 
\begin{align}
\label{eq:tx_sig_DeNB_M2_K2}
x_s[t]&=\boldsymbol{\nu}^{\dagger}[t]\boldsymbol{\eta}\\
\label{eq:tx_sig_RN1_M2_K2}
x_{r,m}[t]&=\boldsymbol{\beta}^{(m)^{\dagger}}[t]\boldsymbol{C}_{\text{RN}_m}\boldsymbol{\eta},
\end{align} 
$\forall m\in[2]$ and $\forall t\in[T]$ for $T=12$, respectively. In \alert{the equation above}, the notation for the precoding scalars is identical to the one utilized in the previous achievability scheme \alert{with the slight difference that $\boldsymbol{\beta}^{(m)}[t]$ is the precoding vector at the $m$-th RN with respect to symbols $\boldsymbol{\eta}$. In \eqref{eq:tx_sig_DeNB_M2_K2} and \eqref{eq:tx_sig_RN1_M2_K2}, the following precoding subvectors in $\boldsymbol{\nu}[t]$ and $\boldsymbol{\beta}^{(m)}[t]$ are set to
\begin{align}
\label{eq:some_precod_subvec_1}
\boldsymbol{\nu}_{i,[3:4]}[t]=\boldsymbol{\nu}_{i,[7:8]}[t]&=\mathbf{0},\quad i=1,2 \\ \label{eq:some_precod_subvec_2}
\boldsymbol{\nu}_{3,[1:4]}[t]=\boldsymbol{\beta}^{(1)}_{3,[1:4]}[t]&=\mathbf{0}, \\ \label{eq:some_precod_subvec_3}
\boldsymbol{\nu}_{4,[5:8]}[t]=\boldsymbol{\beta}^{(2)}_{4,[5:8]}[t]&=\mathbf{0},
\end{align}$\forall t\in[T]$. The binary caching matrices $\boldsymbol{C}_{\text{RN}_1},\boldsymbol{C}_{\text{RN}_2}\in\mathbb{F}_2^{36\times 36}$ at the two RNs can be represented by the following Kronecker products 
\begin{align}
\label{eq:RN1_caching_matrix}
\boldsymbol{C}_{\text{RN}_1}&=
\mathbf{I}_4\:\otimes\:\begin{pmatrix}
\mathbf{I}_4 & \mathbf{0}_{4\times 5} \\[-0.25cm] \mathbf{0}_{5\times 4} & \mathbf{0}_{5\times 5} 
\end{pmatrix},\\ \label{eq:RN2_caching_matrix}
\boldsymbol{C}_{\text{RN}_2}&=
\mathbf{I}_4\:\otimes\:\begin{pmatrix}
\mathbf{0}_{4\times 4} & \mathbf{0}_{4\times 4} & \mathbf{0} \\[-0.25cm] \mathbf{0}_{4\times 4} & \mathbf{I}_4 & \mathbf{0} \\[-0.25cm] \mathbf{0}^{\dagger} & \mathbf{0}^{\dagger} & 0
\end{pmatrix}.
\end{align} Similarly to the previously described scheme, zero elements in $\boldsymbol{\nu}[t]$ (as stated in Eqs. \eqref{eq:some_precod_subvec_1}--\eqref{eq:some_precod_subvec_3}) can be incorporated to the effective DeNB caching matrix\footnote{\alert{The structure of this matrix suggests that effectively the fractional cache size of the DeNB is $\mu=\frac{5}{9}$.}}
\begin{align}
\label{eq:DeNB_eff_caching_matrix_scheme2}
\boldsymbol{C}_{\text{DeNB}}^{\text{eff}}=\begin{pmatrix}
\begin{pmatrix}
1 & 0 & 0 & 0 \\ 0 & 1 & 0 & 0
\end{pmatrix}\otimes\boldsymbol{C}_1 \\ \begin{pmatrix}
0 & 0 & 1 & 0
\end{pmatrix}\otimes\boldsymbol{C}_2 \\ \begin{pmatrix}
0 & 0 & 0 & 1
\end{pmatrix}\otimes\boldsymbol{C}_3 
\end{pmatrix}
\end{align} with
\begin{align*}
\boldsymbol{C}_1=
%\begin{pmatrix}
%\mathbf{I}_2 & \mathbf{0}_{2\times 2} & \mathbf{0}_{2\times 2} & \mathbf{0}_{2\times 2} & \mathbf{0} \\[-0.25cm] \mathbf{0}_{2\times 2} & \mathbf{0}_{2\times 2} & \mathbf{0}_{2\times 2} & \mathbf{0}_{2\times 2} & \mathbf{0} \\[-0.25cm] \mathbf{0}_{2\times 2} & \mathbf{0}_{2\times 2} & \mathbf{I}_{2} & \mathbf{0}_{2\times 2} & \mathbf{0} \\[-0.25cm] \mathbf{0}_{2\times 2} & \mathbf{0}_{2\times 2} & \mathbf{0}_{2\times 2} & \mathbf{0}_{2\times 2} & \mathbf{0}\\[-0.25cm] \mathbf{0}^{\dagger} & \mathbf{0}^{\dagger} & \mathbf{0}^{\dagger} & \mathbf{0}^{\dagger} & 1
%\end{pmatrix}, 
\begin{pmatrix}
\mathbf{I}_2\otimes\begin{pmatrix}
\mathbf{I}_2 & \mathbf{0}_{2\times 2} \\ \mathbf{0}_{2\times 2} & \mathbf{0}_{2\times 2} 
\end{pmatrix} & \mathbf{0} \\ \mathbf{0}^{\dagger} & 1
\end{pmatrix},\:
\boldsymbol{C}_2=\begin{pmatrix}
\mathbf{0}_{4\times 4} & \mathbf{0}_{4\times 5} \\[-0.25cm]
\mathbf{0}_{5\times 4} & \mathbf{I}_{5}
\end{pmatrix}\text{ and } \boldsymbol{C}_3=\begin{pmatrix}
\mathbf{I}_{4} & \mathbf{0}_{4\times 4} & \mathbf{0} \\[-0.25cm]
\mathbf{0}_{4\times 4} & \mathbf{0}_{4\times 4} & \mathbf{0} \\[-0.25cm] \mathbf{0}^{\dagger} & \mathbf{0}^{\dagger} & 1
\end{pmatrix}
\end{align*} such that we can rewrite the DeNB transmit signal in Eq. \eqref{eq:tx_sig_DeNB_M2_K2} by \begin{equation}
\label{eq:tx_sig_DeNB_eff_scheme2}
x_s[t]=\boldsymbol{\nu}^{\dagger}[t]\boldsymbol{C}_{\text{DeNB}}^{\text{eff}}\boldsymbol{\eta}. 
\end{equation} Under these transmit signals, the observation at the $k$-th UE becomes
\begin{align}
\label{eq:rx_sig_UE_k_M2_K2_initial}
y_{u,k}[t]&=g_{k}[t]x_{s}[t]+h_{k1}[t]x_{r,1}[t]+h_{k2}[t]x_{r,2}[t]+z_{u,k}[t]\nonumber\\&=\Big(g_{k}[t]\boldsymbol{\nu}^{\dagger}[t]\boldsymbol{C}_{\text{DeNB}}^{\text{eff}}+h_{k1}[t]\boldsymbol{\beta}^{(1)\dagger}[t]\boldsymbol{C}_{\text{RN}_1}+h_{k2}[t]\boldsymbol{\beta}^{(2)\dagger}[t]\boldsymbol{C}_{\text{RN}_2}\Big)\boldsymbol{\eta}+z_{u,k}[t]\nonumber\\&=\sum_{i=1}^{2}\bigg[\Big(g_{k}[t]\boldsymbol{\nu}_{i,[1:2]}^{\dagger}[t]+h_{k1}[t]\boldsymbol{\beta}_{i,[1:2]}^{(1)\dagger}[t]\Big)\mathbf{I}_{2}\boldsymbol{\eta}_{i,[1:2]}+\Big(g_{k}[t]\boldsymbol{\nu}_{i,[5:6]}^{\dagger}[t]+h_{k2}[t]\boldsymbol{\beta}_{i,[5:6]}^{(2)\dagger}[t]\Big)\mathbf{I}_{2}\boldsymbol{\eta}_{i,[5:6]}\nonumber\\&\quad+h_{k1}[t]\boldsymbol{\beta}_{i,[3:4]}^{(1)\dagger}[t]\mathbf{I}_2\boldsymbol{\eta}_{i,[3:4]}+h_{k2}[t]\boldsymbol{\beta}_{i,[7:8]}^{(2)\dagger}[t]\mathbf{I}_2\boldsymbol{\eta}_{i,[7:8]}+g_{k}[t]\nu_{\eta_{i,9}}[t]\eta_{i,9}\bigg]\nonumber\\&\quad+\Big(g_{k}[t]\boldsymbol{\nu}_{3,[5:8]}^{\dagger}[t]+h_{k2}[t]\boldsymbol{\beta}_{3,[5:8]}^{(2)\dagger}[t]\Big)\mathbf{I}_{4}\boldsymbol{\eta}_{3,[5:8]}+\Big(g_{k}[t]\boldsymbol{\nu}_{4,[1:4]}^{\dagger}[t]+h_{k1}[t]\boldsymbol{\beta}_{4,[1:4]}^{(1)\dagger}[t]\Big)\mathbf{I}_{4}\boldsymbol{\eta}_{4,[1:4]}\nonumber\\&\quad+\sum_{\ell=3}^{4}g_{k}[t]\nu_{\eta_{\ell,9}}[t]\eta_{\ell,9}+z_{u,k}[t].
\end{align} All components in \eqref{eq:rx_sig_UE_k_M2_K2_initial} for which $i,\ell\neq k$ represent interference. The coefficient vector $\boldsymbol{e}^{\text{ZF}}_{\boldsymbol{\eta}_{i,[a:b]},\{k,m\}}[t]\triangleq g_{k}[t]\boldsymbol{\nu}_{i,[a:b]}[t]+h_{km}[t]\boldsymbol{\beta}_{i,[a:b]}^{(m)}[t]$ denotes the effective channel coefficients of $(b-a+1)$ ZF symbols $\eta_{i,[a:b]}$ at UE$_k$. If $a=b$, the vector becomes a scalar and we simply write it as $e^{\text{ZF}}_{\eta_{i,a},\{k,m\}}$. Further, we use modulo-$K$ indexing with respect to symbols $\eta_{k+1,\ell},\forall\ell\in[9]$ and their respective beamformers. These notations help us to state the ZF conditions (cf. ZF map in Fig. \ref{fig:ZF_Mtwo_Ktwo}) which are at UE$_k$:} 
%\begin{subequations}\label{eq:ZF_conditions_Mtwo_Ktwo}
%	\begin{align}
%	\nu_{\eta_{(k+1)\Mod{K},j}}[t]g_k[t]+\beta_{\eta_{(k+1)\Mod{K},j}}[t]h_{k1}[t]&=0,\label{eq:subseq_ZF_condition_1_M2K2}\\
%	\nu_{\eta_{(k+1)\Mod{K},4+j}}[t]g_k[t]+\beta_{\eta_{(k+1)\Mod{K},4+j}}[t]h_{k2}[t]&=0,\label{eq:subseq_ZF_condition_2_M2K2}\\
%	\nu_{\eta_{3,2(k+1)+j}}[t]g_k[t]+\beta_{\eta_{3,2(k+1)+j}}[t]h_{k2}[t]&=0,\\	\nu_{\eta_{4,2(k-1)+j}}[t]g_k[t]+\beta_{\eta_{4,2(k-1)+j}}[t]h_{k1}[t]&=0.
%	\end{align}
%\end{subequations} 
\alert{\begin{subequations}\label{eq:ZF_conditions_Mtwo_Ktwo}
	\begin{align}
	\boldsymbol{e}^{\text{ZF}}_{\boldsymbol{\eta}_{k+1,[1:2]},\{k,1\}}[t]\triangleq g_{k}[t]\boldsymbol{\nu}_{k+1,[1:2]}[t]+h_{k1}[t]\boldsymbol{\beta}_{k+1,[1:2]}^{(1)}[t]&=\mathbf{0},\label{eq:subseq_ZF_condition_1_M2K2}\\
	\boldsymbol{e}^{\text{ZF}}_{\boldsymbol{\eta}_{k+1,[5:6]},\{k,2\}}[t]\triangleq g_{k}[t]\boldsymbol{\nu}_{k+1,[5:6]}[t]+h_{k2}[t]\boldsymbol{\beta}_{k+1,[5:6]}^{(2)}[t]&=\mathbf{0},\label{eq:subseq_ZF_condition_2_M2K2}\\
	\boldsymbol{e}^{\text{ZF}}_{\boldsymbol{\eta}_{3,[2k+1:2k+2]},\{k,2\}}[t]\triangleq g_{k}[t]\boldsymbol{\nu}_{3,[2k+1:2k+2]}[t]+h_{k2}[t]\boldsymbol{\beta}_{3,[2k+1:2k+2]}^{(2)}[t]&=\mathbf{0},\\
	\boldsymbol{e}^{\text{ZF}}_{\boldsymbol{\eta}_{4,[2k-1:2k]},\{k,1\}}[t]\triangleq g_{k}[t]\boldsymbol{\nu}_{4,[2k-1:2k]}[t]+h_{k1}[t]\boldsymbol{\beta}_{4,[2k-1:2k]}^{(1)}[t]&=\mathbf{0}.
	\end{align}
\end{subequations}} In addition to the ZF conditions, we also impose additional IA conditions to the precoding design (as shown in Fig. \ref{fig:IA_Graph_M2_K2}), which are for UE$_k$:
%\begin{subequations}\label{eq:IA_conditions_Mtwo_Ktwo}
% 	\begin{align}
%	\nu_{\eta_{5-k,9}}[t]g_{k}[t]&=\beta_{\eta_{(k+1)\Mod{K},3}}[t]h_{k1}[t]=\beta_{\eta_{(k+1)\Mod{K},7}}[t]h_{k2}[t],\\
%	\nu_{\eta_{k+2,9}}[t]g_{k}[t]&=\nu_{\eta_{4,6-2k}}[t]g_{k}[t]+\beta_{\eta_{4,6-2k}}[t]h_{k1}[t]=\beta_{\eta_{(k+1)\Mod{K},8}}[t]h_{k2}[t]\nonumber\\&=\nu_{\eta_{3,10-2k}}[t]g_{k}[t]+\beta_{\eta_{3,10-2k}}[t]h_{k2}[t],\\\nu_{\eta_{(k+1)\Mod{K},9}}[t]g_{k}[t]&=\nu_{\eta_{4,5-2k}}[t]g_{k}[t]+\beta_{\eta_{4,5-2k}}[t]h_{k1}[t]=\beta_{\eta_{(k+1)\Mod{K},4}}[t]h_{k1}[t]\nonumber\\&=\nu_{\eta_{3,9-2k}}[t]g_{k}[t]+\beta_{\eta_{3,9-2k}}[t]h_{k2}[t].
% 	\end{align}
% \end{subequations} 
\alert{\begin{subequations}\label{eq:IA_conditions_Mtwo_Ktwo}
 	\begin{align}
 	g_{k}[t]\nu_{\eta_{5-k,9}}[t]&=h_{k1}[t]\beta^{(1)}_{\eta_{k+1,3}}[t]=h_{k2}[t]\beta^{(2)}_{\eta_{k+1,7}}[t],\\
 	g_{k}[t]\nu_{\eta_{k+2,9}}[t]&=e^{\text{ZF}}_{\eta_{4,6-2k},\{k,1\}}[t]=h_{k2}[t]\beta^{(2)}_{\eta_{k+1,8}}[t]=e^{\text{ZF}}_{\eta_{3,10-2k},\{k,2\}}[t],\\g_{k}[t]\nu_{\eta_{k+1,9}}[t]&=e^{\text{ZF}}_{\eta_{4,5-2k},\{k,1\}}[t]=h_{k1}[t]\beta^{(1)}_{\eta_{k+1,4}}[t]=e^{\text{ZF}}_{\eta_{3,9-2k},\{k,2\}}[t].
 	\end{align}
\end{subequations}} As opposed to all other symbols in \eqref{eq:IA_conditions_Mtwo_Ktwo}, \alert{symbols $\eta_{k+1,j}$ and $\eta_{k+1,4+j}$, $j=1,2$} (in total 8), are only subjected to ZF conditions \eqref{eq:subseq_ZF_condition_1_M2K2} and \eqref{eq:subseq_ZF_condition_2_M2K2} but not to any alignment conditions. This is due to the fact these are ZF symbols that are canceled at one of the UEs while being desired at the other UE. In consequence, we do not need to align these symbols at any UE. However, as \alert{$\eta_{k+1,1}$ and $\eta_{k+1,2}$} have \emph{identical} ZF conditions (similarly for the pair of symbols \alert{$\eta_{k+1,5}, \eta_{k+1,6}$}), we need to distinguish these symbols at UE$_k$ for the sake of reliable decodability. \alert{To achieve this}, we introduce additional complex random factors \alert{$c_{\eta_{k+1,j}}[t]$ and $c_{\eta_{k+1,4+j}}[t]$}, $j=1,2$, when fixing the precoding scalars \alert{$\nu_{\eta_{k+1,j}}[t],\beta^{(1)}_{\eta_{k+1,j}}[t]$
 and $\nu_{\eta_{k+1,4+j}}[t],\beta^{(2)}_{\eta_{k+1,4+j}}[t]$}. These random factors are \emph{a-priori} known by all transmitting and receiving nodes.  
 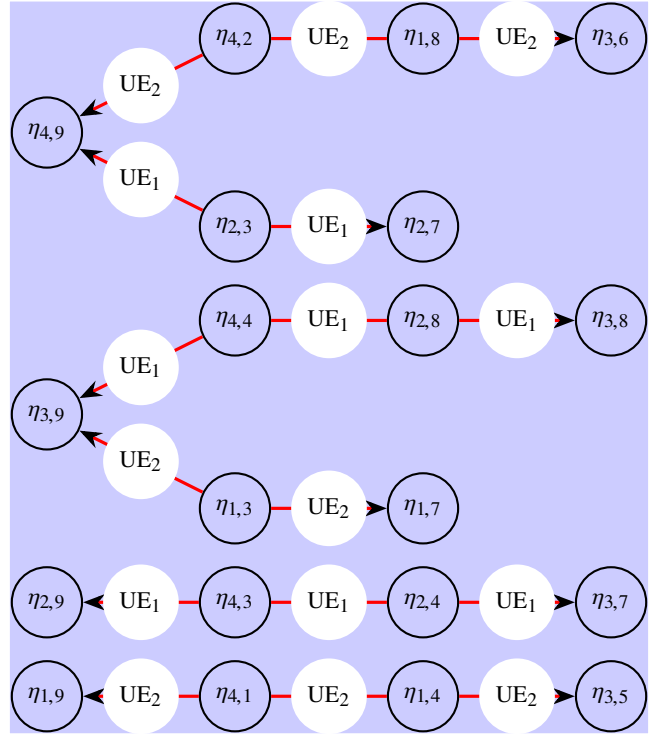
\begin{figure*}[h]
 	\centering
 	\begin{tikzpicture}[scale=1]
 	\AlignGraphMTwoKTwo
 	\end{tikzpicture}
 	\caption[Interference alignment graph for the achievability at corner point $(\frac{4}{9},\frac{4}{3})$ for $M=2$ and $K=2$]{\small Interference alignment graph from the UE perspective for the achievability at corner point $(\frac{4}{9},\frac{4}{3})$ for $M=2$ and $K=2$. The graph consists of six alignment chains and 20 nodes (i.e., symbols). \alert{The remaining 8 symbols ($\eta_{k+1,j}$, $\eta_{k+1,4+j}$, $j,k=1,2$) do not appear in the alignment graph as these symbols are zero-forced at the unwanted UE (according to the ZF conditions \eqref{eq:subseq_ZF_condition_1_M2K2} and \eqref{eq:subseq_ZF_condition_2_M2K2}) while they are desired by the other UE.}} 
 	\label{fig:IA_Graph_M2_K2}
 \end{figure*} 
 
 Now we provide the solution of the system of equations comprising of ZF and IA conditions. For this purpose, we determine the precoding scalars in all 6 alignment chains as a function of $\nu_{\eta_{1,9}}[t],\nu_{\eta_{2,9}}[t],\nu_{\eta_{3,9}}[t]$ or $\nu_{\eta_{4,9}}[t]$. Hereby, we fix these scalars to
% \begin{subequations}
% 	\begin{align}
% 	\nu_{\eta_{1,9}}[t]&=l_{13}[t]l_{23}[t]h_{21}[t],\\
% 	\nu_{\eta_{2,9}}[t]&=l_{13}[t]l_{23}[t]h_{11}[t],\\
% 	\nu_{\eta_{3,9}}[t]&=l_{13}[t]l_{23}[t]h_{12}[t]h_{21}[t]h_{22}[t],\\
% 	\nu_{\eta_{4,9}}[t]&=l_{13}[t]l_{23}[t]h_{11}[t]h_{12}[t]h_{21}[t].
% 	\end{align}
%\end{subequations}
%\begin{subequations}
% 	\begin{align}\label{eq:prec_nu_eta_i9}
% 	\nu_{\eta_{(k+1)\Mod{K},9}}[t]&=l_{13}[t]l_{23}[t]h_{k1}[t],\\
% 	\nu_{\eta_{5-k,9}}[t]=\nu_{\eta_{(k+1)\Mod{K}+2,9}}[t]&=l_{13}[t]l_{23}[t]h_{k1}[t]h_{k2}[t]h_{(k+1)\Mod{K},2}[t],
% 	\end{align}
%\end{subequations} 
\alert{\begin{subequations}
	\begin{align}\label{eq:prec_nu_eta_i9}
	\nu_{\eta_{k+1,9}}[t]&=l_{13}[t]l_{23}[t]h_{k1}[t],\\
	\nu_{\eta_{5-k,9}}[t]&=l_{13}[t]l_{23}[t]h_{k1}[t]h_{k2}[t]h_{k+1,2}[t],
	\end{align}
\end{subequations}} where
 \begin{subequations}\label{eq:l13_l23_def}
 	\begin{align}
 	l_{13}[t]&=g_{1}[t]h_{21}[t]-g_{2}[t]h_{11}[t],\\
 	l_{23}[t]&=g_{2}[t]h_{12}[t]-g_{1}[t]h_{22}[t].
 	\end{align}
 \end{subequations} As a consequence, the relevant precoders for $j=1,2$ become
\alert{\begin{subequations}\label{eq:all_precoders_function_M2_K2}
 	\begin{align}
 	\mathbf{p}_{\eta_{k+1,j}}[t]&\triangleq\begin{pmatrix}
 	\nu_{\eta_{k+1,j}}[t] \\ \beta^{(1)}_{\eta_{k+1,j}}[t]
 	\end{pmatrix}=c_{\eta_{k+1,j}}[t]\begin{pmatrix}
 	h_{k1}[t] \\ -g_{k}[t]
 	\end{pmatrix},  \label{eq:concat_bf_1}  \\
 	\beta^{(1)}_{\eta_{k+1,3}}[t]&=\nu_{\eta_{5-k,9}}[t]\frac{g_k[t]}{h_{k1}[t]}, \\ 	
 	\beta^{(1)}_{\eta_{k+1,4}}[t]&=\nu_{\eta_{k+1,9}}[t]\frac{g_k[t]}{h_{k1}[t]},\\
 	\mathbf{p}_{\eta_{k+1,4+j}}[t]&\triangleq\begin{pmatrix}
 	\nu_{\eta_{k+1,4+j}}[t] \\ \beta^{(2)}_{\eta_{k+1,4+j}}[t]
 	\end{pmatrix}=c_{\eta_{k+1,4+j}}[t]\begin{pmatrix}
 	-h_{k2}[t] \\ g_{k}[t]
 	\end{pmatrix}, \label{eq:concat_bf_2}\\
 	\beta^{(2)}_{\eta_{k+1,7}}[t]&=\nu_{\eta_{5-k,9}}[t]\frac{g_k[t]}{h_{k2}[t]}, \\
 	\beta^{(2)}_{\eta_{k+1,8}}[t]&=\nu_{\eta_{k+2,9}}[t]\frac{g_k[t]}{h_{k2}[t]}, \\
 	\mathbf{p}_{\eta_{3,9-2k}}[t]&\triangleq\begin{pmatrix}
 	\nu_{\eta_{3,9-2k}}[t] \\ \beta^{(2)}_{\eta_{3,9-2k}}[t]
 	\end{pmatrix}=(-1)^{k+1}\nu_{\eta_{k+1,9}}[t]\frac{g_k[t]}{l_{23}[t]}\begin{pmatrix}
 	-h_{k+1,2}[t] \\ g_{k+1}[t]
 	\end{pmatrix}, \label{eq:concat_bf_3}\\
 	\mathbf{p}_{\eta_{3,10-2k}}[t]&\triangleq\begin{pmatrix}
 	\nu_{\eta_{3,10-2k}}[t] \\ \beta^{(2)}_{\eta_{3,10-2k}}[t]
 	\end{pmatrix}=(-1)^{k+1}\nu_{\eta_{k+2,9}}[t]\frac{g_k[t]}{l_{23}[t]}\begin{pmatrix}
 	-h_{k+1,2}[t] \\ g_{k+1}[t]
 	\end{pmatrix},	\label{eq:concat_bf_4}\\
 	\mathbf{p}_{\eta_{4,5-2k}}[t]&\triangleq\begin{pmatrix}
 	\nu_{\eta_{4,5-2k}}[t] \\ \beta^{(1)}_{\eta_{4,5-2k}}[t]
 	\end{pmatrix}=(-1)^{k+1}\nu_{\eta_{k+1,9}}[t]\frac{g_k[t]}{l_{13}[t]}\begin{pmatrix}
 	h_{k+1,1}[t] \\ -g_{k+1}[t]
 	\end{pmatrix}, \label{eq:concat_bf_5}\\
 	\mathbf{p}_{\eta_{4,6-2k}}[t]&\triangleq\begin{pmatrix}
 	\nu_{\eta_{4,6-2k}}[t] \\ \beta^{(1)}_{\eta_{4,6-2k}}[t]
 	\end{pmatrix}=(-1)^{k+1}\nu_{\eta_{k+2,9}}[t]\frac{g_k[t]}{l_{13}[t]}\begin{pmatrix}
 	h_{k+1,1}[t] \\ -g_{k+1}[t]
 	\end{pmatrix} \label{eq:concat_bf_6}.
 	\end{align}
\end{subequations} 
Hereby, the concatenated beamforming vectors in \eqref{eq:concat_bf_1},\eqref{eq:concat_bf_2} and \eqref{eq:concat_bf_3}--\eqref{eq:concat_bf_6} are perpendicular to \begin{subequations}\label{eq:perp_vectors}
  	\begin{align}
	\mathbf{p}^{\perp}_{\eta_{k+1,j}}[t]=\mathbf{\tilde{h}}_{k1}[t]&\triangleq\begin{pmatrix}
  	g_{k}[t] \\ h_{k1}[t]
  	\end{pmatrix},\\	\mathbf{p}^{\perp}_{\eta_{k+1,4+j}}[t]=\mathbf{\tilde{h}}_{k2}[t]&\triangleq\begin{pmatrix}
  	g_{k}[t] \\ h_{k2}[t]
  	\end{pmatrix},\\
  	\mathbf{p}^{\perp}_{\eta_{3,10-2k}}[t]=\mathbf{p}^{\perp}_{\eta_{3,9-2k}}[t]=\mathbf{\tilde{h}}_{k+1,2}[t]&\triangleq\begin{pmatrix}
  	g_{k+1}[t] \\ h_{k+1,2}[t]
  	\end{pmatrix}, \\
  	\mathbf{p}^{\perp}_{\eta_{4,6-2k}}[t]=\mathbf{p}^{\perp}_{\eta_{4,5-2k}}[t]=\mathbf{\tilde{h}}_{k+1,1}[t]&\triangleq\begin{pmatrix}
  	g_{k+1}[t] \\ h_{k+1,1}[t]
  	\end{pmatrix}.
  	\end{align}
  \end{subequations}} This observation concurs with the four ZF conditions of Eq. \eqref{eq:ZF_conditions_Mtwo_Ktwo}. The decoding at RNs and UEs is described next. 
   
\subsubsection*{Decoding at the RNs and the UEs}\alert{In the following, we will formulate the received signals of the RNs and the UEs. Without loss of generality, we will focus on the received signals of UE$_1$ and RN$_1$. 
%In the equations to come, we will use the index $k'$ which satisfies
%\begin{equation}
%(k'+1)\Mod{K}=k
%\end{equation} for $k\in[K]$ to write $y_{u,k}[t]$ in a compact form. Utilizing this definition and the orthogonality observations, $y_{u,k}[t]$ becomes: 
\begin{align}\label{eq:rx_sig_UEk_M2_K2}
y_{u,1}[t]&=D_1(\boldsymbol{\eta}_{1,[1:9]})\nonumber\\&\quad+I_1(\eta_{2,9}+\eta_{4,3}+\eta_{2,4}+\eta_{3,7},\eta_{3,9}+\eta_{4,4}+\eta_{2,8}+\eta_{3,8},\eta_{4,9}+\eta_{2,3}+\eta_{2,7})+z_{u,1}[t]
\end{align} In this equation, $D_1$ denotes a linear combination of UE$_1$'s $9$ desired symbols $\boldsymbol{\eta}_{1,[1:9]}$. $I_1$, on the other hand, is a linear combination of $3$ aligned interference symbols\footnote{These symbols are specified in the argument of $I_1$.}. In explicit form, these two linear combinations are given by
\begin{align}
D_1(\boldsymbol{\eta}_{1,[1:9]})&=\Big(g_{1}[t]\boldsymbol{\nu}_{1,[1:2]}^{\dagger}[t]+h_{11}[t]\boldsymbol{\beta}_{1,[1:2]}^{(1)\dagger}[t]\Big)\mathbf{I}_{2}\boldsymbol{\eta}_{1,[1:2]}+h_{11}[t]\boldsymbol{\beta}_{1,[3:4]}^{(1)\dagger}[t]\mathbf{I}_{2}\boldsymbol{\eta}_{1,[3:4]}\nonumber\\&\quad+\Big(g_{1}[t]\boldsymbol{\nu}_{1,[5:6]}^{\dagger}[t]+h_{12}[t]\boldsymbol{\beta}_{1,[5:6]}^{(2)\dagger}[t]\Big)\mathbf{I}_{2}\boldsymbol{\eta}_{1,[5:6]}+h_{12}[t]\boldsymbol{\beta}_{1,[7:8]}^{(2)\dagger}[t]\mathbf{I}_{2}\boldsymbol{\eta}_{1,[7:8]}\nonumber\\&\quad+g_{1}[t]\nu_{\eta_{1,9}}[t]\eta_{1,9}
\end{align} and
\begin{align}
I_1&=g_{1}[t]\nu_{\eta_{2,9}}[t](\eta_{2,9}+\eta_{4,3}+\eta_{2,4}+\eta_{3,7})+g_{1}[t]\nu_{\eta_{3,9}}[t](\eta_{3,9}+\eta_{4,4}+\eta_{2,8}+\eta_{3,8})\nonumber\\&\quad+g_{1}[t]\nu_{\eta_{4,9}}[t](\eta_{4,9}+\eta_{2,3}+\eta_{2,7}).
\end{align}}
On the other hand, the received signal of RN$_1$ %and RN$_2$ 
after canceling known interfering components (by exploiting the cached content) can be written as follows:
\alert{\begin{align}\label{eq:rx_sig_RN1_M2_K2}
y_{r,1}'[t]&=f_{1}[t]\bigg(\boldsymbol{\nu}_{3,[5:8]}^{\dagger}[t]\mathbf{I}_4\boldsymbol{\eta}_{3,[5:8]}+\nu_{\eta_{3,9}}[t]\eta_{3,9}+\sum_{i=1}^{2}\Big(\boldsymbol{\nu}_{i,[5:6]}^{\dagger}[t]\mathbf{I}_2\boldsymbol{\eta}_{i,[5:6]}+\nu_{\eta_{i,9}}[t]\eta_{i,9}\Big)\nonumber\\&
\quad+\nu_{\eta_{4,9}}[t]\eta_{4,9}\bigg)+z_{r,1}[t]
\end{align}}
When choosing the precoding scalars $\nu_{\eta_{i,9}}[t],i=1,\ldots,4$ independently (e.g., as in Eq. \eqref{eq:prec_nu_eta_i9}) in Eqs. \eqref{eq:rx_sig_UEk_M2_K2} and \eqref{eq:rx_sig_RN1_M2_K2} over $T=12$ channel uses for a time-varying wireless channel, both \alert{RN$_1$ and UE$_1$} make in total $12$ noise-corrupted independent observations \alert{$\{y_{u,1}[t]\}_{t=1}^{12}$ and $\{y_{r,1}'[t]\}_{t=1}^{12}$} as a function of their desired symbols and (possibly aligned) interfering symbols. This facilitates that UE$_1$, %$k\in[2]$, 
on the one hand, is capable of decoding its $9$ desired symbols 
%\begin{itemize}
%	\item $\eta_{(k'+1)\Mod{K},j},j=1,2,\ldots,9$
%\end{itemize} and $3$ distinct sums of interfering symbols in the form 
%\begin{itemize}
%	\item $\eta_{(k+1)\Mod{K},9}+\eta_{4,5-2k}+\eta_{(k+1)\Mod{K},4}+\eta_{3,9-2k}$,
%	\item $\eta_{(k+1)\Mod{K}+2,9}+\eta_{(k+1)\Mod{K},3}+\eta_{(k+1)\Mod{K},7}$,
%	\item and  $\eta_{k+2,9}+\eta_{4,6-2k}+\eta_{(k+1)\Mod{K},8}+\eta_{3,10-2k}$. 
%\end{itemize} 
\alert{\begin{itemize}
	\item $\boldsymbol{\eta}_{1}=(\eta_{1,1},\eta_{1,2},\ldots,\eta_{1,9})^{\dagger}$
\end{itemize} and $3$ distinct sums of interfering symbols in the form 
\begin{itemize}
	\item $\eta_{2,9}+\eta_{4,3}+\eta_{2,4}+\eta_{3,7}$,
	\item $\eta_{3,9}+\eta_{4,4}+\eta_{2,8}+\eta_{3,8}$,
	\item and  $\eta_{4,9}+\eta_{2,3}+\eta_{2,7}$. 
\end{itemize}} 
\alert{RN$_1$}, on the other hand, decodes its $5$ remaining uncached desired symbols \alert{($\eta_{3,5},\eta_{3,6},\ldots,\eta_{3,9}$)} and $7$ interfering symbols \alert{($\eta_{1,5}$, $\eta_{1,6}$, $\eta_{2,5}$, $\eta_{1,6}$, $\eta_{1,9}$, $\eta_{2,9}$ and $\eta_{4,9}$)}. \alert{Due to symmetry, the observation of UE$_2$ and RN$_2$ are similar to UE$_1$ and RN$_1$, respectively.} 

As opposed to the previous scheme, both RNs and UEs require $T=12$ channel uses to determine their desired symbols. Since every file consists of $9$ symbols, the achievable NDT thus corresponds to $\frac{12}{9}=\frac{4}{3}$.

Now, we establish the achievability for corner points $\Big(\frac{1}{2},1\Big)$ (when $M=2,K=1$ as shown in Fig. \ref{fig:one_shot_M2_K1}) and $\Big(\frac{1}{2},\frac{5}{4}\Big)$ (when $M=2,K=2$ as shown in Fig. \ref{fig:one_shot_M2_K2}), respectively. These schemes are one-shot schemes that combine the Maddah-Ali Niesen (MAN) scheme with cooperative zero-forcing beamforming among DeNB and RNs. The generalized scheme is outlined in detail in Section \ref{cha_one_shot}. Thus, in the following, we will only briefly outline the achievability of these two corner points. We start with the cache placement at the RNs.

\subsubsection*{RNs cache placement}

At corner points $\Big(\frac{1}{2},1\Big)$ and $\Big(\frac{1}{2},\frac{5}{4}\Big)$, we assume that each file $W_i$ comprises of $L_1'=2$ and $L_2'=4$ symbols, respectively. Hereby, RN$_m$, $m=1,2,$ caches the $m$-th block of $\frac{L_1'}{2}$ $\Big(\frac{L_2'}{2}\Big)$ symbols. For instance, at corner point $\Big(\frac{1}{2},\frac{5}{4}\Big)$, RN$_2$ prefetches the last $\frac{L_2'}{2}=2$ symbols ($\eta_{i,3}$ and $\eta_{i,4}$) of each file $W_i$ (cf. $S_2$ in Fig. \ref{fig:one_shot_M2_K2}).      

\begin{figure}[h]
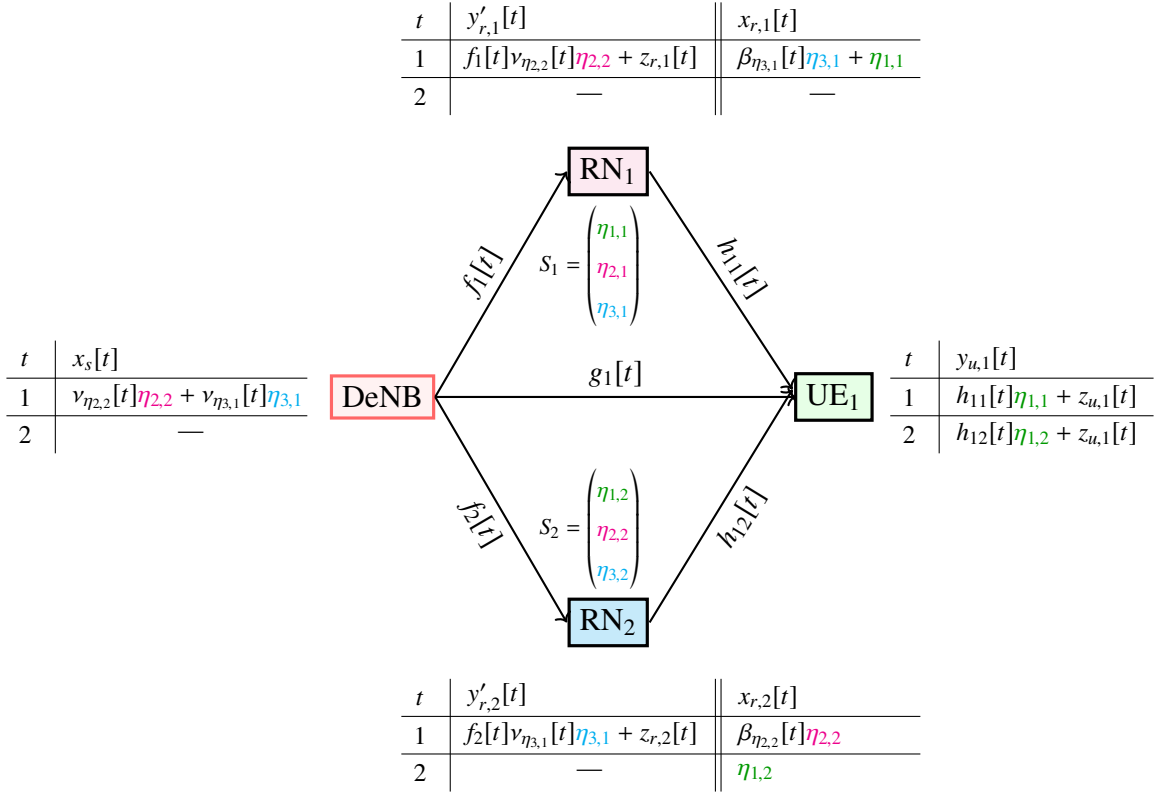

	\begin{center}
		\AchMTwoKone
		\caption{\small One-shot scheme at corner point $\Big(\frac{1}{2},1\Big)$ for $M=2$ RNs, $K=1$ UEs and $N=3$ files. Each file $W_i,i=1,2,3,$ is comprised of two symbols $\eta_{i,1}$ and $\eta_{i,2}$. The figure shows, respectively, the DeNB transmit signal $x_s[t]$, the $m$-th RNs transmit signal $x_{r,m}[t]$, the received signals of UE$_1$ $y_{u,1}[t]$ and RN$_m$ $y_{r,m}'[t]$ (when RN$_m$ exploits its cached content $S_m$ as side information) that leads to an NDT of $1$ after $T=2$ channel uses.}	
		\label{fig:one_shot_M2_K1}
	\end{center}
\end{figure} 

\begin{figure}[h]
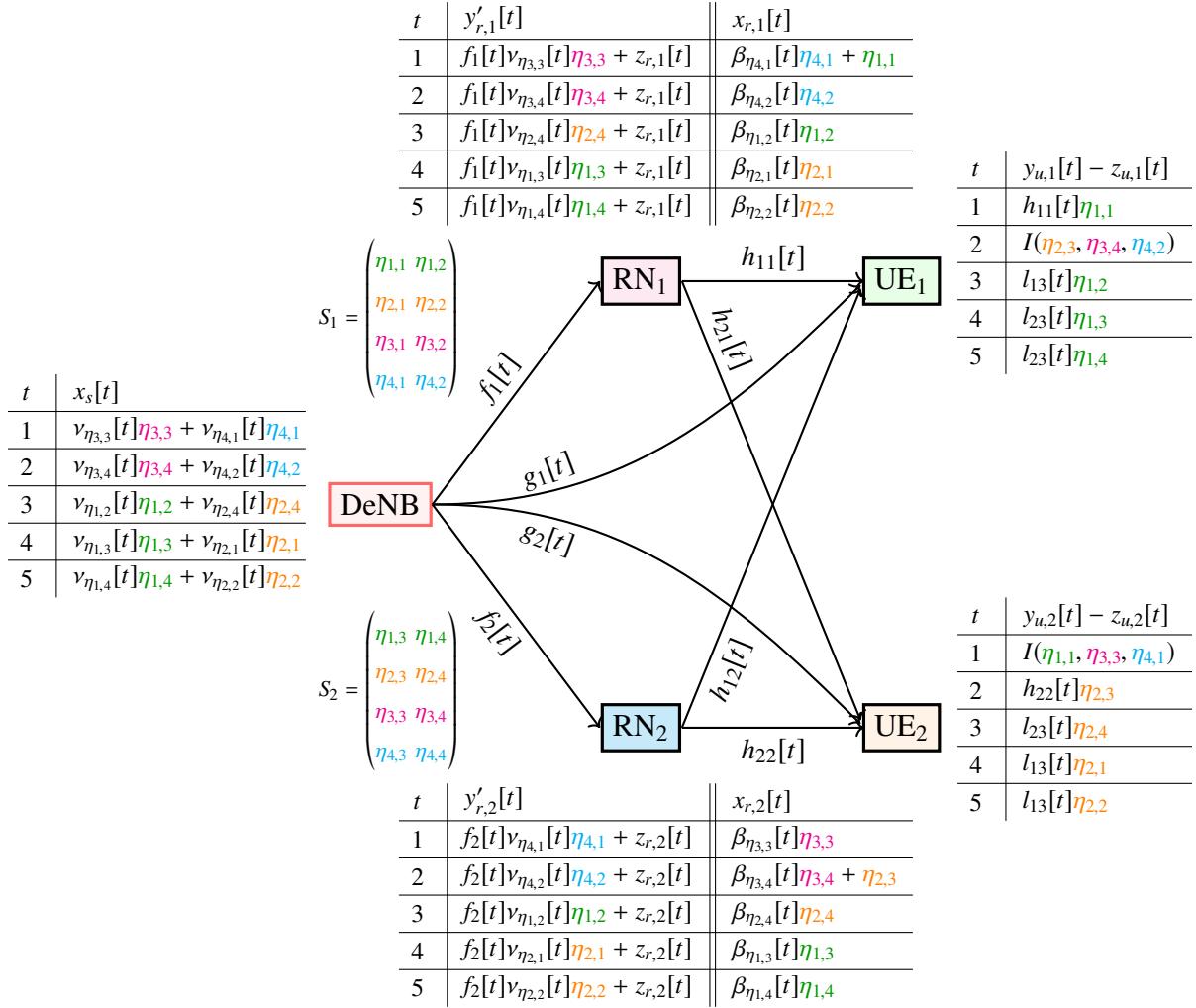

	\begin{center}
		\AchMTwoKTwo
		\caption{\small One-shot scheme at corner point $\Big(\frac{1}{2},\frac{5}{4}\Big)$ for $M=2$ RNs, $K=2$ UEs and $N=4$ files. Each file $W_i,i=1,2,3,$ is comprised of four symbols $\eta_{i,1},\eta_{i,2},\eta_{i,3}$ and $\eta_{i,4}$. The figure shows, respectively, the DeNB transmit signal $x_s[t]$, the $m$-th RNs transmit signal $x_{r,m}[t]$, the received signals of UE$_k$ $y_{u,k}[t]$ and RN$_m$ $y_{r,m}'[t]$ (when RN$_m$ exploits its cached content $S_m$ as side information) that leads to an NDT of $\frac{5}{4}$ after $T=5$ channel uses. We use $I(a_1,a_2,a_3)$ to denote a linear interference term as a function of the three symbols $a_1$, $a_2$ and $a_3$. Further, $l_{13}[t]$ and $l_{23}[t]$ are defined in Eq. \eqref{eq:l13_l23_def}.}	
		\label{fig:one_shot_M2_K2}
	\end{center}
\end{figure} 

\subsubsection*{Encoding at DeNB and RNs}

For both corner points, we use the first $T_1$\footnote{In the schemes of corner points $\Big(\frac{1}{2},1\Big)$ and $\Big(\frac{1}{2},\frac{5}{4}\Big)$, $T_1$ corresponds to $1$ and $2$ channel uses, respectively.} channel uses to multicast all symbols desired by RN$_1$ and RN$_2$ over the DeNB-RN broadcast channel. Simultaneously, we exploit the side information at the RNs caches to zero-force undesired symbols at one of the UEs and provide this particular UE with one of its desired symbols. For example, at $t=1$ for corner point $\Big(\frac{1}{2},\frac{5}{4}\Big)$, we multicast the symbols $\eta_{3,3}$ and $\eta_{4,1}$ at the DeNB in an MAN scheme manner according to:
\begin{equation}
x_{s}[1]=\nu_{\eta_{3,3}}[1]\eta_{3,3}+\nu_{\eta_{4,1}}[1]\eta_{4,1}
\end{equation} This transmit signal allows each RN to retrieve its desired symbol by exploiting its knowledge of the channel and of the undesired symbol in $x_{s}[1]$ through its cache (e.g., RN$_1$ is interested in symbol $\eta_{3,3}$ and has $\eta_{4,1}$ stored in its cache). However, neither UE$_1$ nor UE$_2$ are interested in the symbols of $x_s[1]$ ($\eta_{3,3},\eta_{4,1}$). Thus, we do not only leverage the knowledge of RN$_1$ in $\eta_{4,1}$ and of RN$_2$ in $\eta_{3,3}$ from a receiver-caching perspective, but instead we also capitalize on this cognizance from a transmitter-caching perspective by sending
\begin{subequations}
	\begin{align}
	x_{r,1}[1]&=\beta_{\eta_{4,1}}[1]\eta_{4,1}+\eta_{1,1}\\
	x_{r,2}[1]&=\beta_{\eta_{3,3}}[1]\eta_{3,3}
	\end{align}
\end{subequations} from RN$_1$ and RN$_2$, respectively, to overcome this issue. Hereby, we choose the following concatenated vectors to satisfy:
\begin{subequations}
	\begin{align}
	\begin{pmatrix}
	\nu_{\eta_{3,3}}[1] \\ \beta_{\eta_{3,3}}[1]
	\end{pmatrix}&\perp \begin{pmatrix} g_1[1] \\ h_{12}[1]
	\end{pmatrix},\\ \begin{pmatrix}
	\nu_{\eta_{4,1}}[1] \\ \beta_{\eta_{4,1}}[1]
	\end{pmatrix}&\perp \begin{pmatrix} g_1[1] \\ h_{11}[1]
	\end{pmatrix}.
	\end{align}
\end{subequations} This facilitates that UE$_1$ is free from interference due to undesired symbols $\eta_{3,3}$ and $\eta_{4,1}$. Instead, it receives a noise-corrupted signal that depends on its desired symbol $\eta_{1,1}$. UE$_2$, on the other hand, observes a signal that depends solely on undesired symbols $\eta_{1,1},\eta_{3,3}$ and $\eta_{4,1}$ denoted by $I(\eta_{1,1},\eta_{3,3},\eta_{4,1})$ in Fig. \ref{fig:one_shot_M2_K2}. In the remaining $T_2$\footnote{In the schemes of corner points $\Big(\frac{1}{2},1\Big)$ and $\Big(\frac{1}{2},\frac{5}{4}\Big)$, $T_2$ corresponds to $1$ and $3$ channel uses, respectively.} channel uses, cooperative DeNB-RN zero-forcing is applied for the case when $M=K=2$ and simple RN unicasting when $M=2,K=1$.   

\subsubsection*{Decoding at the RNs and the UEs} As opposed to all prior schemes, RNs and UEs can decode each desired symbol on a single-channel use basis, i.e., symbol decoding over multiple channel uses is not required in attaining the optimal NDT. In consequence, channel diversity over multiple channel uses is not required.  

Specifically, in the proposed scheme for corner point $\Big(\frac{1}{2},1\Big)$, UE$_1$ decodes its $L_1'=2$ desired symbols $\eta_{1,1}$ and $\eta_{1,2}$ in $T_1+T_2=2$ channel uses, namely, $1$ and $2$ (cf. Fig. \ref{fig:one_shot_M2_K1}), respectively. The RNs, on the contrary, only need the first channel use ($T_1=1$) to retrieve $(1-\mu)L_1'=1$ desired symbols. In conclusion, this conforms to an NDT of $\frac{T_1+T_2}{L_1'}=1$. 

In the other scheme, as shown in Fig. \ref{fig:one_shot_M2_K2}, UE$_1$ (UE$_2$) decodes its $L_2'=4$ desired symbols $\eta_{1,1},\eta_{1,2},\eta_{1,3}$ and $\eta_{1,4}$ ($\eta_{2,1},\eta_{2,2},\eta_{2,3}$ and $\eta_{2,4}$) in 4 out of $T_1+T_2=5$ channel uses, namely, $1,2,3$ and $4$ ($4,5,2$ and $3$), respectively. Unlike the UEs, each RN only needs the first $T_1=2$ channel uses to retrieve its $(1-\mu)L_2'=2$ desired symbols. Ultimately, the NDT becomes $\frac{T_1+T_2}{L_2'}=\frac{5}{4}$. 

\subsection{Achievability for $M=3$}\label{subsec:ach_M3}

The optimal delivery-time cache-memory tradeoff for $M=3$ and $K=1$ is presented. Hereby, the following proposition quantifies the achievable NDT.    

\begin{proposition}\label{prop_M3}
	The achievable NDT of the network under study for $M=3$ RNs, $K=1$ UEs and $\mu\in[0,1]$ is given by 
	\begin{equation}
	\label{eq:ach_NDT_M_3_K_1}
	\delta(\mu)=\max\Big\{1,4-9\mu\Big\}.
	\end{equation}
\end{proposition}

Similarly, to the achievability of corner points when $M=2$, it suffices to establish a scheme that attains an NDT of $1$ at $\mu=\frac{1}{3}$. The NDT-optimal scheme that establishes this result is the one-shot scheme of Section \ref{cha_one_shot}. Since the main idea of this scheme has already been illustrated through explicit examples for corner points $\Big(\frac{1}{2},1\Big)$ when $(K,M)=(1,2)$ and $\Big(\frac{1}{2},\frac{5}{4}\Big)$ when $(K,M)=(2,2)$, we omit further details for the sake of brevity. The interested reader can reconstruct the explicit scheme for $(K,M)=(1,3)$ from the generalized scheme of Section \ref{cha_one_shot} for Region B.     
     
\section{\alert{Directions for Future Work}}
\label{sec:fut_work}

\alert{In this section, we discuss some of the open problems and directions for future work on the topic of cache-aided broadcast-relay wireless networks. In particular, we focus on aspects that are left open in this paper.} 

\subsection{\alert{Imperfect CSI}}

\alert{An interesting aspect is the influence of \emph{imprecise} CSI at transmitters (CSIT) on the minimum NDT. In detail, we may consider the system model of Fig. \ref{fig:HetNet} in which at channel use $t$, DeNB and/or RNs may have access to imprecise CSI of $\big\{\mathbf{f}[t],\mathbf{g}[t],\mathbf{H}[t]\big\}$. Specifically, we distinguish the following CSIT settings:}
\begin{enumerate}
	\item \alert{\emph{Delayed} CSIT: At channel use $t$, DeNB and/or RNs are aware of the previous $t-1$ time instants CSI $\big\{\mathbf{f}^{t-1},\mathbf{g}^{t-1},\mathbf{H}^{t-1}\big\}$.
	\item \emph{Mixed} CSIT: In addition to delayed CSIT, DeNB and/or RNs have at time instant $t$ access to \emph{current} CSI estimates $\big\{\mathbf{\hat{f}}[t],\mathbf{\hat{g}}[t],\mathbf{\hat{H}}[t]\big\}$ of some quality $\alpha\in[0,1]$ \cite[Chapter 4]{KakarMDPI}. The extreme cases of $\alpha=0$ and $\alpha=1$ represent, respectively, the cases of delayed CSIT only and perfect current quality CSIT.} 
\end{enumerate} 

\alert{While there is plenty of degrees-of-freedom studies on the impact of delayed and mixed CSIT on interference networks (cf. survey paper \cite{KakarMDPI}), the interplay of caching and imperfect CSI is with the exception of \cite{Zhang17} far less understood. The authors in \cite{Zhang17} identify DoF gains both due to current CSIT and coded caching for an MISO broadcast channel with Rx caching. When focusing on the extreme cases of zero-cache ($\mu=0$) and full cache ($\mu=1$), we can determine the achievable NDT for}
%\begin{itemize}
%	\item \emph{Delayed} CSIT: \begin{equation*}
%	\delta_1(\mu)=\begin{cases}
%	K+M\quad&\text{ for }\mu=0 \\ \sum_{i=1}^{K}\frac{1}{i}\quad&\text{ for }\mu=1,K\leq M+1
%	\end{cases},
%	\end{equation*} 
%	\item 
\alert{\emph{mixed} CSIT 
\begin{equation*}
	\delta_{\text{ach}}(\mu,\alpha)=\begin{cases}
	K+M\quad&\text{ for }\mu=0 \\\frac{\sum_{i=1}^{K}\frac{1}{i}}{(1-\alpha)+\alpha\sum_{i=1}^{K}\frac{1}{i}}\quad&\text{ for }\mu=1,K\leq M+1
	\end{cases}
\end{equation*}} \alert{based on results for the MISO BC \cite{Maddah-AliDCSIT,Kerret16}. The case of delayed CSIT only is given in above equation when $\alpha=0$. It is of interest to understand how cache placement and file delivery needs to be adjusted for the aforementioned CSIT models when $0<\mu<1$.}  

\subsection{\alert{Partial Connectivity}}

\alert{It is of interest to understand the implications of partial connectivity (with respect to the RN-UE links) on the NDT for cache-assisted broadcast-relay networks. In particular, one question of interest is whether the cache placement has to account for the topology. In addition, how does optimal file splitting may look like.} 

\alert{In recent works, caching has been, amongst others, applied to combination networks with receiver caches \cite{Zewail17}, partially-connected interference channels with transmitter caches under complete file placement \cite{Yi16_TCC} and interference networks with Tx and Rx-caching \cite{Xu17_Partial}. However, the impact of partial connectivity on cache-aided channels is unknown at large, especially, with respect to transceiver cache-aided networks.} 

\subsection{\alert{Proof-of-Concept Implementation}}

\alert{A practical proof-of-concept implementation allows to verify to what extent the theoretically postulated delivery times in this paper are achievable. Further, implementation issues such as large subpacketization levels \cite{Shanmugam17Pak}, practicality of centralized cache placement \cite{Maddah-Ali_decentral} and self-interference cancelation in full-duplex communication \cite{Vogt16} have to be handled.}    

\section{\alert{Conclusion}}
\label{sec:conclusion}

\alert{In this paper, we have studied the fundamental information-theoretic limits on the delivery time of a transceiver cache-aided broadcast relay network consisting of a central base station (DeNB), $M$ cache-endowed relay nodes (RNs) and $K$ mobile users (UEs). We used the normalized delivery time (NDT) as our performance metric. The NDT measures the worst-case delivery time per bit with respect to an interference-free system in the high SNR regime. We established a converse result for the NDT of a general broadcast relay network with arbitrary number of RNs and UEs. Next, we presented two achievability schemes which exploit the RNs caches and its full-duplex capability in collaboration with the DeNB. The first scheme is a \emph{one-shot} achievability scheme that synergistically interlaces zero-forcing (ZF) and multicasting strategies proposed in the framework of coded caching. The second scheme, on the other hand, integrates subspace interference alignment (IA) with zero-forcing through carefully designing ZF and IA maps. With these results, we were able to characterize the optimal NDT-cache-memory tradeoff for $K+M\leq 4$. In addition, we identify NDT-optimal regimes of the proposed one-shot scheme and also show that caching more than a fraction of $\frac{\ceil{\nicefrac{(M-1)}{2}}}{M}$ files attains a constant gap of $\frac{8}{3}$ with respect to the optimal NDT. We discussed the inadequacy of the inverse sum DoF metric in capturing the delivery time of cache-assisted broadcast relay networks. Finally, we presented directions for future work.}

\appendices
\section{Maximum Multiplicative Gap for One-Shot Scheme}
\label{cha_gap}

In this section, we present the maximum multiplicative gap of the achievable one-shot scheme presented in Section \ref{cha_one_shot}. We remind the reader about the $(\mu,K,M)$ regions presented in Fig. \ref{fig:Region_plot} where the achievable NDTs vary. To study the gap with respect to lower bounds on the NDT, we consider the pair of regions (B,E) with achievable one-shot NDT
\begin{equation}\label{eq:one_shot_ach_NDT_Reg_BE}
\delta_{\text{OS}}^{\text{(B,E)}}(\mu)=\delta_{\text{MAN}}(\mu)=M\cdot(1-\mu)\frac{1}{1+\mu M}
\end{equation} as well the pair of regions (C,D) with achievable one-shot NDT  
\begin{equation}\label{eq:one_shot_ach_NDT_Reg_CD}
\delta_{\text{OS}}^{\text{(C,D)}}(\mu)=\frac{K+\delta_{\text{MAN}}(\mu)}{1+\mu M},%=\frac{K+M+\mu M(K-1)}{(1+\mu M)^{2}},
\end{equation} separately. Finally, we consider the transitional region of regime pairs (D,E) for fractional cache sizes $\mu\in\Big[0,\frac{1}{M}\Big]$.   

\subsection{(B,E) Region Pair}
\label{subsec:be_reg_pair_gap}

For this region pair, the one-shot NDT corresponds to \eqref{eq:one_shot_ach_NDT_Reg_BE}. We focus on fractional cache sizes $\mu$ inside the interval $\mu\in\Big[\mu'(\theta),\frac{\ceil{\nicefrac{(M-1)}{2}}}{M}\Big]$, where $\theta\in\Big[1,\frac{M-3}{2}\Big]$ adjusts the left endpoint $\mu'(\theta)$ of the interval according to   
\begin{equation}\label{eq:mu_prime_Reg_BE}
\mu'(\theta)=\frac{\ceil{\theta}}{M}.
\end{equation} The NDT $\delta_{\text{OS}}^{\text{(B,E)}}(\mu)$ is decreasing in $\mu$, which is why we upper bound it (in the interval specified above) by
\begin{align}\label{eq:up_bound_delta_os_ach_BE}
\delta_{\text{OS}}^{\text{(B,E)}}(\mu)\leq\delta_{\text{OS}}^{\text{(B,E)}}\bigg(\frac{\theta}{M}\bigg)=\frac{M-\theta}{1+\theta}. 
\end{align} From Theorem \ref{theorem_lower_bound}, we infer that
\begin{align}\label{eq:lw_bound_delta_lb_BE}
\delta^{\star}(\mu)\geq 1.
\end{align} Consequently, the multiplicative gap becomes
\begin{align}\label{eq:final_gap_Reg_BE}
\frac{\delta_{\text{OS}}^{\text{(B,E)}}(\mu)}{\delta^{\star}(\mu)}\leq\frac{M-\theta}{1+\theta}.
\end{align} 

\subsection{(C,D) Region Pair}
\label{subsec:cd_reg_pair_gap}
In these two regions, the achievable NDT is given by  \eqref{eq:one_shot_ach_NDT_Reg_CD}. We consider three different intervals of fractional cache sizes: 
\begin{enumerate*}[label=(\roman*)]
	%\item the first being $\mu\in[0,\frac{1}{M}]$,
	\item the first being $\mu\in[0,\tilde{\mu}]$,
	\item the second being $\mu\in\Big(\tilde{\mu},\frac{1}{M}\Big]$ 
	\item and the third being $\mu\in\Big[\mu'(\kappa_d),\min\Big\{1,\frac{K}{M}\Big\}\Big)$.
\end{enumerate*} In the following, we treat these three cases individually.
%\vspace{0.5em}
\subsubsection*{Case (i)} For this case, the right endpoint of the interval equates to
\begin{align}\label{eq:tilde_mu_Reg_CD_Case1}
\tilde{\mu}=\min\Bigg\{\frac{1}{M},\frac{K+M+1}{(M+1)(K+M-1)}\Bigg\}.
\end{align} The achievable NDT in the interval $\Big[0,\frac{1}{M}\Big]$, which subsumes the interval of case (i), is  attainable through memory sharing of schemes at corner point $(0,K+M)$ and the one-shot scheme at $\mu=\frac{1}{M}$ with corresponding NDT
\begin{equation*}
\delta_{\text{OS}}^{(C,D)}\bigg(\mu=\frac{1}{M}\bigg)=\begin{cases}1&\text{ if }K=1\\\max\Big\{1,\frac{K}{M+1}\Big\}&\text{ if }M=1\\\frac{K+\delta_{\text{MAN}}\big(\mu=\frac{1}{M}\big)}{2}=\frac{K}{2}+\frac{M-1}{4}&\text{ if }K,M\geq 2\end{cases}.
\end{equation*} This results in the achievable NDT
\begin{align}\label{eq:up_bound_delta_os_ach_CD_mu_leq_1_over_M}
\delta_{\text{OS}}^{\text{(C,D)}}(\mu)=\begin{cases}K+M-\mu M\big(K+M-1\big)&\text{ if }K=1\\K+M-\mu M\Big(K+M-\max\Big\{1,\frac{K}{M+1}\Big\}\Big)&\text{ if }M=1\\K+M-\mu M\Big(\frac{K}{2}+\frac{3M+1}{4}\Big)&\text{ if }K,M\geq 2\end{cases}
\end{align} in the interval $\Big[0,\frac{1}{M}\Big]$. By substituting $\ell=M$ and $s=1$ in Eq. \eqref{eq:NDT_lw_bound} of Theorem \ref{theorem_lower_bound}, we get
\begin{align}\label{eq:lw_bound_delta_lb_CD_Case1}
\delta^{\star}(\mu)\geq \delta_{\text{LB}}(\mu,M,1)=K+M-\mu M(K+M-1).
\end{align} We now combine \eqref{eq:up_bound_delta_os_ach_CD_mu_leq_1_over_M} and \eqref{eq:lw_bound_delta_lb_CD_Case1} to upper bound the multiplicative gap according to:
\begin{align}\label{eq:temp_gap_Reg_CD_Case1_Meq1_trivial_case}
\frac{\delta_{\text{OS}}^{\text{(C,D)}}(\mu)}{\delta^{\star}(\mu)}&\leq 1
\end{align} \vspace{0.5em} if $K=1$, 
\begin{align}\label{eq:temp_gap_Reg_CD_Case1_Meq1}
\frac{\delta_{\text{OS}}^{\text{(C,D)}}(\mu)}{\delta^{\star}(\mu)}&\leq\frac{K+M-\mu M\Big(K+M-\max\Big\{1,\frac{K}{M+1}\Big\}\Big)}{K+M-\mu M(K+M-1)}\nonumber\\&\stackrel{(a)}=1+\frac{\mu M \Big(\frac{K}{M+1}-1\Big)^{+}}{K+M-\mu M(K+M-1)}\nonumber\\&\stackrel{(b)}\leq 1+\frac{\tilde{\mu}M\Big(\frac{K}{M+1}-1\Big)^{+}}{\max\Big\{1,\frac{K}{M+1}\Big\}}=1+\frac{M}{K}\cdot\frac{(K+M+1)}{(K+M-1)}\Bigg(\frac{K}{M+1}-1\Bigg)^{+}\nonumber\\&=\begin{cases}
1&\text{ if }K\leq 2\\%1+\frac{1}{K}\cdot\frac{(K+2)}{K}\Big(\frac{K}{2}-1\Big)&\text{ if }K>2
1+\frac{1}{K}\cdot\Big(\frac{K}{2}-\frac{2}{K}\Big)&\text{ if }K>2
\end{cases}\nonumber\\&\leq\begin{cases}
1&\text{ if }K\leq 2\\\frac{3}{2}&\text{ if }K>2
\end{cases}
\end{align} \vspace{0.5em} if $M=1$ and  
\vspace{0.5em}
\begin{align}\label{eq:temp_gap_Reg_CD_Case1}
\frac{\delta_{\text{OS}}^{\text{(C,D)}}(\mu)}{\delta^{\star}(\mu)}&\leq\frac{K+M-\mu M\Big(\frac{K}{2}+\frac{3M+1}{4}\Big)}{K+M-\mu M(K+M-1)}\nonumber\\&\stackrel{(a)}=1+\frac{\mu M \Big(\frac{K}{2}+\frac{M-5}{4}\Big)}{K+M-\mu M(K+M-1)}\nonumber\\&\stackrel{(b)}\leq 1+\frac{\tilde{\mu}M\Big(\frac{K}{2}+\frac{M-5}{4}\Big)}{\max\Big\{1,\frac{K}{M+1}\Big\}}=1+\bigg(\frac{K}{2}+\frac{M-5}{4}\bigg)\cdot\min\Bigg\{1,\frac{M}{K}\cdot\frac{(K+M+1)}{(K+M-1)}\Bigg\}
\end{align} if $K,M\geq 2$. Hereby, step $(a)$ follows by addition and subtraction of the denominator $\delta_{\text{LB}}(\mu,M,1)$ from the numerator and step $(b)$ from the fact that the rational function is maximized by one of the endpoints of the domain of fractional cache sizes (in this specific case, $\mu^{\star}=\tilde{\mu}$ in the domain $0\leq\mu\leq\tilde{\mu}$). 
\vspace{0.5em}
\subsubsection*{Case (ii)} Now we examine fractional cache sizes $\mu\in\Big(\tilde{\mu},\frac{1}{M}\Big]$. From the definition of $\tilde{\mu}$, we determine that $\tilde{\mu}=\frac{1}{M}$ if $K\leq M$ and $\tilde{\mu}<\frac{1}{M}$ if $K>M$. Consequently, it suffices to restrict the focus for this case to $K>M$. However, we may still use the expressions on the achievable NDT of Eq. \eqref{eq:up_bound_delta_os_ach_CD_mu_leq_1_over_M}. In addition, we deduce from Theorem \ref{theorem_lower_bound} that for $K>M$
\begin{align}\label{eq:lower_bound_Reg_CD_Case_2_for_RegC}
\delta^{\star}(\mu)\geq \frac{K}{M+1}.
\end{align} Using \eqref{eq:up_bound_delta_os_ach_CD_mu_leq_1_over_M} and \eqref{eq:lower_bound_Reg_CD_Case_2_for_RegC}, we get for $K>M$ and $M\geq 2$
\begin{align}\label{eq:temp_gap_Reg_CD_Case3}
\frac{\delta_{\text{OS}}^{\text{(C,D)}}(\mu)}{\delta^{\star}(\mu)}&\leq\frac{K+M-\mu M\Big(\frac{K}{2}+\frac{3M+1}{4}\Big)}{\frac{K}{M+1}}\nonumber\\&\stackrel{(b)}\leq\frac{K+M-\tilde{\mu} M\Big(\frac{K}{2}+\frac{3M+1}{4}\Big)}{\frac{K}{M+1}}\nonumber\\&=1+\bigg(\frac{K}{2}+\frac{M-5}{4}\bigg)\cdot\frac{M}{K}\cdot\frac{(K+M+1)}{(K+M-1)}
\end{align}
\vspace{0.5em}
\subsubsection*{Case (iii)}
Consider all fractional cache sizes $\mu\in\Big[\mu'(\kappa_d),\min\Big\{1,\frac{K}{M}\Big\}\Big)$, where we choose $\mu'(\kappa_d)$ according to
\begin{equation}\label{eq:mu_prime_Reg_CD}
\mu'(\kappa_d)=\frac{\ceil{\kappa_d}}{M}
\end{equation} and the numerator $\kappa_d$ in \eqref{eq:mu_prime_Reg_CD} parametrized by $d$ to
\begin{equation}\label{eq:alpha_d_Reg_CD}
\kappa_d=\frac{M+1-d}{d},\quad \forall d\in\Bigg[\frac{M+1}{\min\{K,M\}},\frac{M+1}{2}\Bigg].
\end{equation} Due to the decreasing monotony of $\delta_{\text{OS}}^{\text{(C,D)}}(\mu)$, we may upper bound the achievable NDT for any $\mu$ in the interval $\Big[\mu'(\kappa_d),\min\Big\{1,\frac{K}{M}\Big\}\Big)$ by
\begin{align}\label{eq:up_bound_delta_os_ach_CD}
\delta_{\text{OS}}^{\text{(C,D)}}(\mu)\leq\delta_{\text{OS}}^{\text{(C,D)}}\bigg(\frac{\kappa_d}{M}\bigg)=\frac{K(1+\kappa_{d})+M-\kappa_d}{(1+\kappa_d)^{2}}. 
\end{align} The lower bound of Theorem \ref{theorem_lower_bound} can be further bounded from below as follows:
\begin{align}\label{eq:lw_bound_delta_lb_CD}
\delta^{\star}(\mu)\geq\max\Big\{1,\max_{\substack{\ell\in[\bar{s}:M],\\s\in[\min\{M+1,K\}]}}\delta_{\text{LB}}(\mu,\ell,s)\Big\}\geq\max\Bigg\{1,\frac{K}{M+1}\Bigg\}
\end{align} With \eqref{eq:up_bound_delta_os_ach_CD} and \eqref{eq:lw_bound_delta_lb_CD}, we are able to upper bound the multiplicative gap.
\begin{align}\label{eq:temp_gap_Reg_CD}
\frac{\delta_{\text{OS}}^{\text{(C,D)}}(\mu)}{\delta^{\star}(\mu)}\leq\frac{\delta_{\text{OS}}^{\text{(C,D)}}\Big(\frac{\kappa_d}{M}\Big)}{\max\Big\{1,\frac{K}{M+1}\Big\}}=\frac{K(M+1)(1+\kappa_{d})+(M-\kappa_{d})(M+1)}{\max\{M+1,K\}(1+\kappa_{d})^{2}}
\end{align} When substituting \eqref{eq:alpha_d_Reg_CD} in \eqref{eq:temp_gap_Reg_CD}, we have
\begin{align*}
\frac{\frac{K(M+1)^{2}}{d}+\frac{(M+1)^{2}(d-1)}{d}}{\frac{\max\{M+1,K\}(M+1)^{2}}{d^{2}}}=\frac{dK}{\max\{M+1,K\}}+\frac{d(d-1)}{\max\{M+1,K\}}.
\end{align*} Thus, we conclude
\begin{align}\label{eq:final_gap_Reg_CD}
\frac{\delta_{\text{OS}}^{\text{(C,D)}}(\mu)}{\delta^{\star}(\mu)}\leq\frac{dK}{\max\{M+1,K\}}+\frac{d(d-1)}{\max\{M+1,K\}}.
\end{align}

\subsection{Transition (D,E) Region Pair for $\mu\leq\frac{1}{M}$}
\label{subsec:de_reg_pair_gap}
At this transitional region, the one-shot NDT for arbitrary fractional cache sizes $\mu<\frac{1}{M}$ is achievable through memory sharing of the broadcasting scheme at corner point $(0,K+M)$ and the OS-scheme at $\mu=\frac{1}{M}$ with NDT
\begin{equation*}
\delta_{\text{MAN}}\bigg(\mu=\frac{1}{M}\bigg)=\frac{M-1}{2}.
\end{equation*} This results in the following memory-sharing based transitional NDT for the region pair (D,E):
\begin{align}\label{eq:up_bound_delta_os_ach_DE}
\delta_{\text{OS}}^{\text{(D,E)}}(\mu)=K+M-\mu M\bigg(K+\frac{M+1}{2}\bigg)
\end{align} We substitute $\ell=M$ and $s=1$ in Eq. \eqref{eq:NDT_lw_bound} of Theorem \ref{theorem_lower_bound} to get
\begin{align}\label{eq:lw_bound_delta_lb_DE}
\delta^{\star}(\mu)\geq \delta_{\text{LB}}(\mu,M,1)=K+M-\mu M(K+M-1).
\end{align} Now, from \eqref{eq:up_bound_delta_os_ach_DE} and \eqref{eq:lw_bound_delta_lb_DE}, we have
\begin{align}\label{eq:temp_gap_Reg_DE}
\frac{\delta_{\text{OS}}^{\text{(D,E)}}(\mu)}{\delta^{\star}(\mu)}&\leq\frac{K+M-\mu M\Big(K+\frac{M+1}{2}\Big)}{K+M-\mu M(K+M-1)}\nonumber\\&\stackrel{(a)}=1+\frac{\mu M \frac{(M-3)}{2}}{K+M-\mu M(K+M-1)}\nonumber\\&\stackrel{(b)}\leq 1+\frac{M-3}{2}=\frac{M-1}{2},
\end{align} %where step $(a)$ follows by addition and subtraction of the denominator $\delta_{\text{LB}}(\mu,M,1)$ from the numerator and step $(b)$ from the fact that the rational function is maximized by $\mu^{\star}=\frac{1}{M}$ in the domain $0\leq\mu\leq\frac{1}{M}$.
where the reasoning for the steps $(a)$ and $(b)$ are in agreement with the aforementioned steps of the other region pairs.

\ifCLASSOPTIONcaptionsoff
  \newpage
\fi

\bibliographystyle{IEEEtran}
\bibliography{content/bibliography}

\iffalse
\begin{IEEEbiography}{Jaber Kakar}
Biography text here.
\end{IEEEbiography}

\begin{IEEEbiographynophoto}{Anas Chabaan}
Biography text here.
\end{IEEEbiographynophoto}

\begin{IEEEbiographynophoto}{Aydin Sezgin}
Biography text here.
\end{IEEEbiographynophoto}
\fi

\end{document}